\RequirePackage{etex} 
\documentclass[12pt]{article}
\usepackage[utf8]{inputenc}
\usepackage{amsmath}
\usepackage{amssymb}  
\usepackage{mathrsfs}
\usepackage{tabularx}
\usepackage{booktabs}
\usepackage{enumerate}
\usepackage{amsthm}
\usepackage{color}
\usepackage{xcolor}
\usepackage{graphicx}
\usepackage{chngcntr}
\usepackage{enumitem}
\usepackage[sort]{natbib}
\usepackage{url}
\usepackage{ulem}
\usepackage{graphicx}
\usepackage{geometry}
\usepackage{varioref}
\usepackage{etoolbox,fontawesome}
\usepackage[section]{algorithm}
\usepackage{hyperref}
\usepackage{setspace}

\usepackage[capitalise,nameinlink
    ,noabbrev]{cleveref}
  \hypersetup{
    colorlinks=true,
    linkcolor=blue,
    citecolor=blue
}
\AtBeginEnvironment{thebibliography}{\def\url#1{\href{#1}{\faExternalLink}}}

\counterwithin{figure}{section}
\counterwithin{table}{section}
\usepackage{autonum}

\usepackage{algpseudocode}

\newtheorem{theorem}{Theorem}[section]
\newtheorem{assumption}{Assumption}[section]

\newtheorem{proposition}{Proposition}[section]

\newtheorem{remark}{Remark}[section]
\newtheorem{lemma}{Lemma}
\newtheorem{corollary}{Corollary}[section]

\newcommand{\T}{\top}
\newcommand{\TT}{\mathcal{T}}
\newcommand{\op}{o_{\mathbb{P}}}
\newcommand{\Op}{O_{\mathbb{P}}}

\newcommand{\Ic}{\overline{\mathbb{I}}}

\newcommand{\E}{\mathbb{E}}
\newcommand{\F}{\mathcal{F}} 
\newcommand{\FF}{\mathcal{F}}
\newcommand{\R}{\mathbb{R}}
\newcommand{\B}{\mathbb{B}}
\newcommand{\PP}{\mathbb{P}}
\newcommand{\mf}{\mathbf}
\newcommand{\bs}{\boldsymbol}
\newcommand{\lf}{\lfloor}
\newcommand{\rf}{\rfloor}

\newcommand{\nb}{\lceil n b \rceil}
\newcommand{\bt}{\mathrm{boot}}
\newcommand{\nt}{\lceil n \tau_n \rceil}

\numberwithin{equation}{section}

\definecolor{darkgreen}{rgb}{0.0, 0.5, 0.0}
\definecolor{ashgrey}{rgb}{0.7, 0.75, 0.71}

\usepackage{authblk}
\makeatletter

\title{\bf Time-varying correlation network analysis of non-stationary multivariate time series with complex trends}
\date{}
\author[*]{Lujia Bai}
\author[*]{Weichi Wu}
\affil[*]{Center for Statistical Science and Department of Industrial Engineering, \protect\\
Tsinghua University, Beijing 100084, China}

\begin{document}
\maketitle
\onehalfspacing
\begin{abstract}
This paper proposes a flexible framework for inferring large-scale  time-varying and time-lagged correlation networks from multivariate or high-dimensional non-stationary time series with piecewise smooth trends. Built on a novel and unified multiple-testing procedure of time-lagged cross-correlation functions with a fixed or diverging number of lags, our method can accurately disclose flexible time-varying network structures associated with complex  functional structures  at all time points. We broaden the applicability of our method to the structure breaks by developing difference-based nonparametric estimators of cross-correlations,  achieve accurate family-wise error control via a bootstrap-assisted procedure adaptive to the complex temporal dynamics, and enhance the probability of recovering the time-varying network structures using a new uniform variance reduction technique.  We prove the asymptotic validity of the proposed method and demonstrate its effectiveness in finite samples through simulation studies and empirical applications.\end{abstract}
\noindent%

\footnotetext[1]{E-mail addresses: \href{blj20@mails.tsinghua.edu.cn}{blj20@mails.tsinghua.edu.cn}(L.Bai), \href{wuweichi@mail.tsinghua.edu.cn}{wuweichi@mail.tsinghua.edu.cn}(W.Wu)} 
\section{Introduction}

Estimating network structures plays a fundamental role in many fields, such as finance (\cite{Marti2021}), biology (\cite{langfelder2008wgcna}) and psychology (\cite{borsboom2021network}). The correlation network is arguably the most widely used network  which
  boils down to  inferring the set of non-zero correlations, see \cite{kolaczyk2014statistical}, \cite{efron2012large} and \cite{basu2021graphical}.
In particular, for a random vector sequence $(\mf Y_t)_{t=1}^n = ((Y_{t,1}, \cdots, Y_{t,p})^{\T})_{t=1}^n $, the correlation network is defined by the association graph $\mathcal G= (\mathcal V, \mathcal E)$ with vertex set $\mathcal V = 1,2, \cdots, p$ and edge set $\mathcal E= \{(i,j) \in  \mathcal V \times \mathcal V: c_{i,j} \neq 0\}$ 
where $(c_{i,j})$ is a certain correlation-based similarity measure.
 Although the classical correlation network has provably achieved success in many applications over the last decades, the restrictive independent assumption of $(\mf Y_t)_{t=1}^n$ and the pairwise construction have limited its applicability in the massive data with complex structures arising today. For example, in bioinformatics, a lot of microarray data displays dependence among different columns, see \cite{efron2012large}. In economic and financial studies where the correlation network is often constructed from time series, the seminal work of \cite{DIEBOLD2014119} builds connectedness measures on variance decomposition, arguing that ``Correlation-based measures remain widespread, yet they measure only pairwise association and are largely wed
to linear, Gaussian thinking, making them of limited value in financial-market contexts.''  Moreover, the time-evolving feature of connectivity is often  of central interest in the cutting edge of many areas such as risk management (\cite{DIEBOLD2014119}, \cite{Marti2021}) and  dynamic operations of biological networks (\cite{kim2014inference}), necessitating new methods and theories for building correlation networks. We refer to \cref{rm:tvnetwork} for a detailed discussion of current literature on correlation and dynamic networks. 

\begin{figure}
    \centering
    \includegraphics[width = 0.8\textwidth ]{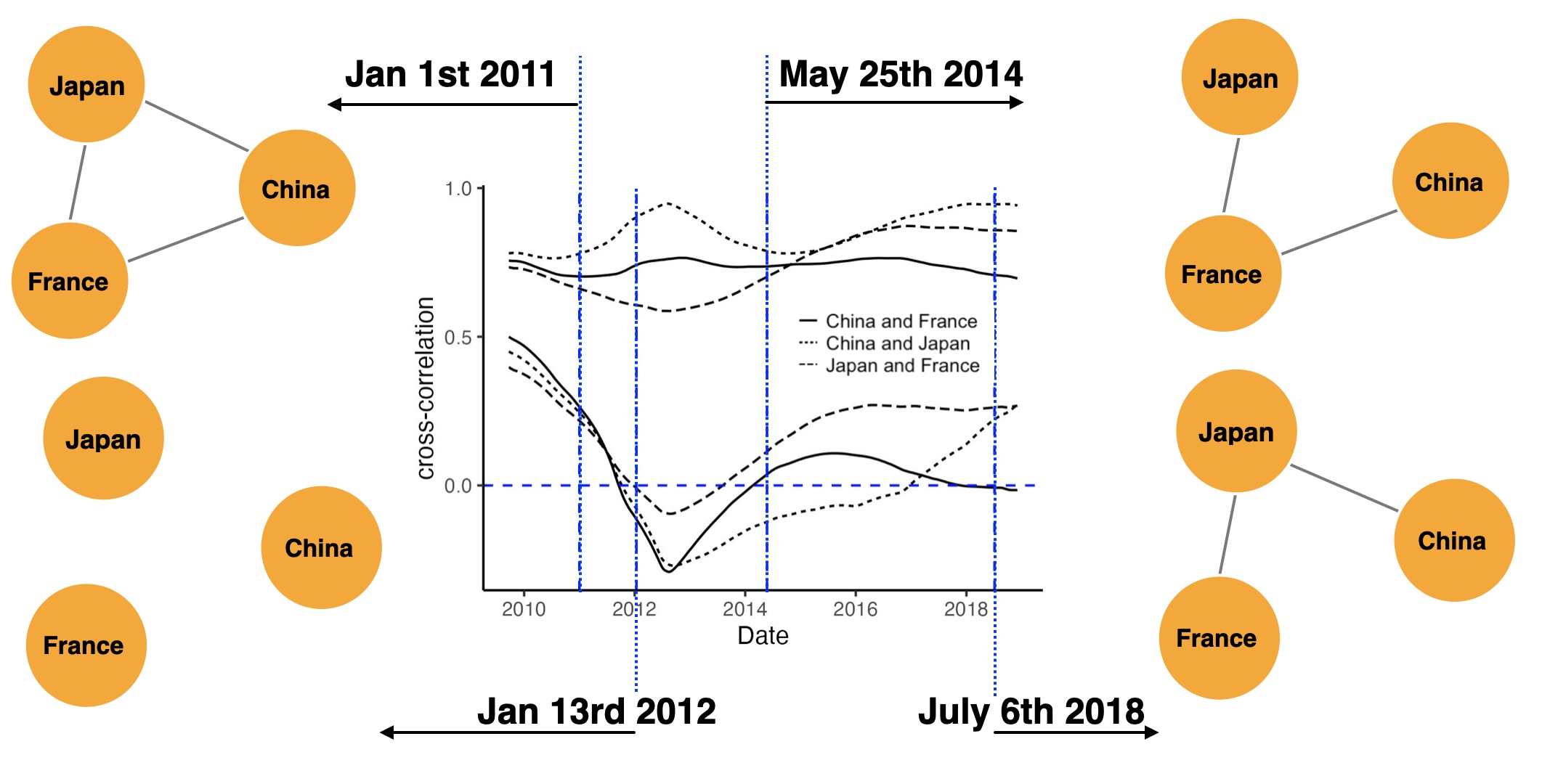}
    \caption{Middel panel: solid, dashed and dotted curves denote the SCBs of the cross-correlations of China and France, Japan and France, and China and Japan, respectively; Left and right panels: the snapshot of time-varying networks induced by $H_0 = 0$  on Jan 1st 2011, Jan 13rd 2012, May 25th 2014 and July 6th 2018, respectively. }
    \label{fig:my_label}
\end{figure}


In this paper, we establish a new framework for constructing time-varying correlation networks from non-stationary multivariate time series via multiple testing, largely addressing the drawbacks of correlation networks pointed out by \cite{DIEBOLD2014119}.  
  Specifically, we build the time-varying correlation networks on a group of correlation-based similarity measures through a {\it simultaneous} test of the hypotheses for
  $(i,j, k) \in \B$, where $\B$ is a index subset of $\{(i,j,k): i,j = 1,\cdots,p, k=0,\cdots, d_n\}$,
\begin{align}
    H_{0ijkt}: c_{ijk}(t) = g_{ijk}(t) \quad \text{versus}\quad  H_{1ijkt} : c_{ijk}(t) \neq g_{ijk}(t), \quad t \in (0,1), \label{eq:test_t}
\end{align}
where $c_{ijk}(\cdot)$ is a time-rescaled (on $(0,1)$) correlation-based measure function  between $(Y_{s,i})_{s=1}^{n-k}$ and $(Y_{s-k,j})_{s=k+1}^n,$ $g_{ijk}(\cdot)$ is a pre-specified and possibly time-varying function. Simple and meaningful choices of $g_{ijk}(\cdot)$ include but are not limited to: (1) $g_{ijk}(\cdot)\equiv 0$ which tests for the presence of correlations; (2) $g_{ijk}(t)$ are constant functions for a given index set of $(i,j)$ to test for time-invariance of the corresponding sub-graph.  
Let $T(t)=(T_{ijk}(t))$ be the corresponding test statistics. {\it We connect $i$ and $j$ \textbf{at time $t$} if there exists some $k$ such that $T_{ijk}(t)\geq c_{1-\alpha}(t,i,j,k)$ where $c_{1-\alpha}(t,i,j,k)$ is the threshold 
that controls the family-wise error rate (FWER) of  hypotheses in \eqref{eq:test_t}  at the nominal level $\alpha$.} 
For brevity and without loss of generality, we consider $c_{ijk}(t)$ as time-varying and time-lagged cross-correlations of general time series $\mf Y_t:=f(\mf Z_t,\cdots,\mf Z_{t-r})$ for $p_1$-dimensional non-stationary multivariate time series $\mf Z_t$, where $r$ is some fixed integer, and $f:(\mathbb R^{p_1})^r\rightarrow \mathbb R^{p}$ is a smooth vector function.
Importantly, we alleviate the aforementioned limitation of the pairwise construction of correlation-based networks via the flexible choices of $f$. When $\mf Y_t=\mf Z_t$, our method is related to the linear Granger-causality test. When $\mf Y_t =((\mf Z_{t}\circ \mf Z_{t})^\top,...,(\mf Z_{t}\circ \mf Z_{t-r})^\top)^\top$,  where $\circ$ stands for the Hadamard product, our method provides an alternative to \cite{DIEBOLD2014119}, gauging the connectedness in terms of volatility.   Moreover, by considering high-order polynomials for $f$, we can measure high-order dependence in time series vectors even if the data is non-Gaussian, while the classic correlation can only capture pairwise linear dependence under Gaussianity and leads to the `linear, Gaussian thinking' issue raised by  \cite{DIEBOLD2014119}.
Our framework for recovering time-varying network structures is based on a new class of  difference-based estimators and a novel bootstrap-assisted procedure, admitting flexible choices of $\B$ considered in \eqref{eq:test_t}.
Our difference-based estimator  circumvents the pre-estimation of complex trends and  can be applied to 
general non-stationary linear or nonlinear processes (\cite{zhou2010simultaneous},\cite{vogt2012}, \cite{dahlhaus2019bej}, and \cite{dette2022}) with piecewise smooth trends that allow for structure breaks commonly identified in many applications (see for instance \cite{granger2005}, \cite{bouri2019modelling}, and \cite{karavias2022structural}). 
To tackle the functional constraints in \eqref{eq:test_t} and control FWER, we provide a uniform bootstrap-assisted device for constructing simultaneous confidence bands (SCBs) for time-lagged cross-correlations of $(\mf Y_t)$. The SCBs are defined by $\hat U_{\alpha,ijk}(\cdot)$ and $\hat L_{\alpha,ijk}(\cdot)$ estimated from data such that for $\alpha\in (0,1)$ as $n \to \infty$, \begin{align}\label{SCB-original}
P(\hat L_{\alpha,ijk}(t/n)\leq \mathrm{corr}( Y_{t,i}, Y_{t-k,j}) \leq \hat U_{\alpha,ijk}(t/n),~ 1\leq t\leq n, (i,j,k) \in \B) \to 1-\alpha,
\end{align}
where 
$Y_{t,i}$ is the $i_{th}$ component of $\mf Y_t$. By the dual property of hypothesis testing and confidence sets, we construct the time-varying network by connecting $i$ and $j$ at time $s \in (0,1)$ if $g_{ijk}(s)\not \in [ \hat L_{\alpha,ijk}(s),\hat U_{\alpha,ijk}(s)]$.  We present \cref{fig:my_label} for illustration, where we infer dynamic networks from SCBs of the financial data ($S_i$ of \cref{sec:data}). We display several sub-graphs of the induced network at different times that have distinct edge connections, demonstrating the changing dynamics. By controlling the FWER through SCBs,  we are able to recover the time-varying network at all times with high probability, which is useful to investigate the evolving dynamics and transition of the interconnectedness of complex data.
Finally, we develop an easily implementable variance reduction technique that uniformly reduces the widths of the equivalent SCBs \eqref{SCB-original} and {\it increases} the recovery probability that all the pairs of nodes are correctly linked. 
To the best of the authors' knowledge, we are the first to consider the uniform variance reduction to obtain a smaller type II error rate of SCBs at a given significance level.  In particular, we build the REduction-in-widths and Difference-based SCBs (RED-SCBs) applicable to high-dimensional time series data via a nonstandard wild bootstrap procedure.  In the data analysis, equipped with the variance-reduced algorithm, we produce SCBs of cross-correlations contained fully in $[-1,1]$ when the original bands lie beyond this range. 
Theoretically, the uniform width reduction effect of RED-SCBs requires delicate analysis of the tail behavior of high dimensional Gaussian processes, which differentiates itself from its stationary univariate counterpart \cite{cheng2007reducing}. 
Based on the results of \cite{HUSLER198891} and \cite{royen2014simple}(see also \cite{latala2017royen}), our argument provides a theoretical tool to evaluate tail probabilities for the maximum of multivariate and high dimensional Gaussian vectors with flexible covariance structures, which is of independent interest.

The rest of the paper is organized as follows. In \cref{sec:main}, we present our main results on disclosing the time-varying network structure, including the estimation and inference procedures for the cross-correlation functions with many time-lags and  a simple uniform variance reduction method.  In \cref{sec:theory}, we establish a Gaussian approximation scheme for the estimated cross-correlation curves, the asymptotic FWER control of the proposed variance-reduced bootstrap-assisted algorithm, and the improvement of its overall recovery probability. We should point out that the asymptotic behavior of the maximum deviation of (auto)correlations over a fixed number of lags differs drastically from that over a diverging number of lags, making the simultaneous inference of (auto)correlations a long-standing difficult problem. We solve this problem as an important by-product.  \cref{sec:impl} gives the selection scheme of smoothing parameters.  \cref{sec:sim} reports the finite sample performance in simulation studies. In \cref{sec:data} we analyze the time-varying cross-correlation networks induced from Daily WRDS World indices. \cref{sec:future} provides conclusion remarks and discussions on future work. \cref{ap:est} presents detailed formulae for the estimators used in the RED-SCBs-based algorithm.   In \cref{sec:diffscb} we provide the detailed algorithm for time-varying cross-correlation analysis via difference-based SCBs and its theoretical properties. \cref{sec:proof} offers auxiliary results and the proofs of theorems.
The online supplement includes the algorithm of using plug-in estimators when the trends are smooth and the corresponding theoretical properties, as well as proofs of \cref{nonasynetwork} and auxiliary results.


\section{Main results}\label{sec:main}

We first summarize the notation that will be used throughout the paper before stating our results formally. For a vector $\mathbf{v} = (v_1,\cdots,v_p)  \in \R^p$, let $|\mathbf{v}| = (\sum_{j=1}^p v^2_j)^{1/2}$. For a random vector $\mathbf{V}$,  $q \geq 1$, let $\|\mf V\|_q  = (\E|\mf V|^q)^{1/q}$ denote the $\mathcal{L}_q$-norm of the random vector $\mathbf{V}$. Let $|\mf z|_{\infty}$ denote the maximum element of the vector $\mf z$.
In this paper we consider the kernel function $K(\cdot)$ that is zero outside $(-1, 1)$, and write $K_{\eta} = K(\cdot/\eta)$ for some bandwidth parameter $\eta$. For an index set $\mathbb A$, let $|\mathbb A|$ denote its cardinality. Let $\mf 1(\cdot)$ denote the indicator function, and $\overset{p}{\to}$ denote convergence in probability.

We consider the time series model of the form \footnote{The intervals $a_{i,l} \leq t<a_{i, l+1}$ can be replaced by $a_{i,l} < t\leq a_{i, l+1}$, but the results will remain the same}: for $i=1,\cdots, p,~j = 1,\cdots,n$,
\begin{align}
    Y_{j,i} = \mu_i(t_j) + \epsilon_{j,i}, \quad
    \mu_i(t)=\sum_{l=0}^{d_i} \mu_{i,l}(t)  \mf 1(a_{i,l} \leq t<a_{i, l+1}), \label{eq:model_spec}
\end{align}
where $t_j = j/n$, $t\in[0,1]$, $(\epsilon_{j,i})_{j=1}^n$ is a locally stationary process\footnote{After careful examining our theoretical arguments, our method can be applied to the piecewise locally stationary models, see \cite{zhou2013heteroscedasticity}, which allows higher-order breaks in the locally stationary models, with much more involved  mathematical arguments. For presentational simplicity we stick to the locally stationary error in this paper.}(see \cref{sec:theory} for the definition),  $\mu_i(\cdot)$ is a deterministic function on $[0,1]$ with abrupt change points $0 = a_{i,0} < a_{i,1} < \cdots < a_{i,d_i} < a_{i, d_i+1} =1$, $d_i$ is the number of change points,  $\mu_{i,l}(t)$ is Lipschitz continuous over $[a_{i,l}, a_{i, l+1}]$, and the Lipschitz constants of $\mu_{i,l}(\cdot)$ are uniformly bounded for $i=1,\cdots, p$, $0 \leq l \leq d_i$. 
To construct the equivalent SCBs of  \eqref{SCB-original},
 we start by introducing the definition of cross-correlation functions under time series non-stationarity. Define the $k_{th}$ $(k\in \mathbb Z)$ order cross-covariance, cross-marginal variance and  cross-correlation function as
\begin{align}
    \gamma_{k}^{i,l}(t) = \mathrm{Cov}(\epsilon_{ \lf nt \rf, i}, \epsilon_{\lf nt \rf + k, l}), \quad \sigma^{2}_{i,l}(t) = \gamma^{l,l}_{0}(t)\gamma^{i,i}_{0}(t), \quad \rho_{k}^{i, l}(t) = \gamma_{k}^{i, l}(t)/\sigma_{i,l}(t).\label{def:cor}
\end{align}
Notice that when $k\neq 0$,  \eqref{def:cor} indicates that in general $\rho^{i,l}_{k}(t) \neq \rho^{i,l}_{-k}(t)$ when $i\neq l$. In this case, the construction based on \eqref{eq:test_t} will yield a directed network. 
 In the rest of the paper, we will use a single index $s$ to stand for the double indices $i,l$ of the superscript   when $i=l = s$ if no confusion is caused.

\begin{remark}\label{rm:ncp}
If it's known that there are no change points, i.e., $d_i=0$, $1 \leq i \leq p$, one can estimate the correlation functions in \eqref{def:cor} via the non-parametric residual
$\hat \epsilon_{j,i}=Y_{j,i}-\hat \mu_i(t_j)$, where $\hat \mu_i(\cdot)$ is the local linear estimator of $\mu_i(\cdot)$, see \cite{zhou2010simultaneous} and \cref{sec:plg}  of the online supplement where  we provide algorithms for constructing correlation networks based on  $\hat \epsilon_{j,i}$. In the main article, we shall focus on the $d_i> 0$ case and propose difference-based estimators, as in the subsequent sections.
\end{remark}
\subsection{Difference-based sample correlation curves}
When $(\mf Y_{t})_{t=1}^n$ is contaminated by unknown abrupt change points, the direct estimation of their piecewise smooth trends will be sophisticated. In fact, there have been no methods designed for this problem. A related problem is to consistently identify a (diverging) number of abrupt change points from the piecewise smooth mean of general non-stationary time series, which has been only considered recently by \cite{wu2019multiscale}. Therefore, a straightforward approach to remove piecewise smooth trends is to separately apply the local linear estimation to the subseries between the identified jump points. However, this will cause severe boundary issues in practice, which are well-known in the field of kernel estimation, see for example \cite{fan1997}. 


Alternatively, we propose a difference-based approach for \eqref{def:cor}, which has succeeded in estimating variance (see \cite{muller1987estimation} and \cite{hall1990bio}), long-run variance (see  \cite{dette2019detecting}), and autocovariance (see \cite{gomez2017sjs} and \cite{cui2020estimation}) without the pre-estimation of the trend function.
Define  $\tilde \epsilon^{i}_{j, k} = \epsilon_{j,i} - \epsilon_{j-k, i}$.
 Let $\beta^{i,l}_{k}(t_{j}) := \E(\tilde \epsilon^{i}_{j+k, k}\tilde \epsilon^{l}_{j+k, h})$ for some large $h = o(n)$. Then under the locally stationary \cref{Ass:error} stated in \cref{sec:theory}, $\beta^{i,l}_{k}(\cdot)$ is a smooth function on $[0,1]$. Notice that  by \cref{Ass:error} we can select a positive $\tilde h<h$  such that
 $\gamma^{i,l}_{k}(t) \approx 0$ for  $|k| \geq \tilde  h$. 
 Note that for $k \leq h-\tilde h$, 
 \begin{align}
    \beta^{i,l}_{k}(t_{j}) &= E(\epsilon_{j+k, i} - \epsilon_{j,i})(\epsilon_{j+k,l} - \epsilon_{j+k-h,l}) \\ &= \gamma^{i,l}_{0}(t_{j+k}) + \gamma^{i,l}_{k-h}(t_{j}) - \gamma^{i,l}_{-h}(t_{j+k}) - \gamma^{i,l}_{k}(t_{j}) \approx  \gamma^{i,l}_{0}(t_{j}) - \gamma^{i,l}_{k}(t_{j}).
     \label{eq:beta_gamma}
 \end{align}
 
\noindent Further define $\tilde y^{i}_{j,k} =  Y_{j, i} - Y_{j-k, i}$, $\tilde \mu^{i}_{j,k} =  \mu_{i}(t_j) - \mu_{i}(t_{j-k})$, and hence $\tilde y_{j,k}^i=\tilde \mu_{j,k}^i+\tilde \epsilon_{j,k}^i$. By the piecewise smoothness of $\mu_i(\cdot)$, $1\leq i\leq p$, under some mild conditions it follows that
 \begin{align}
    \sum_{j=1}^n K_b (t_j - t) \tilde y^{i}_{j, k}\tilde y^{l}_{j, h} &= \sum_{j=1}^n K_b (t_j - t) (\tilde \epsilon^{i}_{j ,k} \tilde \epsilon^{l}_{j ,h}+ \tilde \mu^{i}_{j,k}\tilde \mu^{l}_{j,h} + \tilde \epsilon^i_{j , k}\tilde \mu^{l}_{j,h} +  \tilde \epsilon^l_{j, h}\tilde \mu^{i}_{j,k})\\ &\approx \sum_{j=1}^n K_b (t_j - t) \tilde \epsilon^{i}_{j , k} \tilde \epsilon^{l}_{j , h},
    \label{eq:def_ytilde}
\end{align}
see the proof of \cref{lm:loclin}  of the supplement for details. Hence, by the continuity of $\beta_k^{i,l}(\cdot)$, we have
\begin{align}
\sum_{j=1}^n K_b (t_j - t) E\big(\tilde y^{i}_{j, k}\tilde y^{l}_{j, h}\big)\approx \sum_{j=1}^n K_b (t_j - t) \beta^{i,l}_{k}(t_{j-k})\approx \gamma^{i,l}_0(t)-\gamma_k^{i,l}(t).
\end{align}
Therefore, we can consistently estimate $\gamma_0^{i,l}(t)-\gamma_k^{i,l}(t)$  by the local linear estimator $\hat \beta_{k}^{i,l}(t)$ such that
\begin{align}
    (\hat \beta_{k}^{i,l}(t),\hat \beta_{k}^{i,l,\prime}(t))^{\T} = \underset{(\eta_0, \eta_1) \in \R^2}{\mathrm{argmin}}\sum_{j=1}^n\left(\tilde y^{i}_{j, k}\tilde y^{l}_{j, h} - \eta_0 - \eta_1(t_j - t)\right)^2 K_{b^{i,l}_k}(t_j - t),
    \label{eq:loclin}
\end{align}
where $b^{i,l}_{k}$ is the bandwidth parameter. Define the 
centered version of $\tilde \epsilon^{i}_{j ,k} \tilde \epsilon^{l}_{j ,h}$ as
 \begin{align}
     \tilde e^{i,l}_{j,k} := \tilde \epsilon^{i}_{j ,k} \tilde \epsilon^{l}_{j ,h}-\beta^{i,l}_{k}(t_{j-k}),
     \label{eq:def_eptilde}
 \end{align}
with $\tilde e^{i,l}_{j,0} = 0$ since $\tilde \epsilon^i_{j,0}=0$ by definition.  For $(i,l,k) \in \B$, since  $\tilde \epsilon^i_{j ,k}\tilde \epsilon^l_{j ,h} $ is the product of the differences of locally stationary processes, $\tilde e^{i,l}_{j,k}$ is also locally stationary (see definitions in \cref{sec:theory}). Under mild conditions,  by stochastic expansion (see \cref{cor:suprhok} of the supplement for details),  we have 
\begin{align}
     \max_{(i,l,k) \in \B}\sup _{t \in \mathcal{T}}\left|\hat{\beta}^{i,l}_{k}(t)-\beta^{i,l}_{k}(t)-\frac{1}{nb_k^{i,l}}\sum_{j=1}^n K_{b^{i,l}_k}(t_j - t) \tilde e^{i,l}_{j,k} \right| =\op(1),\label{eq:sto_expansion}
\end{align}
where $\B$ is as defined in \eqref{eq:test_t}.
Analogous to \eqref{eq:loclin} and using \eqref{eq:beta_gamma}, the local linear estimator $\hat \beta_{h}^{i,l}(t)$  is a consistent estimator of $2\gamma_0^{i,l}(t)$. As a consequence, the cross-covariance function $\gamma^{i,l}_{k}(t)$  can be estimated by
     $
     \tilde \gamma^{i,l}_{k}(t)  = \hat \beta^{i,l}_{h}(t)/2 - \hat \beta^{i,l}_{k}(t) 
$
and the cross-correlation function $\rho_k^{i,l}(t)$ can be estimated by
\begin{align}
      \tilde \rho^{i,l}_{k}(t) = \tilde \gamma^{i,l}_{k}(t)/\tilde \sigma_{i,l}(t), \quad \tilde \sigma_{i,l}(t) = \sqrt{\tilde \gamma^{i}_{0}(t) \tilde \gamma^{l}_{0}(t)}.\label{eq:rho_network}
\end{align}
For simplicity, let $z$ denote the triad $(i,l,k)$, and we use $b_z$, $\rho_z(t)$, $\sigma_z(t)$ to denote $b^{i,l}_k$, $\rho^{i,l}_k(t)$, $\sigma_{i,l}(t)$.

\subsection{Variance reduction}\label{sec:reduction}

Albeit by definition the correlation function lies inside $[-1,1]$, 
in practice the nonparametric confidence band of the correlation curve could exceed $[-1,1]$ (see for instance Figure 5 in  \cite{zhao2015inference}) and becomes less informative and less sensitive, especially for the inference of connections.
To enhance the probability of recovering true connections, we develop a uniform variance reduction technique by interpolating at selected points to further narrow the SCBs without changing the nominal level. The improved SCBs (i.e., RED-SCBs) admit simple forms with only a slight computational cost and  project effectively the corresponding asymptotic reduction effect into finite samples. Our inspiration comes from the literature on variance reduction for $i.i.d.$ errors and pointwise inference (see for example \cite{efron1990} and \cite{cheng2007reducing}). 
However, their dependent and non-stationary counterparts remain largely untouched, let alone the extension of their pointwise reduction effects for fixed $t \in [0,1]$ to the uniform reduction effect on $(0, 1)$.
 

\par Specifically, we propose to use the following linear combination of $\hat \beta_{z}(t)$ to refine the estimate of $ \beta_{z}(t)$,   
$$
    \check{\beta}_{z}(t) = \{\check{\beta}_{z,+}(t) + \check{\beta}_{z,-}(t)\}/2,\quad 
    \check{\beta}_{z,\pm}(t)=\sum_{j=0, 1, 2}A_{j}(\pm r)\hat \beta_{z}(t-(\pm r + 1 - j)\omega_z(t)),
$$
where $A_{0}(r)=r(r-1) / 2, A_{1}(r)=\left(1-r^{2}\right),  A_{2}(r)=r(r+1) / 2$, for some selected $r \in (-1,1)$, and $\omega_z(t) = \delta(t) b_z$, $\delta(t)=\min \{\delta, (t- b) /[(r+1) b], (1- b - t) /[(r+1) b]\}$ where $\delta$ is a non-negative constant. The above estimate utilizes the correlation  between $\hat \beta_{z}(t_1)$ and $\hat \beta_{z}(t_2)$ when $t_1, t_2$ fall into the vicinity of $t$ to achieve smaller variance, while the coefficients $A_i(r)$, $i=1,2,3$ are carefully designed such that the asymptotic bias is not changed. 
In the formula of $\check \beta_{z}(\cdot)$, $\omega_z(t)$ controls the range of the smoothing neighborhood 
through $\delta$ and $b_z$. In particular, when $\delta = 0$, $\check \beta_{z}(t)$ equals the original estimator $\hat  \beta_{z}(t)$ and the variance remains unchanged. In practice, a large $\delta$ is superior when the trend function is smooth 
 and the noise level is low, while a smaller $\delta$ is preferred  when the signals contain many abrupt change points.
Our final variance-reduced estimators of cross-correlations are
\begin{align}
    \check  \rho_{z}(t) = \check \gamma_{z}(t) / \check \sigma_{z}(t),\quad \check \sigma_{z}(t) = \sqrt{\check{\gamma}^i_0(t)\check{\gamma}^l_0(t)},  \quad z = (i,l,k) \in \B,\label{eq:estimationreduce}
\end{align}
where 
$
\check \gamma_{z}(t) = \check \beta^{i,l}_{h}(t)/2 -\check \beta^{i,l}_{k}(t)$ and $\check \gamma^l_0(t) = \check \beta^l_{h}(t)/2.
$




\subsection{Controlling FWER via bootstrap-assisted inference }\label{bootstrapsec}

To motivate our bootstrap algorithm for the inference of time-varying networks,  notice that the key to valid SCBs of \eqref{SCB-original} is the quantiles of 
\begin{align}
   \max_{z\in \B}\sup_{t \in [b,1 - b]} \sqrt{nb_z}|\check \rho_z(t)-\rho_z(t)|,\label{eq:dist_net_check}
\end{align}
where $b_z$ is the bandwidth parameter and $b=b_n:=\max_{z \in \B} b_z$ converges to $0$ as $n \to \infty$, so that $\cup_{n}[b_n,1-b_n]= (0,1)$.
In addition to the time-varying data-generating mechanism, we allow $p$ and $d_n$ to be either fixed or divergent for the index set $\B$.
We shall begin with deriving the SCBs of 
\begin{align}
   \max_{z\in \B}\sup_{t \in [b,1 - b]} \sqrt{nb_z}|\tilde \rho_z(t)-\rho_z(t)|, \label{eq:dist_net}
\end{align}
and show how the result of \eqref{eq:dist_net} can lead to the solution of \eqref{eq:dist_net_check}.

The distributional properties of \eqref{eq:dist_net} have been only partially theoretically investigated in time series analysis of autocorrelations (i.e., $p=1$). 
For example, \cite{zhao2015inference} tackles the case when $p= d_n = 1$, i.e., the simultaneous inference of the local correlation curve. \cite{xiao2014portmanteau}  and \cite{braumann2021simultaneous} consider $p=1$ under the assumption of stationarity, where $\rho_k(t) = \rho_k$. 
Specifically, the theoretical conclusions of \cite{xiao2014portmanteau} necessitate that $d_n$ diverges, while the asymptotic results of \cite{braumann2021simultaneous} allow for finite $d_n$ but require the underlying time series to be {\it linear}. Due to the sophisticated distributional properties and to achieve better finite sample performance,  \cite{xiao2014portmanteau} proposes blocks of block bootstrap and \cite{braumann2021simultaneous}  develops AR-sieve-based bootstrap for the linear process. However, 
 their method cannot be applied to the inference of time-varying correlation networks, mainly because those methods are designed for stationary processes of which the correlation curves are constants. 
  


We start by investigating the stochastic expansion of the cross-correlation estimate.  Following \eqref{eq:sto_expansion} and \eqref{eq:rho_network}, we could approximate the maximum deviation of the cross-correlation estimate via $\vartheta_{z}(t)$, the moving weighted average of innovations, i.e.,
\begin{align}
 \max_{z\in\B}\sup_{t \in [b,1-b]} \sqrt{nb_z}\left|\tilde \rho_z(t)-\rho_z(t)-\vartheta_{z}(t)\right| = \op(1), \label{eq:rho_network_ap}
\end{align}
where 
$\vartheta_{z}(t) = (nb_z)^{-1}\sum_{j=1}^n  K_{b_z}(t_j - t)\Xi_{z, j}$, 
and
\begin{align}
   \Xi_{z, j} &:= (\tilde e^{i,l}_{j,h}/2 - \tilde e^{i,l}_{j,k})/\sigma_{i,l}(t_j) - 4^{-1}\rho^{i,l}_{k}(t_j) \left(\tilde e^i_{j, h}/\gamma^i_{0}(t_j) + \tilde e^l_{j, h}/ \gamma^l_{0}(t_j)\right).\label{eq:Xidiscrete}
 \end{align}
 When $i = l$,  $\Xi_{z, j}$   reduces to 
$ \sigma^{-1}_i(t_j)(\tilde e^{i}_{j,h}/2 - \tilde e^{i}_{j,k}- \rho^{i}_{k}(t_j) \tilde e^i_{j, h}/2)
$, which is the stochastic error of the 
difference-based
counterpart of (A.8) in \cite{zhao2015inference}, where the data is required to be zero-mean. The first term of \eqref{eq:Xidiscrete} is due to the approximation to $(\tilde \gamma_z(t)-\gamma_z(t))/\sigma_z(t)$, while the second term of \eqref{eq:Xidiscrete} mainly accounts for the approximation error of $(\tilde \sigma _z(t)- \sigma_z(t))/\tilde \sigma _z(t)$. 

As a result of \eqref{eq:dist_net} and \eqref{eq:rho_network_ap}, we can obtain SCBs and infer the time-varying network structures through the quantiles of $\max_{z\in \mathbb B}\sup_{t \in [b,1-b]}\sqrt{nb_z}|\vartheta_z (t)|$, where $\B$ is defined in \eqref{eq:test_t}. Let $c_z = (b/b_z)^{1/2}$, where $b=\max_{z \in \B} b_z$, and $\tilde \Gamma^2_z(t)$ be the limiting variance of $\sqrt{nb_z}|\vartheta_z(t)|$, of which the existence and non-degeneracy  are ensured by  Lemma C.3 of \cite{Dette2021ConfidenceSF} and \cref{Ass:diff}. Furthermore, we can construct time-varying networks with time-varying and edge-specific thresholds of \eqref{eq:test_t} via estimating $\tilde \Gamma_z(t)$ and deriving the quantiles of $\max_{z\in \mathbb B}\sup_{t \in [b,1-b]}\sqrt{nb_z}|\vartheta_z (t)|/\tilde \Gamma_z(t)$, which can be approximated by the maximum of a possibly high dimensional vector.
To this end, we concatenate the related random variables adjusted by variances and bandwidths in a block of dimension $|\B|$
:
\begin{align}
     \bar {\mf \Xi}_{j,s}^{\B}:= (c_z K_{b_z}(t_j - t_s)\Xi_{z,j}/\tilde {\Gamma}_z(t_s), z\in \B)^{\T}, \quad 1 \leq j,s  \leq n, ~t_s = s/n.
\end{align}
Under mild conditions, we have
$\underset{1 \leq s\leq n}  {\max}\left|\sum_{j=1}^{n}\bar{\mf \Xi}_{j,s}^{\B}/\sqrt{nb}\right|_{\infty} \approx \underset{z \in \B}{\max} \underset {t \in [b,1-b]} {\sup} \sqrt{nb_z}|\vartheta_z (t)|/\tilde \Gamma_z(t)$, see the proof of \cref{nonasynetwork} in the supplement for details. The popular method to mimic the distributional properties of $\underset{1 \leq s\leq n}  {\max}\left|\sum_{j=1}^{n}\bar{\mf \Xi}_{j,s}^{\B}/\sqrt{nb}\right|_{\infty}$ is the multiplier bootstrap using the block sums of time series vector $(\mf{\bar \Xi}^{\B, \T}_{j,1},\cdots,\mf{\bar \Xi}^{\B, \T}_{j,n})^{\T}_{1\leq j\leq n}$, see \cite{zhang2018}.  However, such an approach will be inconsistent due to the sparsity of $\mf{\bar \Xi}^{\B}_{j,s}$ caused by the bounded support of kernel $K(\cdot)$, see the discussion in Section 2.1 of \cite{Dette2021ConfidenceSF}. To address this issue, we 
compress the aforementioned sparse vector series, rearrange the blocks and obtain the following $(n-2\nb+1)|\B|$ dimensional vectors 
 \begin{align}\label{XiB}
     \bar{\mf \Xi}_{j}^{\B} = (\bar{\mf \Xi}_{j,\nb}^{\B, \T}, \bar{\mf \Xi}_{j + 1,\nb + 1}^{\B, \T}, \cdots, \bar{\mf \Xi}_{n - 2\nb +j, n - \nb}^{\B, \T})^{\T}, 1\leq j\leq 2\lceil nb\rceil,
 \end{align} such that  $\left| \sum_{j=1}^{2\nb}\bar{\mf \Xi}_j^{\B}/\sqrt{nb} \right|_{\infty} = \underset{1 \leq s\leq n}{\max}  \left|\sum_{j=1}^{n}\bar{\mf \Xi}_{j,s}^{\B}/\sqrt{nb}\right|_{\infty}$. In fact, $\bar{\mf \Xi}_j^{\B}$ in \eqref{XiB} extends the vector $\hat{\tilde Z}^{\hat \sigma}_j$ in Equation (2.26) of \cite{Dette2021ConfidenceSF} to multivariate non-stationary second-order processes allowing various bandwidths, which enables
us to develop a ``non-standard block wild bootstrap''  which recovers the desired correlation networks based on \eqref{eq:test_t}  while controlling the family-wise type I error regardless of the divergence of the cardinality of the index set $\B$ or the presence of nonlinearity and non-stationarity in the time series. The detailed algorithm is deferred to \cref{alg:jointnetwork} in \cref{sec:diffscb}.

Note that the linear combination of $\hat \beta_z(t)$ at nearby points is asymptotically equivalent to the local linear estimator of  $\beta_z(\cdot)$ using the high-order kernel  $\check K(\cdot)$, where $\check{K}(t) = (\check{K}_{+}(t) + \check{K}_{-}(t))/2$, and $\check{K}_{\pm}(t)=\sum_{j=0, 1, 2}A_{j}(\pm r)K (t+(\pm r + 1 - j)\delta)$. Therefore, the bootstrap procedure based on the high-dimensional vector \eqref{XiB} can be easily adapted 
to \eqref{eq:dist_net_check} via changing the kernel function $K(\cdot)$ by $\check{K}(\cdot)$ in $\vartheta_z(t)$, $\tilde \Gamma_z(t)$ and $\bar{\bs \Xi}^{\B}_{j,s}$, which are denoted by $\check \vartheta_z(t)$, $\check \Gamma_z(t)$ and $\check{\bs \Xi}^{\B}_{j,s}$. 
 We can define the variance-reduced estimators $\hat{\check{\mf \Xi}}_i^{\B}$ and $\hat{\check \Gamma}_{z}^{2}(t)$  as the estimators of $ {\check{\mf \Xi}}_i^{\B} $ and ${\check \Gamma}_{z}^{2}(t)$, see \cref{ap:est} for detailed formulae.  In \cref{alg:jointreduce}, we provide the full algorithm. 

\begin{algorithm}[!t]
    \caption{Time-varying cross-correlation analysis via RED-SCBs}
    \begin{algorithmic} [1]
      \State Compute $\check \rho^{i,l}_k(t)$, $(i,l,k) \in \B$ defined in \eqref{eq:estimationreduce}. 
        \State Compute the $|\B|$-dimensional vectors $\hat{\check{\mf \Xi}}_{j,s}^{\B}$ using \eqref{eq:checkXi}, $1 \leq j, s \leq n$.
        \State For a window size $w$, compute $\hat{\check{\mf S}}_{l,j}^{\B} = \sum_{s = j-w+1}^{j} \hat{\check{\mf \Xi}}_{s+l, \nb + l}^{\B}-\sum_{s = j+1}^{j+w} \hat{\check{\mf \Xi}}_{s+l, \nb + l}^{\B}$, where $l = 0,\cdots, n - 2\nb$.
        \For{$r = 1, \cdots, B$}
        \State Generate independent standard normal random variables $R_j^{(r)}$, $j=1,\cdots,n$. 
         \State Recall that $b = \max_{(i,l,k)\in \B} b^{i,l}_k$. Calculate 
          $$\check Z_{\bt}^{(r)} =\frac{\underset{0 \leq l \leq n- 2\nb}{\max}  \left| \sum_{j= w}^{ 2\nb - w} \hat{\check{\mf S}}_{l, j}^{\B}  R^{(r)}_{l+j}\right|_{\infty}}{\sqrt{2 w\nb }}.$$
           
    \EndFor
    \State Let $\check  r_{\bt}$ denote the $(1-\alpha)$-quantile of the bootstrap sample $\check Z_{\bt}^{(1)}, \cdots, \check Z_{\bt}^{(B)}$.
    \State Connect $i$ and $l$ at time $t \in [b, 1-b]$ if  $|g_{ilk}(t)-\check \rho^{i,l}_k(t)| > \check r_{\bt}  (nb^{i,l}_k)^{-1/2} \hat{\check \Gamma}^{i,l}_k(t)$  for some $k$ such that $(i,l, k) \in \B $.
    \end{algorithmic}
    \label{alg:jointreduce}
\end{algorithm} 



\cref{alg:jointreduce} is different from the conventional block bootstrap (\cite{lahiri2003resampling}) in two aspects. First, most conventional block bootstrap methods are aimed at imitating the original time series and deriving the limiting distribution of a statistic (see for instance Section 2.5 of \cite{lahiri2003resampling}), while \cref{alg:jointreduce}  approximates the maximum of the multivariate or possibly high-dimensional estimation errors.
Therefore, it produces asymptotically correct tests for \eqref{eq:test_t} even though the distribution of \eqref{eq:dist_net_check} under  finite $|\B|$  differs drastically from that under diverging $|\B|$.  Second, the summands $\hat {\check {\mf S}}_{l_1,j_1}^{\B}$ and $\hat {\check {\mf S}}_{l_2,j_2}^{\B}$ share the same Gaussian multipliers if $l_1 + j_1 = l_2 + j_2$, which contrasts
with the classic procedures where the summands are multiplied by independent Gaussian variables, see for example Theorem 5 of \cite{zhou2013heteroscedasticity}.
Furthermore,  \cref{alg:jointreduce} is suitable for efficient parallel computing. In  \cref{alg:jointreduce} we can compute $\left| \sum_{j= w}^{2\nb - w} \hat{\check{\mf S}}_{l, j}^{\B}  R^{(r)}_{l+j}\right|_{\infty}$ for separate blocks of $l$ and combine their maximums to obtain $\check r_{\bt}$, based on which we can further infer time-varying networks. 



\begin{remark}
To incorporate the prior knowledge of the network connections into our quantitative analysis, it is important to allow for flexible choices of $\B$. In \cref{sec:theory} we show that \cref{alg:jointreduce} is  consistent under very mild conditions on $\B$. In particular, $|\B|$ can be either fixed or diverging depending on the practical interest.  
To the best of our knowledge, there are no existing justified methods for inferring correlations that are valid for general time series under both scenarios of a fixed number of lags and a diverging number of lags.
\end{remark}

\begin{remark}\label{rm:tvnetwork}
Inference of correlation networks and testing of correlation matrices have been investigated by 
\cite{efron2007correlation}, \cite{cai2016}, and \cite{ bailey2019multiple}, where they focus on static networks and require independence over time. Recently, time-varying or dynamic networks inferred from time series are increasingly studied, see \cite{basu2021graphical} and \cite{chen2022}. The newly proposed methods therein construct one uniform network structure from time series with changing data generating mechanisms.  However, those methods  cannot be directly used to infer infinite-dimensional time-varying network structures caused by the complex dynamic structure and generating mechanism of the system. In contrast,  through controlling the FWER, the proposed \cref{alg:jointreduce} recovers time-varying networks at all time points with high probabilities, see the next section for theoretical guarantee.
\end{remark}

\section{Theoretical properties}\label{sec:theory}
In this section, we shall examine the theoretical performance of \cref{alg:jointreduce} in the aspects of type I error and recovery probability, while the theoretical properties of \cref{alg:jointnetwork} can be found in \cref{sec:diffscb}. To state the theoretical results rigorously, we introduce the following notation and assumptions. Assume the processes $(\epsilon_{i,j})_{i=1}^n$ in \eqref{eq:model_spec} and $(\tilde e_{z, j})_{j=1}^n$ ($\tilde e^{i,l}_{j,k}$ in \eqref{eq:def_eptilde}) admit
$$\epsilon_{j,i} = G_i(t_j,\FF_j), \quad \tilde e_{z, j} = H_z(t_j, \FF_j), \quad 1 \leq j \leq n,~1 \leq i\leq p,~z =(i,l,k) \in \B,$$
where $ G_i(\cdot, \cdot)$ and  $H_{z}(\cdot, \cdot)$ are  measurable functions of $[0,1]\times \R^{\mathbb Z} \rightarrow \R$,  
 $\FF_j=( \bs \xi_{-\infty},..., \bs \xi_j)$ is a filtration and $(\bs \xi_{i})_{i\in \mathbb Z}$ are $i.i.d.$  random elements.
 Define $\tilde \B(t) = \{z \in \B: \rho_z(t) \neq g_{z}(t)\}$, $t \in [0, 1]$. For a set $A$, let $\bar A$ denote its complement. For $t \in [0,1]$, write $\sqrt{nb_z} |g_{z}(t)-\check \rho_z(t)|/\hat{\check \Gamma}_z(t) $ as $\check T_z(t)$ and let $\check N(t) = \{\check T_z(t)
    \leq \check r_{\bt},z\in \B \cap \bar{\tilde \B}(t)\} \cap  \{\check T_z(t)
    > \check r_{\bt}, z \in \tilde \B(t)\}$.
   The FWER (conditional on data) is  $$1-P_{H_0}(\{\check T_z(t)
    \leq \check r_{\bt}, z=(i,l,k)\in \B, t\in[b,1-b]\}| \F_n),$$ where $H_0$ denotes the null hypotheses of \eqref{eq:test_t}. The recovery probability of {\it time-varying} networks conditional on data based on RED-SCBs can be expressed as
\begin{align}
  \quad \int_0^1 P(\check N(t)| \F_n) dt.
\end{align}
We shall show that \cref{alg:jointreduce} can control the FWER and enjoys the property of recovering the time-varying network structure with probability at least $1-\mathrm{FWER}$ in the subsequent analysis via SCBs.  

The crucial ingredient of our \cref{alg:jointreduce} is to mimic directly the maximum deviation of the estimated correlation curves  instead of
its limiting distribution since the limiting behavior of the maximum deviation of the difference-based non-parametric estimators for non-stationary nonlinear multivariate time series \eqref{eq:model_spec} rests on $\B$ in a complicated way.
To see this, consider univariate stationary time series $X_1$,...,$X_n$, with autocorrelations $\rho(k)$ and their sample version $\hat \rho(k)=\hat \gamma(k)/\hat \gamma(0)$ where $\hat \gamma(k)=\frac{1}{n}\sum_{i=1}^{n-h}X_iX_{i+k}$. The asymptotic distribution of  
$\sqrt{n}\max_{1\leq k\leq d_n}|\hat \rho(k)-\rho(k)|$
for fixed $d_n$ depends on the fourth order structure of the time series, while for diverging $d_n$ the asymptotic distribution is Gumbel-type and solely rests on the second-order properties, provided that $E(|X_i|^q)<\infty$ for $q>4$, see the discussion in \cite{braumann2021simultaneous}.  We avoid the difficulties caused by complicated limiting distributions by approximating the maximum deviation of the estimate via the Gaussian approximation and comparison techniques for sparse and high dimensional time series established by \cite{Dette2021ConfidenceSF}, and show that \cref{alg:jointreduce} yields valid tests adaptive to the size of $|\B|$. 

 For a process $ L(t,\FF_i)$,  we say it is $\mathcal L^q$ stochastic Lipschitz continuous (denoted by $L(\cdot, \cdot) \in \mathrm{Lip}_q$) if for $t_1,t_2\in [0,1]$, there exists a constant $C>0$ such that
  \begin{align}
      \|  L(t_1,\FF_0)-  L(t_2,\FF_0)\|_q\leq C|t_1-t_2|.
  \end{align}
 If $q \geq 2$, then the process  $ L(t,\FF_j)$ is {\it locally stationary} (LS). The LS process 
models the complex and smooth temporal dynamics of the error processes in \eqref{eq:model_spec}.  Our definition of the LS process 
is based on the Bernoulli shift process, which provides a fundamental framework for modeling nonlinear non-stationary processes, see  \cite{dahlhaus2019bej} for a comprehensive review. Let $\F^{*}_j = (\F_{-1}, \bs \varepsilon^{\prime}_0, \bs \varepsilon_1,\cdots, \bs \varepsilon_{j-1}, \bs \varepsilon_{j})$ where $\bs \varepsilon_0'$ is an $i.i.d.$ copy of $\bs \varepsilon_0$. The {\it physical dependence measure} of the nonlinear filter $ L(\cdot, \cdot) \in \mathrm{Lip}_q$ ($q>0$)  is defined by $\delta_q( L, k) = \underset{t \in [0,1]}{\sup}\| L(t,\F_k) -  L(t, \F_k^*) \|_q$, which  quantifies the influence of the input $\bs \varepsilon_0$ on the output $L(\cdot,\FF_k)$. Finally, for a univariate LS time series $L(t, \FF_j)$,  its long-run variance function is defined as 
$$ \sigma^2(L, t) = \sum_{k \in \mathbb Z} \mathrm{Cov}(L(t, \FF_0),L(t, \FF_{k})), ~t \in [0,1].$$
 
\begin{assumption}  \label{Ass:error}
    For some $q \geq 4$, the error process $\epsilon_{j,i} = G_i(t_j, \F_j)$ satisfies, 
    \begin{enumerate}[label = (A\arabic*)]
        \item  $G_i(\cdot,\cdot) \in \mathrm{Lip}_{2q}$, $i=1,2,\cdots, p$ and their Lipschitz constants are uniformly bounded.\label{A:G_LS}
         \item There exists a positive constant $t_0$ such that $\underset{t\in [0,1], i=1,2,\cdots, p}{\sup}\E(t_0\exp(G_i(t, \F_0)))<\infty$.\label{A:exp}
        \item  $\max_{i=1,\cdots, p}\delta_{2q}(G_i, l) = O(\chi^{l})$, for some $\chi \in (0,1)$ \label{A:G_delta}.
         \item  The second derivative $\gamma_{z}^{\prime\prime}(\cdot)$ of the function $\gamma_z(\cdot)$ exists and is Lipschitz continuous on $[0,1]$. The Lipschitz constants are bounded for all $z \in  \B$. 
         \label{A:gamma}
    \end{enumerate}
\end{assumption}
The conditions \ref{A:G_LS}, \ref{A:G_delta} and \ref{A:gamma} in \cref{Ass:error} are  standard in the kernel-based nonparametric analysis of LS time series. Similar assumptions have been posited by  \cite{zhao2015inference} for the inference of the univariate autocorrelation functions.
\cref{Ass:error} is mild in the sense of admitting time series with sub-exponential tails.
\begin{assumption}
  The kernel $K(\cdot)$ is a symmetric function which is zero outside $(-1,1)$ such that $\int_{-\infty}^{\infty} K(u) du = 1$, $\int_{-\infty}^{\infty} u^2 K(u) du = 0$, $\int_{-\infty}^{\infty} (K^{\prime}(u))^2 du < \infty$, and the second order derivative $K^{\prime\prime}$ is Lipschitz continuous on $(0,1)$.  \label{A:K}
\end{assumption}

It can be verified that if $K(\cdot)$ satisfies \Cref{A:K}, so does $\check K(\cdot)$ which is the equivalent kernel for the variance-reduced correlation curve estimator $\check \rho(\cdot)$.
\begin{assumption}\label{Ass:ck}
    There exists a constant $c^*>0$ such that $\min_{z \in \B}b_z/b\geq c^*$. 
\end{assumption}
\cref{Ass:ck} imposes that all $(b_z)_{z \in \B}$ share the same magnitude of orders.
\begin{assumption}\label{Ass:diff}
The following assumptions hold for $\tilde \Gamma_z(\cdot)$ and change points:
    \begin{enumerate}[label = (B\arabic*)]
    \item  The limiting  variance function $\tilde \Gamma^{2}_{z}(t):= \lim_{n \to \infty} \mathrm{Var}(\sqrt{nb_z}\vartheta_{z}(t))$ ($\vartheta_{z}(t)$ is defined in \eqref{eq:Xidiscrete}) is finite and well defined for $t\in [0,1]$, and that $\min_{z\in \B} \inf_{t \in [0,1]} \tilde \Gamma^{2}_{z}(t) > 0$.\label{S4}
    \item $\max_{z\in \B} \sup_{t \in (0,1)}|\tilde \Gamma_z^{\prime}(t)| < \infty$, where $\tilde\Gamma_z^{\prime}(t)$ is the derivative of $\tilde\Gamma_z(t)$.\label{derive}
    \item The number of abrupt  change points $d:=\max_{1 \leq i \leq p} d_i = O(n^{\phi})$, $0 \leq \phi < 2/5$, and for a sufficiently large constant $D$, $\underset{1\leq i \leq p}{\max} \underset{1\leq l \leq d_i}{\max} |\mu_{i,l}(a_{i,l}^{+}) -\mu_{i,l}(a_{i,l}^{-})| \leq D< \infty$.  \label{A:jump}
\end{enumerate}
\end{assumption}
Condition \ref{S4} guarantees that $\tilde   \Gamma^{2}_z(t)$  is non-degenerate. By the Lemma C.3 of \cite{Dette2021ConfidenceSF}, $\tilde \Gamma_z(t)$ exists and an equivalent condition of \ref{S4} is the non-degeneracy of the long-run variance of $\{\Xi_{z, j}\}_{j=1}^n$. 
It is worth noting that $\check \Gamma_z(t)$ and $\check \Gamma^{\prime}_z(t)$ also satisfy the conditions of \ref{S4} and \ref{derive}, since $\check \Gamma_z^2(t) = \lim_{n \to \infty} \mathrm{Var}(\sqrt{nb_z}\check \vartheta_{z}(t))$ and  $\check \vartheta_z(t) = (nb_z)^{-1} \sum_{j=1}^b 
\check K_{b_z}(t_j - t) \Xi_{z,j}$ is a linear combination of $\vartheta_z(t)$.
Condition \ref{A:jump} assures that the jump size is uniformly bounded. 
Note that \ref{A:jump} allows the number of change points in the mean to diverge. Recall that  $\tilde \B(t) = \{z: \rho_z(t) \neq g_{z}(t)\},~ t \in [0, 1]$.
\begin{assumption}\label{ass:alter}
   We assume that
   $
        \lambda (\{t: |\rho_z(t) - g_{z}(t)| \geq \eta_n, z \in \tilde 
        \B(t)\}) \to 1
$,
    where $\lambda$ is the Lebesgue measure on $[0,1]$, $\eta_n = \lambda_n \log(n|\B|)/\sqrt{nb}$,  where $\lambda_n \to \infty$ arbitrarily slowly.
\end{assumption}
\cref{ass:alter} imposes that if $\rho_z(\cdot) \neq g_z(\cdot)$, their absolute difference should be sufficiently large.

\begin{assumption}
The following condition holds for $\Xi_{z,j}$, $K(\cdot)$ and $\check K(\cdot)$:
For any two different triads $z$ and $z^{\prime}$ in $\B$, $0 \leq k , l \leq n-2\nb$, the absolute values of the correlations between  $\sum_{s=1}^{2\nb} \Xi_{z,  s+ k}K_{b_z}(\frac{s-\nb}{n})$ and $\sum_{s=1}^{2\nb} \Xi_{z^{\prime}, s+ l}K_{b_{z^{\prime}}}(\frac{s-\nb}{n})$  are uniformly upper bounded by a constant $0 <\rho < 1$,  and the same argument holds for $\check K(\cdot)$.
 \label{ass:rho}
\end{assumption}
\cref{ass:rho} ensures that there is no perfect correlation between the kernel-weighted errors for different pairs of nodes and lags. 

Our first result \cref{nonasynetwork} approximates the maximum deviation of the difference-based estimators for all correlation curves, i.e., $\max_{z\in \mathbb B}\sup_{t\in \mathcal T}\sqrt{nb_z}|\check \rho_z(t)-\rho_z(t)|/\check \Gamma_z(t)$, by the maximum of a Gaussian vector. Based on  \cref{nonasynetwork}, we can infer all considered correlation curves simultaneously by further investigating the Gaussian vector. The Gaussian approximation theory  for high-dimensional time series has been recently studied by for example \cite{zhang2018} and \cite{Dette2021ConfidenceSF}. The Gaussian approximation for plug-in estimators is a  direct application of existing Gaussian approximation theory, see \cref{sec:plg} of the online supplement. In contrast, \cref{nonasynetwork} provides the first Gaussian approximation scheme for second-order process and  difference-based estimators, which is crucial for developing RED-SCBs.  
\begin{theorem}\label{nonasynetwork}
Under the Assumptions  \ref{Ass:error}, \ref{A:K}, \ref{Ass:ck}, \ref{Ass:diff}, and the bandwidth conditions  $nb^4 \to \infty$, $nb^6 \to 0$, $|\B|^{1/q} h n^{\phi - 1/2}b^{-1/2-1/q} \to 0$, $|\B|^{1/q} n^{1/2}b^{7/2-1/q} \to 0$, and $(nb)^{-1/2}\{|\B|/(hn^{\phi-1}+b^4 + n^{-1/2}hb)\}^{1/(q+2)} \to 0$, there exists a sequence of zero-mean Gaussian vectors $(\tilde{\mf Z}_i)_{i=1}^{2\nb} \in \R^{(n - 2\nb+1)|\B|}$, which share the same autocovariance structure with the vectors $(\bar{\mf \Xi}^{\B}_i)_{i=1}^{2\nb} $  such that
\begin{align}
    \sup_{x \in \R}\left|\PP\left(\max_{(i,l,k) \in \B} \sup_{t \in \TT} \sqrt{nb_z}|\tilde \rho_z(t) - \rho_z(t)|/\tilde \Gamma_z(t)\leq x\right)- \PP\left(\left|\frac{1}{\sqrt{n b}} \sum_{i=1}^{2\nb} \tilde{\mf Z}_i\right|_{\infty}\leq x\right)\right| = O(\tilde \theta_n) = o(1), 
\end{align}
where  $\tilde \theta_n = (n b)^{-(1 - 11\iota)/8}
+ \Theta\left((\sqrt{nb}c_n)^{q/(q+1)} , n |\B|\right) + \Theta\left(|\B|^{1/(q+1)}(nb)^{-q/(q+1)}, n |\B|\right)$,
 $c_n = (n^{\phi-1}b^{-1}h + b^{3} + n^{-1/2}h)(|\B|/b)^{1/q}$.  
\par 
Further, there exists a sequence of zero-mean Gaussian vectors $(\check{\mf Z}_i)_{i=1}^{2\nb} \in \R^{(n - 2\nb+1)|\B|}$, which share the same autocovariance structure with the vectors $(\check{\mf \Xi}^{\B}_i)_{i=1}^{2\nb}$(replacing $K(\cdot)$ by $\check K(\cdot)$ in $(\bar{\mf \Xi}^{\B}_i)_{i=1}^{2\nb}$ ) such that
\begin{align}
    \sup_{x \in \R}\left|\PP\left(\max_{z \in \B} \sup_{t \in \TT} \sqrt{nb_z}|\check \rho_z(t) - \rho_z(t)|/\check \Gamma_z(t)\leq x\right)- \PP\left(\left|\frac{1}{\sqrt{n b}} \sum_{i=1}^{2\nb} \check{\mf Z}_i\right|_{\infty}\leq x\right)\right|  = o(1).
\end{align}
\end{theorem}
The bandwidth conditions can be fulfilled with $|\B| = O(n^{\ell})$ for some $\ell >0$, namely when $|\B|$ is of  polynomial order of $n$. 
Notably, \cref{nonasynetwork} holds also for finite $|\B|$, bridging the gap of distributional properties between finite and diverging lags of time series. 
  \par Recall that $w$ is the block size in Step 3 of \cref{alg:jointreduce}, $m_z$ and $\eta$ are the smoothing parameters for the estimator of $\check \Gamma_z(\cdot)$, see the definitions in \cref{ap:est}. 
 Let $m = \max_{z\in\B} m_z$, and $\tilde h = \lfloor M \log n \rfloor$ for some sufficiently large constant $M$ related to $\chi$ in Assumption \ref{Ass:error}. 
We use $a \vee b$ to denote $\max\{a, b\}$.  Let $\vartheta_n = \frac{\log^2 n}{w} + \frac{w}{nb} + \sqrt{\frac{w}{nb}}(n|\B|)^{4/q}$, $g_n = |\B|^{1/q}\big(\sqrt{\frac{m}{n \eta^{2}}}+m^{-1}+\eta  +(\frac{m}{nb})^{1/2}(\frac{mb}{n})^{-1/(2q)}+ (nb)^{-1/2}b^{-1/q}\big) + w^{3/2}/n + n^{\phi} h/\sqrt{w}$.  
 Let $\check{\mathcal C}_n = \{\mf x(\cdot) = (x_z(\cdot), z \in \B)^{\T} \in [b,1-b]^{|\B|}: \sqrt{nb_z}|x_z(t) - \check \rho_{z}(t)| \leq \check r_{\bt}  \hat{\check \Gamma}_{z}(t), \forall t\in [b,1-b], z \in \B\}$.  The following \cref{boot:reduce} establishes the control of FWER and the theoretical improvement in the probability of recovering time-varying networks compared with the algorithm without variance reduction (\cref{alg:jointnetwork} of \cref{sec:diffscb}), where we obtain the asymptotic correctness and uniform variance reduction effect of RED-SCBs as a by-product.
\begin{theorem}\label{boot:reduce}
    Under the conditions of \cref{nonasynetwork} and the bandwidth condition
    \begin{align}
        \vartheta_n^{1/3}\left( 1 \vee \log (n|\B|/\vartheta_n) \right)^{2/3} +\big(g_n (n|\B|)^{1/q}\big)^{q/(q+2)} \to 0.\label{eq:rate}
    \end{align} We have the following results:\par
    (i) (Type I error control.) As $n$ and $B$ go to infinity, 
    \begin{align}
         P_{H_0}(\{\check T_z(t)
    \leq \check r_{\bt}, z\in \B, t\in[b,1-b]\}| \F_n) =  P(\{\rho_{z}(\cdot), z \in \B\} \in  \check{\mathcal C}_n | \F_n) \overset{p}{\to} 1-\alpha.
    \end{align} \par
    (ii) (The improved recovery probability.)   Under Assumptions \ref{ass:alter} and  \ref{ass:rho}, for a pre-specified significance level $0 < \alpha <1 $, $|\B| = O(\log n)$,  conditional on data as $n \to \infty$, $B \to \infty$, for any $\delta > 0$ and $r \in (-1,1)$, the probability of network recovery is improved in the sense that for sufficiently small $\alpha >0 $, 
   \begin{align}
\lim_{n \to \infty} \lim_{B \to \infty}\int_0^1 P(\check N(t) | \F_n) dt  
\geq \lim_{n \to \infty} \lim_{B \to \infty} \int_0^1 P(\tilde N(t)| \F_n) dt \geq 1 - \alpha,
\label{eq:recovery}
   \end{align}
 with probability approaching $1$, where $\tilde N(t)$ is the counterpart of $\check N(t)$ without variance reduction, see \eqref{eq:define_recover}.
    \end{theorem}

\begin{remark}
 \cref{boot:reduce} admits $|\B|$ to be either fixed or diverging due to \cref{nonasynetwork}.
The conditions of  \cref{boot:reduce} can be satisfied for a sufficiently large $q$,  if $w \asymp \lfloor n^{2/5} \rfloor$, $\eta \asymp  n^{-1/7}$, $b \asymp  n^{-1/5}$, $m \asymp  \lfloor n^{2/7} \rfloor$. 
\end{remark}

As the bootstrap iteration $B$ and the sample size go to infinity, \cref{boot:reduce} shows that conditional on data the RED-SCBs achieve the nominal level asymptotically and the FWER of \cref{alg:jointreduce} is asymptotically controlled.
The rate of \eqref{eq:rate} consists of two parts:  the first term accounts for the convergence rate of the bootstrap procedure if all the errors are observed,  
while the second term reflects the convergence rates of  $\check \rho_z(\cdot)$ and $\hat{ \check \Gamma}_z(\cdot)$. 
In the SCBs $\check{\mathcal C}_n$,  $\check r_{\bt}$  accounts for the randomness and dependence of $\sqrt{nb_z}|\check \vartheta_z(t)|/\hat{\check \Gamma}_z(t)$ among different lags, dimensions and over time, while $\hat{\check \Gamma}_z(t)$  explains idiosyncratic dispersion of $\sqrt{nb_z}|\check \vartheta_z(t)|$.

The result (ii) ensures the high probability of recovering the time-varying network structure as $\alpha$ can be arbitrarily small as well as the improved probability of recovering the networks. The equality of $1-\alpha$ will be achieved when $|\tilde \B(t)| = 0$ for all $t \in [b, 1-b]$. When no variance reduction technique is used (see \cref{alg:jointnetwork} of \cref{sec:diffscb}),  the bootstrap quantile and estimated variance of the correlations are denoted as $\tilde r_{\bt}$ and $\hat {\tilde \Gamma}_z(t)$ as counterparts of $\check r_{\bt}$ and $\hat {\check \Gamma}_z(t)$.  The improved recovery probability (ii) is shown by 
    \begin{align}
        |\check r_{\bt}/ \tilde r_{\bt} - 1| = \op(1), \quad  \hat{\check \Gamma}_{z}(t)/\hat{\tilde \Gamma}_{z}(t) \overset{p}{\to} \sqrt{\check 
 \kappa/\kappa} < 1, \label{eq:uniformreduce}
    \end{align}
    uniformly for $t \in [b, 1-b]$, 
    where  $\kappa = \int K^2(t) dt$, $ \check \kappa =  \int \check K^2(t)$. According to \cref{boot:reduce} and \cref{boot:net}, the widths of the SCB for the $k_{th}$ order cross-correlation function of $Y_i$ and $Y_l$  at time $t$ of \cref{alg:jointreduce}  and \cref{alg:jointnetwork} equal $2 \check r_{\bt}  (nb_{z})^{-1/2} \hat{\check \Gamma}_{z}(t)$ and 
$2 \tilde r_{\bt}  (nb_{z})^{-1/2} \hat{\tilde \Gamma}_{z}(t)$, respectively, where $z = (i,l,k) \in \B$.   The equation \eqref{eq:uniformreduce} yields the uniform variance reduction effect with probability approaching $1$, which leads to uniformly narrower SCBs and higher network recovery probability of \eqref{eq:recovery}.  With the kernel function satisfying Assumption \ref{A:K}, when $r= 1/\sqrt{2}$ and $\delta \geq 2/(\sqrt{2} - 1)$, we attain the same reduction proportion for variances as \cite{cheng2007reducing} but in a uniform sense, i.e., $\hat{\check \Gamma}_{z}(t) \approx 5\hat{\tilde \Gamma}_{z}(t)/16 $.
Additionally, numerical studies in \cref{sec:sim} show that such asymptotic improvement projects to the finite samples.

\section{Implementation}\label{sec:impl}

 In this section we discuss the selection of smoothing parameters $w, \eta$, $m_z, b_z, z=(i,l,k) \in \B$ to implement \cref{alg:jointreduce}. For the selection of $b_z, z \in \B $,  we employ the Generalized Cross Validation (GCV) proposed by \cite{craven1978smoothing}.
We select $b_z$ by minimising
\begin{equation}
    \operatorname{GCV}(b)=\frac{n^{-1}|\mathbf{Y}_z-\hat{\mathbf{Y}}_z|^{2}}{[1-\operatorname{tr}\{\mf Q_z(b)\} / n]^{2}},\quad \hat{\mf Y}_z(b) = \mf Q_z(b)\mf Y_z, \label{GCV_b}
   \end{equation}
where the  square matrix
$\mf Q_z(b)$ depends on $b$. 
 For $k \neq 0$, we use $\mf Y_z = (\tilde y_{1,k}^i\tilde y_{1,h}^l , \cdots, \tilde y_{n,k}^i\tilde y_{n,h}^l )^{\T}$ and $\hat{\mf Y}_z=(\hat {\beta}_{z} (t_1),...,\hat {\beta}_{z} (t_n))^\top$. If $k = 0$, we consider  $\mf Y_z = (\tilde y_{1,h}^i\tilde y_{1,h}^l , \cdots, \tilde y_{n,h}^i\tilde y_{n,h}^l )^{\T}$ and $\hat{\mf Y}_z=(\hat {\beta}_{h}^{i,l} (t_1),...,\hat {\beta}_{h}^{i,l} (t_n))^\top$. 
We use $b^{i,l}_k$ for the estimation of  $ \rho_{k}^{i,l}(\cdot)$, i.e., using it in the formulae of $\hat \beta_k^{i,l}(\cdot)$,  $\hat \beta_h^l(\cdot)$, $\hat \beta_h^i(\cdot)$ and $\hat \beta_h^{i,l}(\cdot)$ of $\check \rho_{k}^{i,l}(\cdot)$.\par

 To select $w$ and $\eta$ in \cref{alg:jointreduce}, we recommend using the extended minimum volatility (MV) method proposed in Chapter 9 of \cite{politis1999subsampling}, which is robust under complex dependence structures and independent of parametric assumptions. To be concrete, 
we first propose a grid of possible block sizes $\{w_{1}, w_{2},\cdots, w_{M_1}\}$ and bandwidths $\{\eta_1, \eta_2,\cdots, \eta_{M_2}\}$. Denote the sample variance of the bootstrap
statistics by $s^2_{w_{j_1},\eta_{j_2}}$ , i.e., 
\begin{align}
 s^2_{w_{j_1},\eta_{j_2}} =  \sum_{r=w_{j_1}}^{n-w_{j_1}}\sum_{l=0}^{n - 2\nb} \hat{\check {\mf  S}}_{l,r,w_{j_1},\eta_{j_2}}^{\B, \top} \hat{\check{\mf  S}}_{l,r, w_{j_1},\eta_{j_2}}^{\B},\quad 1 \leq j_1 \leq M_1, 1 \leq j_2 \leq M_2, \label{eq:mv}
\end{align}
where $\hat{\check{\mf  S}}_{l, r,w_{j_1},\eta_{j_2}}^{\B}$ is as defined in the Step 3 of \cref{alg:jointreduce}  using  $w_{j_1}$ in the block sums,  $m_z = \lf n^{2/7} \rf$ and $\eta_{j_2}$ in $\hat{\check \Gamma}_z(t)$, $z \in \B$, see \eqref{eq:Gammacrossreduce}. Then we select $(j_1, j_2)$ which minimizes the following criterion,
   \begin{align}
    \mathrm{MV}(j,j^{\prime}):= \mathrm{SD}\left\{\cup_{r=-1}^{1}\{s^2_{w_{j}, \eta_{j^{\prime}+r}}\} \cup \cup_{r=-1}^{1}\{s^2_{w_{j+r}, \eta_{j^{\prime}}}\}\right\},\nonumber
    \end{align} where SD stands for the  sample standard deviation. Given $z=(i,l,k)\in \B$ and the selected $(w_{j_1},\eta_{j_2})$,
 we  could refine $m_z$ by the MV method, first proposing a grid of block sizes $\{m_{z,1}, \cdots, m_{z, M_3}\}$ and then  selecting the $j_3$ which minimizes 
   \begin{align}
       \mathrm{SD}(j) =  \int_{0}^1 \mathrm{SD} \left\{ 
        \hat{\check \Gamma}^2_{z, m_{z,(r + j)}, \eta_{j_2}}(t), r = -1,0,1 \right\}dt,
    \end{align}
    where $\eta_{j_2}$ is the smoothing parameter we select in the last step.

\section{Simulation}\label{sec:sim}
 We consider the following multivariate time series models, $\mf Y_j = \bs \mu_j + \mf G(t_j, \FF_j)$, where $\mf G(t, \FF_j) = ( G_1(t,  \FF_j), G_2(t, \FF_j), G_3(t,  \FF_j))^{\T}$, $\bs \mu_j := \bs \mu(t_j) :=(\mu_1(t_j), \mu_2(t_j), \mu_3(t_j))^{\top}$, $1 \leq j \leq n$. Let $(\xi_{1j})$, $(\xi_{2j})$, $(\xi_{3j})$ be $i.i.d.$ standard Gaussian random variables, $\bs \xi_j=(\xi_{1j},\xi_{2j},\xi_{3j})^\top$ and $\FF_j=(\cdots, \bs\xi_j).$ 
 For the error process, we consider 
 \begin{align}
 \mf G(t, \FF_j) =\begin{pmatrix}
     0.075 &0,&0\\
     0& 0.15(0.9 + 0.1\sin(2\pi t) )&0\\
     0&(0.9 + 0.1t)&0.1
 \end{pmatrix} \mf G(t, \FF_{j-1}) + \bs \xi_j.
 \end{align}
 The mean function  $\bs \mu_j$ in the simulation is 
\begin{align}
      \mu_1(t) &= 4 - 0.5\sin(4t) + 0.5t ,\\ 
\mu_2(t)& =(1 - (t-0.5)^2)
\mf 1( 0 \leq t < 0.4 ~\text{or}~0.5 \leq t < 0.6) \\ &+ (3-0.5\sin (4t)) 
\mf 1(  0.4 \leq t < 0.5 ~\text{or}~ 0.6 \leq t \leq 1),\\
\mu_3(t) & = 0.3 \mf 1(0 \leq t < 0.35) +0.7 \mf 1(0.35 \leq t <0.55) + 0.2 \mf 1(0.55 \leq t \leq 1).
 \end{align}
In the simulation, we examine (i) the empirical coverage rates of the SCBs for the lag-$1$ cross-correlation functions; (ii) the simulated recovery rates of the time-varying correlation networks induced by the multiple testing problems \eqref{eq:test_t} with $g_{ij1}(t) \equiv 0.3+0.3t$, $1 \leq i \neq j \leq 3$. Notice that the  cross-correlation functions considered are asymmetric, i.e.,$\gamma^{2,3}_1(t) \neq \gamma^{3,2}_1(t)$.\par 
 We implement \cref{alg:jointreduce} and \cref{alg:jointnetwork} using the selection scheme introduced in \cref{sec:impl} and \cref{sec:diffscb}, respectively. For RED-SCBs, we select $r = 1/\sqrt{2}$ as recommended in \cite{cheng2007reducing} and $\delta = 1.3$ since there may exist many change points in trends. We first investigate the sensitivity of the proposed algorithms against the perturbation of bandwidths $b^{i,l}_k$. 
In \cref{tb:bandwidth} the lines of $b = \pm 10\%$ report the average widths, simulated coverage rates, and simulated network recovery rates when using $1 \pm 10\%$ of the GCV selected bandwidths. The results of \cref{tb:bandwidth} show that both \cref{alg:jointreduce} and \cref{alg:jointnetwork} are not sensitive to the perturbation of the bandwidths and demonstrate that both algorithms yield asymptotic correct SCBs since the widths of the SCBs decrease with sample size. The variance reduction technique works very well for the average widths of RED-SCBs are about $10\%$ narrower than the original ones on average when their simulated coverage rates are all close to the nominal levels. Moreover, we conduct two-sample t-tests at the nominal level of $0.1\%$ and find the variance reduction effect is significant. More importantly, equipped with RED-SCBs the simulated recovery rates of \cref{alg:jointreduce} outperform those using \cref{alg:jointnetwork} and 
rise close to $1$ as the sample size increases.
 \begin{table}[ht]
\centering
\small
\begin{tabular}{rrrrrrrrrrrrrr}
  \hline
 & &\multicolumn{6}{c}{\cref{alg:jointreduce}}
 &\multicolumn{6}{c}{\cref{alg:jointnetwork}}\\ 
  \hline
 & &\multicolumn{2}{c}{width} &\multicolumn{2}{c}{coverage}&\multicolumn{2}{c}{recovery }&\multicolumn{2}{c}{width} &\multicolumn{2}{c}{coverage}&\multicolumn{2}{c}{recovery }\\ 
 $n$ & $b$   & 90\%         & 95\%       & 90\%        & 95\%  & 90\%         & 95\% 
 & 90\%         & 95\%       & 90\%        & 95\%   & 90\%         & 95\%  \\
  \hline
&-10\% & 0.76 & 0.86 & 87.0 & 94.4 & 96.3 & 94.0 & 0.68 & 0.77 & 86.2 & 93.9 & 93.5 & 90.5 \\ 
500 & 0 & 0.70 & 0.79 & 90.9 & 96.0 & 97.8 & 96.3 & 0.64 & 0.73 & 87.2 & 93.6 & 95.4 & 93.1 \\ 
 & 10\% & 0.67 & 0.76 & 92.0 & 96.6 & 99.2 & 98.5 & 0.60 & 0.69 & 88.9 & 95.1 & 97.0 & 95.2 \\ 
   \hline
&-10\% & 0.63 & 0.71 & 87.9 & 94.8 & 98.9 & 97.7 & 0.56 & 0.63 & 87.6 & 94.4 & 96.7 & 94.5 \\ 
800  &0 & 0.59 & 0.67 & 89.2 & 95.4 & 99.3 & 98.5 & 0.53 & 0.60 & 87.1 & 94.6 & 98.1 & 96.7 \\ 
  &10\% & 0.56 & 0.63 & 89.7 & 95.9 & 99.7 & 99.3 & 0.50 & 0.57 & 88.8 & 94.6 & 98.9 & 97.9 \\ 
  \hline 
&-10\% & 0.58 & 0.66 & 86.6 & 94.6 & 99.3 & 98.6 & 0.51 & 0.58 & 88.5 & 94.4 & 97.9 & 96.2 \\ 
  1000 &0 & 0.55 & 0.62 & 90.0 & 95.3 & 99.6 & 99.1 & 0.48 & 0.55 & 87.3 & 94.2 & 98.8 & 97.6 \\ 
  &10\% & 0.51 & 0.58 & 90.9 & 96.6 & 99.8 & 99.6 & 0.46 & 0.52 & 89.1 & 95.4 & 99.2 & 98.5 \\ 
 \hline
\end{tabular}
\caption{Average widths, simulated coverage rates (in $\%$) and simulated recovery rates  (in $\%$) of \cref{alg:jointreduce} and \cref{alg:jointnetwork} at the nominal levels of $90\%$ and $95\%$ for $\B = \{(i,j,k), i=1,2,3, j=1,2,3, i\neq j , k=1\}$ with sample sizes $n = 500, 800, 1000$. }
\label{tb:bandwidth}
\end{table}

To further illustrate the uniform variance effect of RED-SCBs based methods, we plot the time-varying simulated false negative rates (FNR) of networks induced by $g_{ij1}(t) \equiv 0.3+0.3t$, $1 \leq i \neq j \leq 3$, which is defined by 
\begin{align}
    \mathrm{FNR}(t) = \frac{\sum_{z \in \B} \mf 1(\rho_z(t)\neq  g_{ij1}(t),  T_z(t) \leq r_{\bt})}{\sum_{z \in \B} \mf 1(\rho_z(t)\neq  g_{ij1}(t))},
\end{align}
where $T_z(t) = \tilde T_z(t)$, $r_{\bt} = \tilde r_{\bt}$ for the difference-based SCBs, and $T_z(t) = \check T_z(t)$, $r_{\bt} = \check r_{\bt}$ for the RED-SCBs.
\cref{fig:tvfn} shows that RED-SCBs reduce FNR uniformly in time and as the samples size increases, the power of both SCB-based methods for the inference of time-varying networks grows close to $1$ uniformly for each time point.
\begin{figure}
    \centering
    \includegraphics[width= 0.45\textwidth]{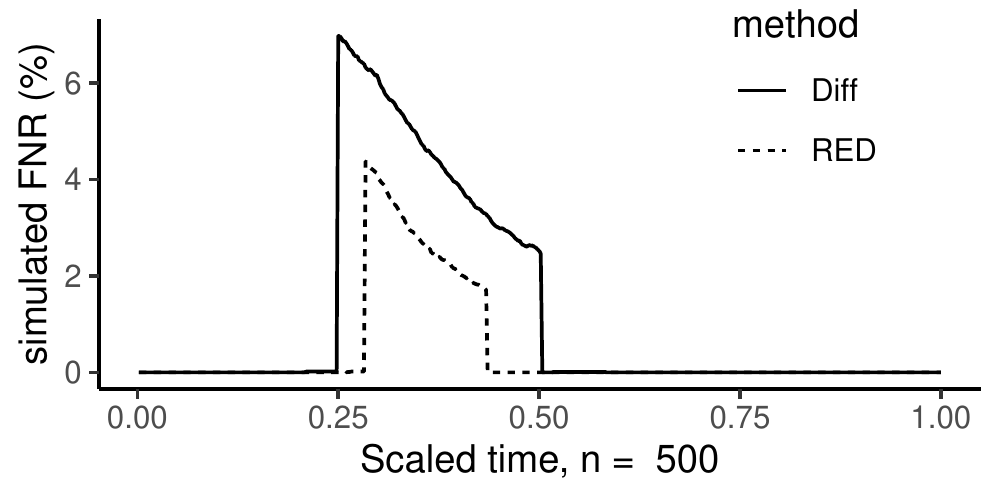} \includegraphics[width= 0.45\textwidth]{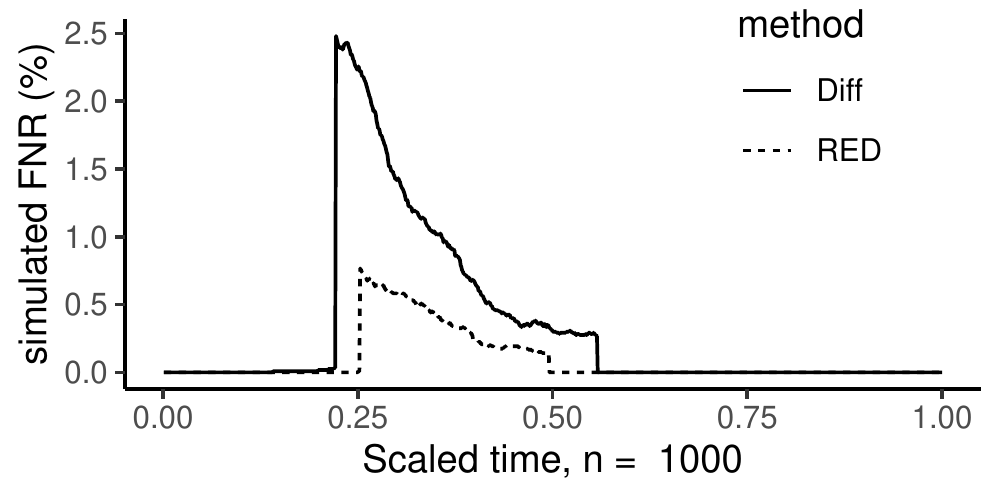}
    \caption{ Solid lines and dotted lines denote time-varying simulated false negative rates of \cref{alg:jointreduce}(RED) and \cref{alg:jointnetwork}(Diff), respectively. Left panel: $n=500$. Right panel: $n=1000$.}
    \label{fig:tvfn}
\end{figure}

\section{Data anaylsis}\label{sec:data}
Since the financial crisis of 2007-2009, it has been widely recognized that network structures and connections are 
central to risk measurement, for example, the interconnectedness among assets, credit, and intersectoral input-output linkages are
important for quantifying risks in financial portfolios, systemic and macroeconomic risks, respectively (see \cite{puliga2014credit}, \cite{DIEBOLD2014119}, \cite{acemoglu2012network} and reviews therein).  Among others, \cite{brunetti2019interconnectedness} classifies the financial network as  `physical network' and `correlation network'. The former is based on bank transactions and agent choices, while the latter,  which is also the focus of this paper, is built on exploring correlation-related interactions among asset returns. Moreover, correlation networks are essentially time-varying  and evolving which have been frequently investigated for the forecast of financial crises, see \cite{DIEBOLD2014119}.
In this section, we analyze the time-varying correlation networks of Daily WRDS World Indices, which are market-capitalization weighted
indices with dividends at daily frequency, see \url{https://wrds-www.wharton.upenn.edu/pages/get-data/world-indices-wrds/} for details. Specifically, we investigate the weekly average of Daily Country Return with Dividends (PORTRET) from June 1st, 2006 to June 1st, 2022. There are 819 time points for each country. We denote the weekly average of the absolute value of  PORTRET by $S_i$ and the weekly average of the squared logarithm of PORTRET by $R_i$, which are frequently employed as measures of volatility and risk.  First, we examine the independent assumption for each country by testing for the autocorrelations of $R_i$ and $S_i$ of lags 1,2,3 simultaneously and find RED-SCBs reject the null hypotheses of zero serial correlations of $R_i$ of China and Germany, and $S_i$ of Germany, France, and Japan at the significance level of $0.5\%$, respectively. Hence, the independent assumption of classical methods can be restrictive for these time series of indices. To infer the underlying time-varying network structure from this data set that has complex dynamics and trends, we employ \cref{alg:jointreduce} which is based on RED-SCBs. For better visualization of the time-varying nature of the connections, we display in the left panels of  \cref{fig:corr_net} and \cref{fig:corr_net_R} the heatmaps of confidence levels (i.e., the smallest confidence level that the SCBs cover $g_{ijk}(t)\equiv 0$, $(i,j,k) \in \B$, $t \in [b, 1-b]$). 
Notice that the correlation networks can be recovered from the heatmaps  via linking two nodes where the corresponding confidence levels are greater than a pre-specified threshold, see the right panels of \cref{fig:corr_net} and \cref{fig:corr_net_R}.   For the sake of simplicity, we use CHN, HK, DE, GB, FR, and JP  short for China, Hong Kong, Germany, the United Kingdom, France, and Japan. 
For brevity, \cref{fig:corr_net} and \cref{fig:corr_net_R} only display the time-varying networks with selected time points of June 1st of each year (the years with no connections are omitted), while \cref{fig:my_label} shows the sub-graphs at important time points based on $S_i$. Since our inference results hold simultaneously over time, other time points of interest can also be visualized similarly. 
As shown in  \cref{fig:corr_net}, around 06/01/2012 and 06/01/2018 when Hong Kong witnessed stock market falls, there are major changes in the connections between Hong Kong and the rest of the world in terms of absolute market-capitalization weighted indices. By contrast, Japan maintained significant connections with Europe countries (GB, FR, and DE) from 2010 to 2018.\par
\cref{fig:corr_net_R} illustrates the time-varying cross-correlation networks constructed from $R_i$. The aforementioned stock market crashes of Hong Kong in 2012 and 2018 have been identified from these volatility networks of 2011 and 2017 when there are drastic decreases in confidence in the cross-correlations between Hong Kong and the rest of the world. In the year of 2012, a series of financial crises including the European debt crisis occurred, Standard$\&$Poor’s downgraded France and eight other eurozone countries, and  connections among eurozone countries weakened in general. However, during the years of recovering from a series of financial crises, the interconnectedness between Japan and the United Kingdom remained relatively strong due to the favoring policies and relations, including the agreement to jointly develop weapons systems in 2012, led by Prime Minister David Cameron and the Brexit agreement on trading in 2016.
Between 2012 and 2016 the confidence levels of connections between China and the United Kingdom or Japan also fell violently, which might be driven by changing diplomatic and trading relations among those pairs of countries.
\begin{figure}
    \centering
     \includegraphics[width=0.45\linewidth]{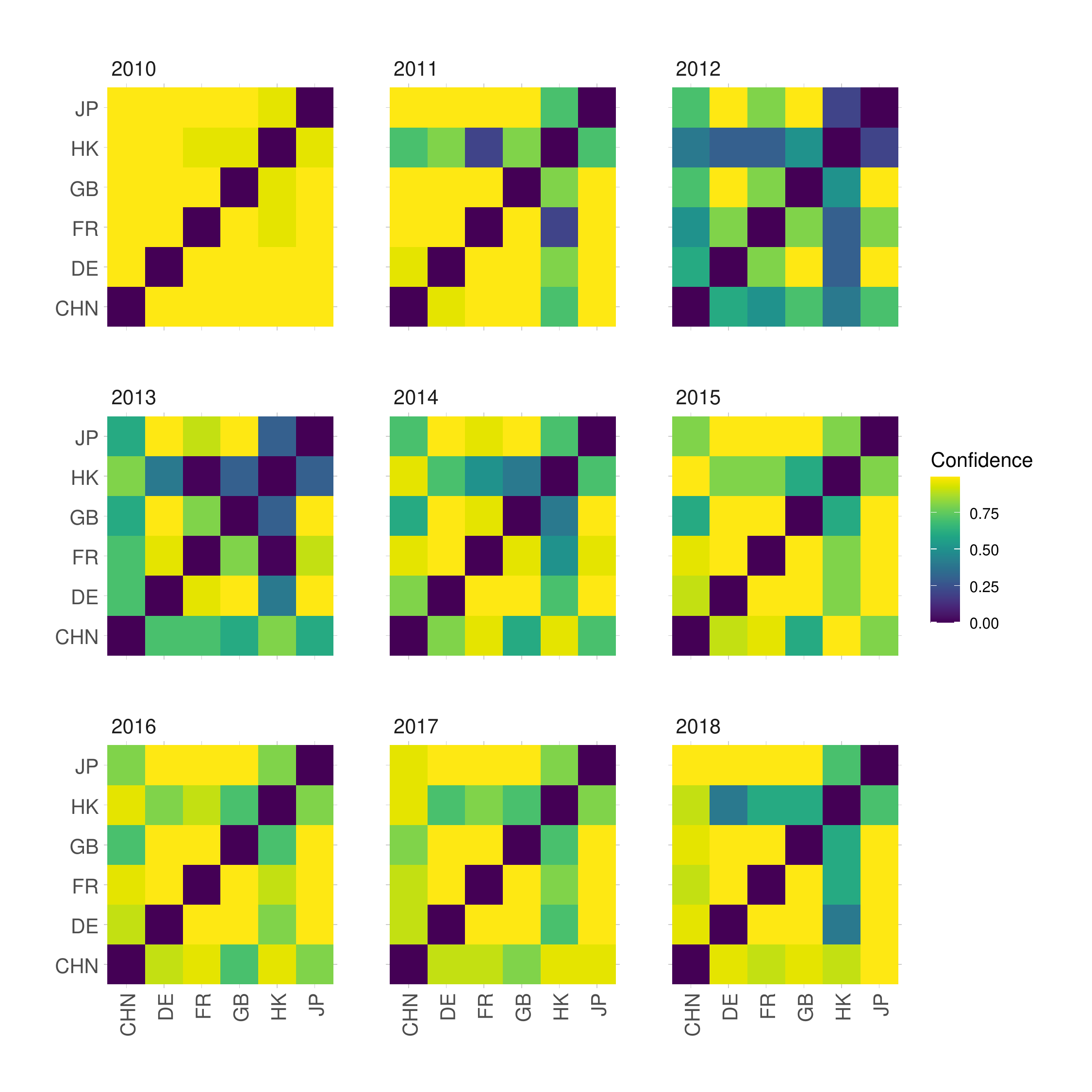}
    \includegraphics[width=0.45\linewidth]{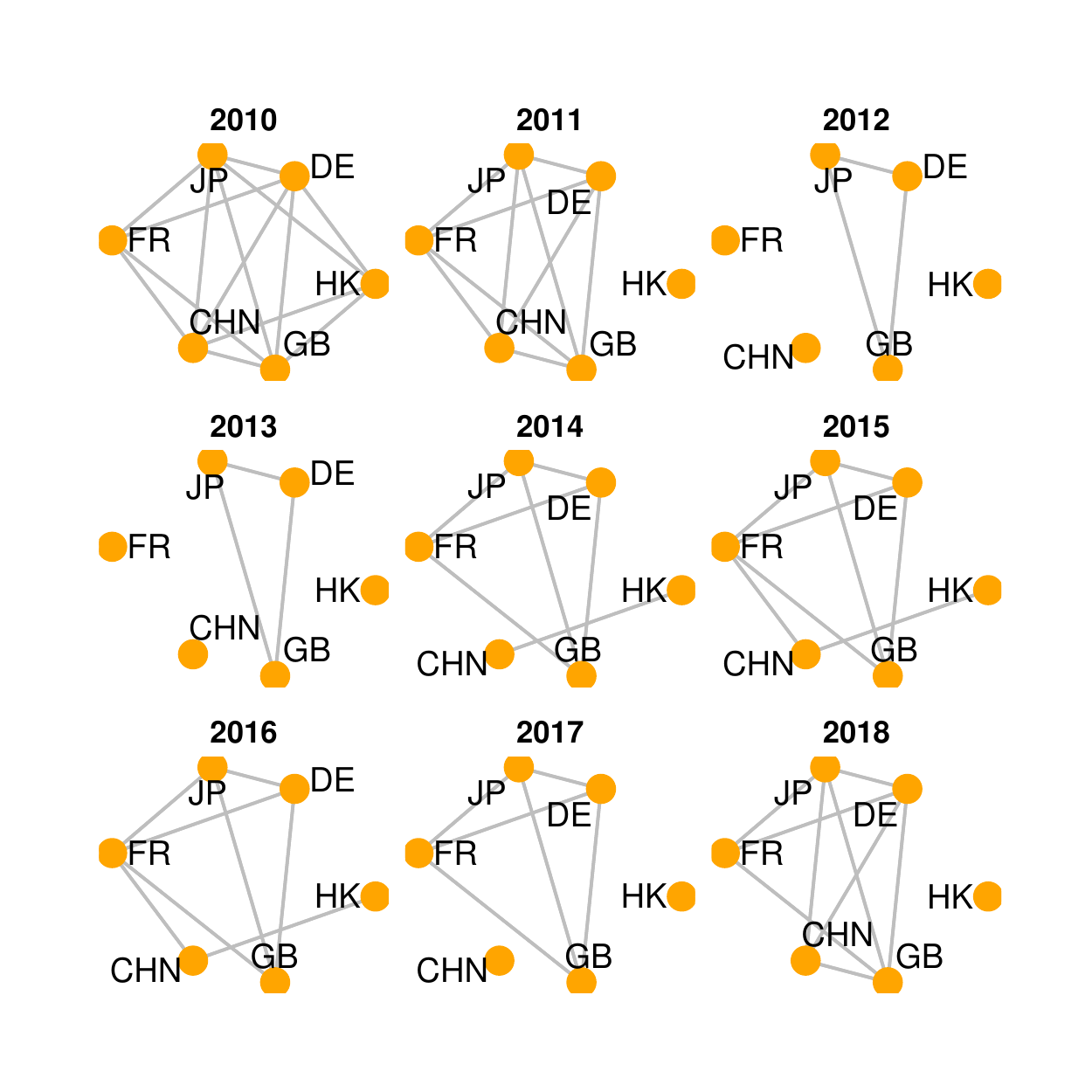}
    \caption{Left panel: the heatmap of confidence levels of cross-correlations of $S_i$ from 2010 to 2018 using  \cref{alg:jointreduce}.  Right panel: the snapshots of induced networks with threshold $0.9$.}
    \label{fig:corr_net}
\end{figure}

\begin{figure}
    \centering
     \includegraphics[width=0.45\linewidth]{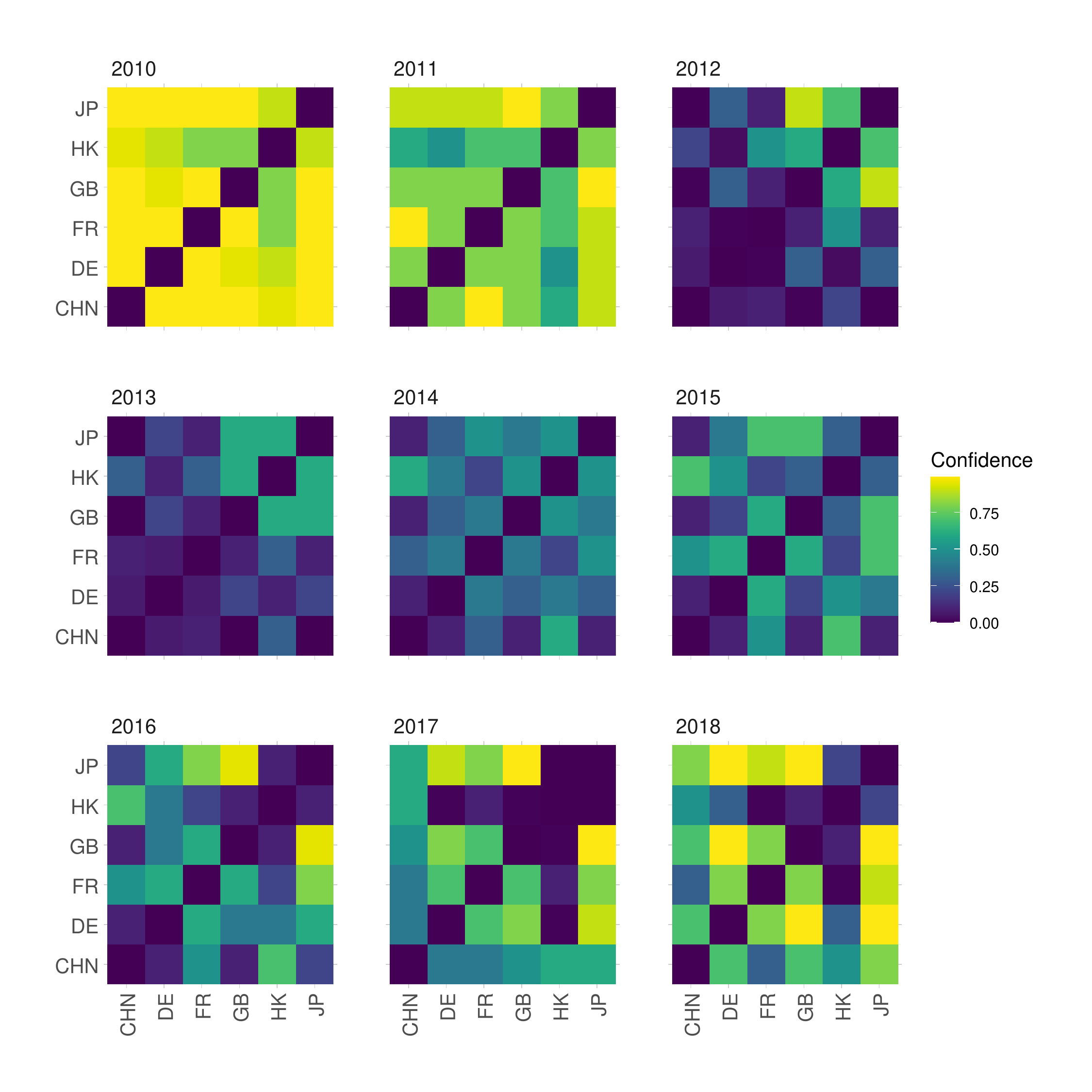}
    \includegraphics[width=0.45\linewidth]{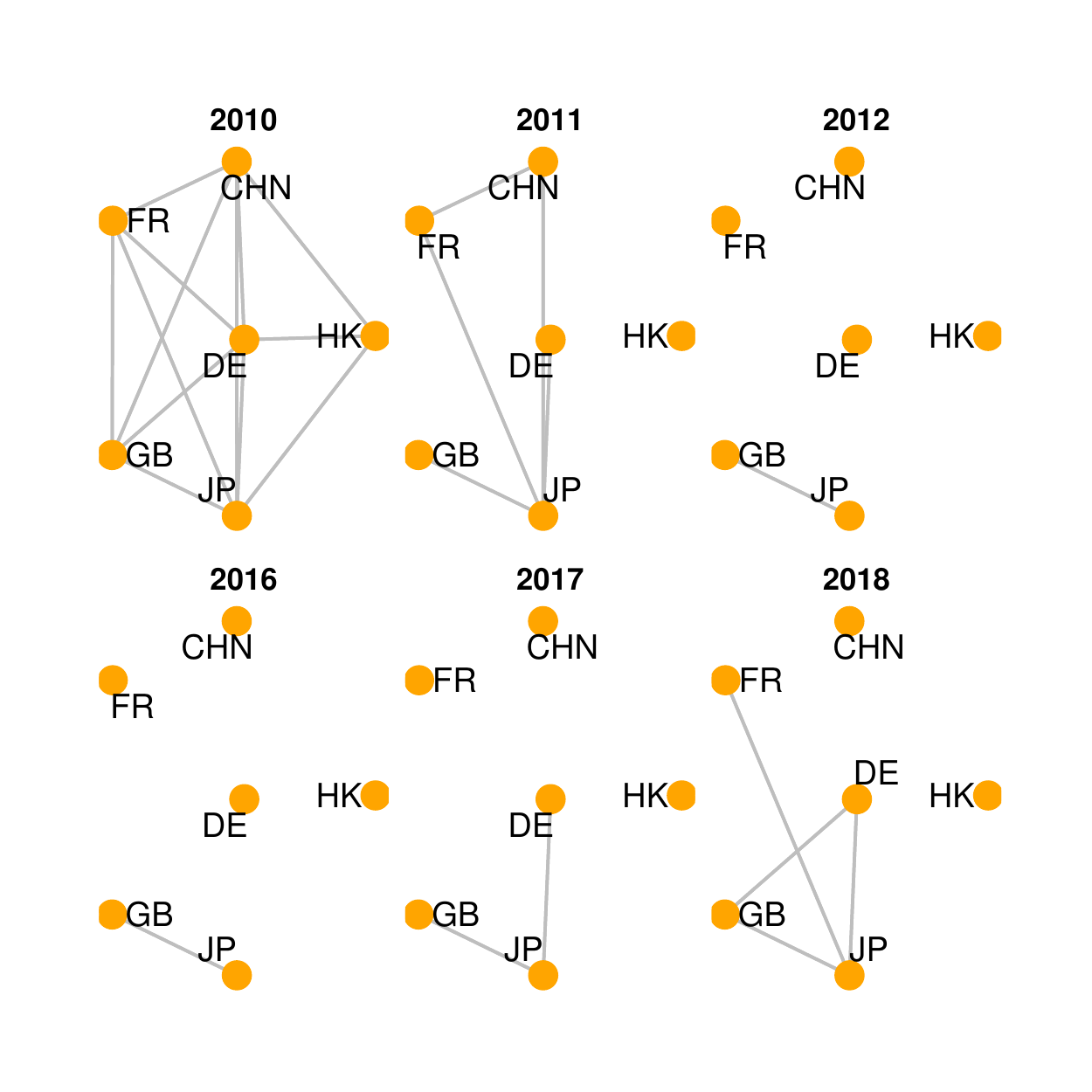}
    \caption{Left panel: the heatmap of confidence levels of cross-correlations of $R_i$ from 2010 to 2018 using  \cref{alg:jointreduce}.  Right panel: the snapshots of induced networks with threshold $0.85$, where networks with no connections are omitted.}
    \label{fig:corr_net_R}
\end{figure}

 \section{Conclusion and future work}\label{sec:future}
 In this paper, we propose a unified framework for inferring time-varying cross-correlation networks from multivariate time series vectors with complex trends, where the dimension of the vectors and the number of correlation-based measures can be both fixed and divergent. We leverage the difference-based estimators to circumvent the pre-estimation of unknown change points in the piecewise smooth trends and develop a bootstrap-assisted approach to infer the time-varying network structures under time series nonstationarity. 
As a by-product, we bridge the gap of correlation inference between fixed and diverging lags in the time series literature.  
We improve the probability of recovering time-varying networks by extending the variance reduction technique for univariate time series nonparametric estimation to the uniform width reduction of SCBs for non-stationary and nonlinear time series. We provide theoretical justifications for our proposed bootstrap-assisted algorithms and evaluate the finite sample performance by simulations studies and the analysis of a financial data set.

The interconnectedness in the dynamic networks can be gauged by not only cross-correlations but also other measures such as Granger-causality, variance decomposition, and cross-quantilogram, see \cite{ANDRIES2022100963} for a thorough review.  Statistical properties of other measures for the inference of large-scale networks constructed from general and possibly nonlinear and non-stationary processes have not been thoroughly investigated yet. Besides FWER, the false discovery rate (FDR) is also a popular and powerful tool for multiple testing and building networks.  We leave constructing networks through other measures or via other criteria including FDR as rewarding future work. 




\section*{Acknowledgement} Weichi Wu is the corresponding author and gratefully acknowledges NSFC (12271287).

\bigskip
\begin{center}
{\large\bf SUPPLEMENTARY MATERIAL}
\end{center}

\begin{description}

\item[Title:] Supplement to ``Time-varying correlation network analysis of non-stationary multivariate time series with complex trends"\par
The algorithm of using plug-in estimators when the trend functions are smooth, its theoretical properties, proofs of Theorem 3.1 and auxiliary results  are included in the supplement.
\end{description}
\appendix 
\section{Detailed formulae for Algorithms}\label{ap:est}
This section gives the exact formulae of the  estimators used in \cref{alg:jointreduce} and \cref{alg:jointnetwork}. Recall the definitions in  \cref{boot:reduce}, $\kappa = \int K^2(t) dt$, $ \check \kappa =  \int \check K^2(t) dt $. 
\subsection{Estimators for \texorpdfstring{\cref{alg:jointreduce}}{Algorithm 2} }
Recall that the definitions of $\check \beta_z(t)$, $\check \gamma^i_{0}(t)$, $\check \sigma_{i,l}(t)$,  $\check \rho_z(t)$ and $\check K(t)$, $z = (i,l,k) \in \B$, $1 \leq i, l \leq p$,
 in \cref{sec:reduction}. Define the variance-reduced residuals  $\hat{\check e}^{i,l}_{j,k}$ as $ \tilde y_{j,k}^i  \tilde y_{j, h}^i - \check \beta_{k}^{i,l}(t_j)$.
Then, the variance-reduced estimators of $ \Xi_{z,j}$ and $ \check \Gamma_{z}(t)$  are 
\begin{align}
    \check  \Xi_{z, j} = (\tilde y^i_{j, h} \tilde y^l_{j,h}/2 - \tilde y^i_{j,k} \tilde y^l_{j, h})/\check \sigma_{i,l}(t)-4^{-1}\check \rho^{i,l}_{k}(t_j) \left(\tilde y_{j,h}^{i,2}/\check\gamma^i_{0}(t_j) + \tilde y_{j,h}^{l,2}/\check \gamma^l_{0}(t_j) \right), \label{eq:Xicrossreduce}
\end{align}
and \begin{align}
    \hat{\check \Gamma}_z^{2}(t) = \frac{\check \kappa}{m} \sum_{i=1}^{n}\check \Delta^{2}_{z,i}\omega(t,i),\quad \omega(t,i) = K_{\eta}(t-t_i)/\sum_{j=1}^n K_{\eta}(t_j -t),\label{eq:Gammacrossreduce}
\end{align}
where $\check \Delta_{z,s} = \sum_{j=s}^{s+m -1} \hat{\check \Xi}_{z,j}$, and $\hat{\check \Xi}_{z,j} = (\hat{\check e}^{i,l}_{j,h}/2 - \hat{\check e}^{i,l}_{j,k})/\check \sigma_{i,l}(t)-4^{-1}\check \rho^{i,l}_{k}(t) \big(\hat{\check e}^i_{j,h}/\check\gamma^i_0(t_j) + \hat{\check e}^l_{j,h}/\check \gamma^l_0(t_j) \big)$. 
We then define $\hat{\check{\mf \Xi}}_j^{\B} $, $1 \leq j \leq 2\nb$, as
\begin{align}
   \hat{\check {\mf \Xi}}_{j}^{\B} = (\hat{\check {\mf \Xi}}_{j,\nb}^{\B, \T}, \cdots, \hat{\check {\mf \Xi}}_{n - 2\nb +j, n - \nb}^{\B, \T})^{\T},~\text{with}~\hat{\check{\mf \Xi}}_{j, s}^{\B} = (c_z K_{b_z}(t_j - t_s)\check \Xi_{z, j}/ \hat{\check \Gamma}_z(t_s), z\in \B)^{\T}.
   \label{eq:checkXi}
\end{align}

\section{Time-varying cross-correlation analysis via difference-based SCBs}\label{sec:diffscb}
In this section, we list the algorithm based on difference-based SCBs without variance reduction, and provide the estimators therein and its theoretical properties.

 \begin{algorithm}
    \caption{Time-varying cross-correlation analysis 
    }
    \begin{algorithmic} [1]
        \State Compute $\tilde \rho^{i,l}_k(t)$, $(i,l,k) \in \B$ defined in \eqref{eq:rho_network}. 
        \State Compute the $|\B|$-dimensional vectors $\hat{\bar{\mf \Xi}}_{j,s}^{\B}$ using \eqref{eq:hatXi}, $1 \leq j,s \leq n$.
        \State For a window size $w$, compute $\hat{\tilde{\mf S}}_{l,j}^{\B} = \sum_{s = j-w+1}^{j} \hat{\bar{\mf \Xi}}_{s+l, \nb + l}^{\B}-\sum_{s = j+1}^{j+w} \hat{\bar{\mf \Xi}}_{s+l, \nb + l}^{\B}$, where $l = 0,\cdots, n - 2\nb$.
        \For{$r = 1, \cdots, B$}
        \State Generate independent standard normal random variables $R_j^{(r)}$, $j=1,\cdots,n$. 
         \State Recall that $b = \max_{(i,l,k)\in \B} b^{i,l}_k$. Calculate 
          $$\tilde Z_{\bt}^{(r)} =\frac{\underset{0 \leq l \leq n- 2\nb}{\max}  \left| \sum_{j= w}^{ 2\nb - w} \hat{\tilde{\mf S}}_{l, j}^{\B}  R^{(r)}_{l+j}\right|_{\infty}}{\sqrt{2 w\nb }}.$$
           
    \EndFor
    \State Let $\tilde  r_{\bt}$ denote the $(1-\alpha)$-quantile of the bootstrap samples $\tilde Z_{\bt}^{(1)}, \cdots, \tilde Z_{\bt}^{(B)}$.
    \State Connect $i$ and $l$ at time $t \in [b,1-b]$ if  $|g_{ilk}(t)-\tilde \rho^{i,l}_k(t)| > \tilde r_{\bt}  (nb^{i,l}_k)^{-1/2} \hat{\tilde \Gamma}^{i,l}_k(t)$  for some $k$ such that $(i,l, k) \in \B $.

    \end{algorithmic}
    \label{alg:jointnetwork}
\end{algorithm}

\subsection{Estimators}
 Recall that $\tilde y^i_{j,k} = Y_{j,i} - Y_{j-k, i}$. For the estimation of $\Xi_{z, j}, z =(i,l,k) \in \B$ in \eqref{eq:Xidiscrete}, we use the following difference-based  estimator:
\begin{align}
    \tilde  \Xi_{z, j} = (\tilde y^i_{j, h} \tilde y^l_{j,h}/2 - \tilde y^i_{j,k} \tilde y^l_{j, h})/\tilde \sigma_{i,l}(t_j)-4^{-1}\tilde \rho^{i,l}_{k}(t_j) \left(\tilde y_{j,h}^{i,2}/\tilde \gamma^i_{0}(t_j) + \tilde y_{j,h}^{l,2}/\tilde \gamma^l_{0}(t_j) \right). \label{eq:Xicross}
\end{align}
Define the residuals of the local linear regression \eqref{eq:loclin}  as $\hat{\tilde e}^{i,l}_{j,k} = \tilde y_{j,k}^i  \tilde y_{j, h}^l - \hat \beta_{k}^{i,l}(t_j)$. For the estimation of ${\tilde \Gamma}_{z}^{2}(\cdot)$, we use 
\begin{align}
    \hat{\tilde \Gamma}_{z}^{2}(t) =  \frac{\kappa}{m}\sum_{i=1}^{n}\tilde \Delta^{2}_{z,i}\omega(t,i),\quad \omega(t,i) = K_{\eta}(t-t_i)/\sum_{j=1}^n K_{\eta}(t_j -t),\label{eq:Gammacross}
\end{align}
where $\tilde \Delta_{z,s} = \sum_{j=s}^{s+m -1} \hat \Xi_{z, j}$, and $\hat \Xi_{z, j} = (\hat{\tilde e}^{i,l}_{j,h}/2 - \hat{\tilde e}^{i,l}_{j,k})/\tilde \sigma_{i,l}(t)-4^{-1}\tilde \rho^{i,l}_{k}(t) \big(\hat{\tilde e}^i_{j,h}/\tilde \gamma^i_0(t_j) + \hat{\tilde e}^l_{j,h}/\tilde \gamma^l_0(t_j) \big)$. For $1 \leq j \leq 2\nb$, $\hat{\bar{\mf \Xi}}_j^{\B} $ is the estimator of $\bar{\mf \Xi}_j^{\B}$ (i.e., \eqref{eq:Xidiscrete}) using estimators  $\tilde  \Xi_{z, j}$ and $ \hat{\tilde \Gamma}_{z}^{2}(t)$ for $ \Xi_{z, j}$ and ${\tilde \Gamma}_{z}^{2}(t)$, i.e., 
\begin{align}
   \quad     \hat{\bar {\mf \Xi}}_{j}^{\B} = (\hat{\bar {\mf \Xi}}_{j,\nb}^{\B, \T}, \cdots, \hat{\bar {\mf \Xi}}_{n - 2\nb +j, n - \nb}^{\B, \T})^{\T},~\text{with}~\hat{\bar{\mf \Xi}}_{j,s}^{\B} = (c_z K_{b_z}(t_j - t_s)\tilde \Xi_{z, j}/ \hat{\tilde \Gamma}_z(t_s), z\in \B)^{\T}.
   \label{eq:hatXi}
\end{align}
\subsection{Theoretical properties}
Write $\sqrt{nb_z} |g_{z}(t)-\tilde \rho_z(t)|/\hat{\tilde \Gamma}_z(t) $ as $\tilde T_z(t)$, $t \in [0,1]$. Then, the FWER (conditional on data) can be directly expressed as $1-P_{H_0}(\{\tilde T_z(t)
    \leq \tilde r_{\bt}, z=(i,l,k)\in \B, t\in[b,1-b]\}| \F_n)$, where $H_0$ denotes the null hypotheses of \eqref{eq:test_t}.  The recovery probability of {\it time-varying} networks conditional on data is 
\begin{align}
  \quad \int_0^1 P(\tilde N(t)| \F_n) dt,\label{eq:define_recover}
\end{align}
where $\tilde N(t) = \{\tilde T_z(t)
    \leq \tilde r_{\bt}, z\in \B \cap \bar{\tilde \B}(t)\} \cap  \{\tilde T_z(t)
    > \tilde r_{\bt}, z \in \tilde \B(t)\}$.
Define
$
\tilde{\mathcal C}_n = \{\mf x(\cdot) = (x_z(\cdot), z\in \B)^{\T}\in [b,1-b]^{|\B|}:\sqrt{nb_z} | x_z(t) -\tilde \rho_z(t) | \leq \tilde r_{\bt}   \hat{\tilde \Gamma}_z(t), \forall t\in [b,1-b], z \in \B\}$. Recall that $\alpha \in (0,1)$ is the prespecified significance level.\par The following theorem ensures the FWER control and the recovery probability of \cref{alg:jointnetwork}. The proof is deferred to \cref{sec:proof}.

\begin{theorem}\label{boot:net} 
    Assume the conditions of \cref{nonasynetwork} and 
    \begin{align}
        \vartheta_n^{1/3}\left( 1 \vee \log (n|\B|/\vartheta_n) \right)^{2/3} +\big(g_n (n|\B|)^{1/q}\big)^{q/(q+2)} \to 0.
    \end{align}
(i) (Type I error control.) Conditional on data as $n$ and $B$ go to infinity,  we have
    \begin{align}
       P_{H_0}(\{\tilde T_z(t)
    \leq \tilde r_{\bt},z\in \B, t\in[b,1-b]\}| \F_n) =  P(\{\rho_z(\cdot), z \in \B\} \in  \tilde{\mathcal C}_n | \F_n) \overset{p}{\to} 1-\alpha. \label{eq:lim}
    \end{align}
(ii) (The lower bound of the recovery probability.) Under \cref{ass:alter}, we have $$\underset{n \to \infty}{\lim} \underset{B \to \infty}{\lim} \int_0^{1} P(\tilde N(t)| \F_n) dt\geq 1-\alpha,$$ with probability approaching $1$ for any arbitrarily small $\alpha>0$,
    where $\tilde N(t)$ is the event for correctly recovering the network at time $t$, see \eqref{eq:define_recover}. 
\end{theorem}
\subsection{Implementation}
 The selection of smoothing  parameters for \cref{alg:jointnetwork} follows similarly to that of \cref{alg:jointreduce}. 
In particular, we use the same scheme for the selection of $b_z$, $m_z$,  $z = (i,l,k)\in \B$ and for $w$ and $\eta$, while the sample variance $s^2_{w_{i},\eta_j}$ of \eqref{eq:mv}  should be modified by replacing $\hat{\check {\mf S}}_{l,j}^{\B}$ with the variance-reduced estimator $\hat{\tilde
{\mf  S}}^{\B}_{l,j}$ of the Step 3 of \cref{alg:jointnetwork}.

\section{Proofs}\label{sec:proof}
For the clarity of proof, we may use $(i,l,k)$ for $z$ in the index of correlations, bandwidths and estimators. Recall that $b = \max_{(i,l,k) \in \B} b^{i,l}_k$. For any two random sequences $a_n$ and $b_n$, let $a_n \sim_1 b_n$ denote $\lim_{n \to \infty} a_n/b_n \to 1$ with probability approaching $1$, 
     $a_n \lesssim_1 b_n$ denote $\limsup_{n \to \infty} a_n/b_n  \leq 1$ with probability approaching $1$.  
 We use $a \vee b$ to denote $\max\{a, b\}$. Let $\Theta(a, b) = a\sqrt{1 \vee \log(b/a)}$, where $a,b$ are positive constants. Let  $\TT = [b, 1-b]$.
 
 We present the necessary propositions used in the main paper and the proofs of the theorems, whose proofs can be found in  \cref{sec:aux} of the online supplement. Define the $\FF_n$ measurable event $$B^{\prime}_n = 
    \left\{ \max_{(i,l, k) \in \B} \sup_{t \in \TT} \left|\tilde \gamma^{i,l}_k(t) - \gamma^{i,l}_k(t) \right|  > f_n q_n^{\prime}\right\},$$ where $f_n = (nb)^{-1/2}(|\B|/b)^{1/q}$, $q_n^{\prime}$ is a positive sequence such that $q_n^{\prime} \to \infty$, $f_nq_n^{\prime} \to 0$.
\begin{proposition}[Asymptotic behavior of estimated cross-correlation curves]\label{prop:rhok}
Let $\TT = [b, 1-b]$.
Under Assumptions \ref{Ass:error}, \ref{A:K}, \ref{Ass:ck}, \ref{Ass:diff},  and the
bandwidth conditions $nb^4 \to \infty$, $nb^7 \to 0$, $n^{2\phi - 1}b^{-1}h \to 0$, we have the following results:
\begin{enumerate}[label=(\roman*)]
\item Recall  $\vartheta^{i,l}_k(t) = \frac{1}{nb_{k}^{i,l}}\sum_{j=1}^n K_{b_{k}^{i,l}}(t_i-t)\Xi^{i,l}_{j,k}$, where $\Xi^{i,l}_{j,k}$ is as defined in \eqref{eq:Xidiscrete}. Then, we have 
$$
   \max_{(i,l,k) \in \B} \left\|\sup _{t \in \mathcal{T}}\left| \tilde{\rho}^{i,l}_{k}(t)-\rho^{i,l}_{k}(t)- \vartheta^{i,l}_k(t)\right| \mf 1(\bar B^{\prime}_n) \right\|_q= O\left(n^{\phi-1}b^{-1-1/q}h+ b^{3-1/q} + n^{-1/2}b^{-1/q}h\right).
    $$
\item Recall that $\check \vartheta^{i,l}_k(t) := \frac{1}{nb_{k}^{i,l}}\sum_{j=1}^n \check K_{b_{k}^{i,l}}(t_i-t)\Xi^{i,l}_{j,k}$. Then, we have 
\begin{align}
\max_{(i,l,k) \in \B}  \left\|\sup _{t \in \mathcal{T}}\left|\check{\rho}_{k}^{i,l}(t)-\rho^{i,l}_{k}(t)-\check \vartheta^{i,l}_k(t)\right|\mf 1(\bar B^{\prime}_n)\right\|_q=O\left(n^{\phi-1}b^{-1-1/q}h+ b^{3-1/q}  + n^{-1/2}b^{-1/q}h\right).
\end{align}

\end{enumerate}

\end{proposition}

\begin{proposition}\label{prop:Gammacp}
    Under the condition of \cref{prop:rhok}, assuming bandwidth conditions $\eta \to 0$, $n\eta \to \infty$, $m \to \infty$,  $m = O(n^{1/3})$, $|\B|^{1/q}(mb/n)^{-1/(2q)} \sqrt{m/(nb)}\to 0$, $nb^3 \to \infty$, $nb^6 \to 0$. Suppose for $(i,l, k) \in \B$, the twice derivatives of $\tilde \Gamma^{i,l,2}_{k}(t)$ are uniformly bounded on $(0,1)$. Then, we have 
    \begin{align}
      \left\|  \max_{(i,l,k) \in \B} \sup_{t \in \TT} \left|\hat{\tilde \Gamma}^{i,l,2}_{k}(t) - \tilde \Gamma^{i,l,2}_{k}(t) \right|  \right\|_q = O\left(|\B|^{1/q} g_n^{\prime}\right),
    \end{align}
    where $g_n^{\prime} =\sqrt{m/(n \eta^{2})}+1/m+\eta  +\sqrt{m/(nb)}(mb/n)^{-1/(2q)}$.
\end{proposition}

\subsection{Proof of \texorpdfstring{\cref{boot:net}}{Theorem B.1}}
 Let $\tilde Z_{\bt}$ denote $\tilde Z^{(r)}_{\bt}$ in one iteration of \cref{alg:jointnetwork}. By \cref{nonasynetwork}, it's sufficient to show that 
    \begin{align}
        \sup_{x \in \mathbb R} \left| P(\tilde Z_{\bt} \leq x| \F_n) - P\left(\left| \frac{1}{\sqrt{n b}}\sum_{i=1}^{2 \nb} \tilde{\mf Z}_i \right|_{\infty} \leq x\right)\right|  =\op(1),\label{eq:tildeZcompare}
    \end{align}
    where $(\tilde{\mf Z}_i)_{i=1}^{2\nb} \in \R^{(n - 2\nb+1)|\B|}$ is a sequence of zero-mean Gaussian vectors which share the same autocovariance structure with the vectors $(\bar{\mf \Xi}^{\B}_i)_{i=1}^{2\nb} $ as defined in \cref{nonasynetwork}.
    In the following proof, we omit the index $\B$ in $\hat{\tilde{\mf S}}^{\B}_{l,j}$ for simplicity. Define 
    \begin{align}
        \tilde{\mf Z}^{\diamond}_{a|\B| + c} = \left( (2w\nb)^{-1/2}\sum_{j = w}^{2\nb - w} \hat{\tilde S}_{a-1, j, c} R_{a+j-1}, a=1, \cdots, n-2\nb +1, 1 \leq c \leq |\B| \right),
    \end{align}
    where $\hat{\tilde S}_{l, j, r}$ denote the $r$th element of $\hat {\tilde{\mf S}}_{l, j}$. Let
    \begin{align}
        \tilde{\mf Z}^{\diamond} = \left( \tilde{\mf Z}_1^{\diamond, \T},\cdots, \tilde{\mf Z}_{(n-2\nb +1)|\B|}^{\diamond, \T}\right)^{\T},
    \end{align}
    and it follows that $\tilde Z_{\bt} = |\tilde{\mf Z}^{\diamond}|_{\infty}$. 
    Given the data, the conditional variance  for $k_1 \leq k_2$ is
    \begin{align}
        \sigma^{\tilde{\mf Z}^{\diamond}}_{k_1,k_2, j_1, j_2} &: = (2w\nb )^{-1}\E\left( \sum_{j=w}^{2\nb - w } \hat{\tilde S}_{k_1 -1, j, j_1} R_{k_1-1+j}  \sum_{j=w}^{2\nb - w }\hat{\tilde S}_{k_2 -1, j, j_2} R_{k_2-1+j}|\F_n\right)\\ 
        &= \sum_{j=w}^{2\nb - w -(k_2 - k_1)-1} \hat{\tilde S}_{k_1-1, j+k_2-k_1,j_1}\hat{\tilde S}_{k_2 - 1, j, j_2}/(2w\nb ).
    \end{align}
    Define $\bar \Xi_{i,s,j}^{\B}$ as the $j$th element of $\mf{\bar \Xi}_{i,s}^{\B}$. For simplicity, we omit the index $\B$ in $\bar \Xi_{i,s,j}^{\B}$.
    Define 
    \begin{align}
       \tilde S_{a-1, j, c } = \sum_{i=j-w+1}^{j} \bar{\Xi}_{i + (a-1), \nb + (a-1), c} - \sum_{i=j+1}^{j+w}\bar{\Xi}_{i + (a-1), \nb + (a-1), c},
    \end{align}
    and $\tilde{\mf Z}^{\dagger}$ by substituting $\hat{\tilde {S}}_{(a-1), j, c}$ in $\tilde{\mf Z}^{\diamond}$ by $\tilde{S}_{(a-1), j, c}$.
    \begin{align}
     \sigma^{\tilde{\mf Z}^{\dagger}}_{k_1,k_2, j_1, j_2} &:= (2w\nb )^{-1} \E\left( \sum_{j=w}^{2\nb - w } \tilde S_{k_1 -1, j, j_1} R_{k_1-1+j}  \sum_{j=w}^{2\nb - w } \tilde S_{k_2 -1, j, j_2} R_{k_2-1+j}|\F_n\right)\\ 
       &= \sum_{j=w}^{2\nb - w -(k_2 - k_1)} \tilde S_{k_1 - 1, j+k_2-k_1, j_1} \tilde S_{k_2 - 1, j, j_2}/(2w\nb).
    \end{align}
    Let $\tilde Z_{i,j}$ denote the $j$th element of  $\tilde{\mf Z}_{i}$ . Since $\tilde{\mf Z}_{i}$ has the same covariance structure as $\bar{\mf \Xi}_i^{\B}$, the covariance structure of $\frac{1}{\nb}\sum_{i=1}^{2\nb} \tilde{\mf Z}_{i}$ is 
    \begin{align}
        \sigma^{\tilde{\mf Z}}_{k_1, k_2, j_1, j_2} 
        &:= \E\left(\frac{1}{\nb}\sum_{i=1}^{2\nb }\tilde Z_{i, (k_1- 1)|\B| +j_1 }\sum_{i=1}^{2\nb }\tilde Z_{i, (k_2- 1)|\B| +j_2 }\right) \\ 
        & =\E \left(\frac{1}{\nb} \sum_{i=1}^{2\nb }\bar \Xi_{i+(k_1- 1), \nb +(k_1 -1),j_1}\sum_{i=1}^{2\nb }\bar \Xi_{i+(k_2- 1), \nb +(k_2 -1),j_2}\right).
    \end{align}
    By (5.12)  of \cite{Dette2021ConfidenceSF}, since $(\bar \Xi^{i,l}_{j, k})_{j=1}^n$ are LS and satisfy  geometric metric contraction due to \cref{Ass:error}, we have 
    \begin{align}
      \left\| \max_{k_1, k_2, j_1, j_2}|\sigma^{\tilde {\mf Z}^{\dagger}}_{k_1, k_2, j_1, j_2} - \sigma^{\tilde{\mf Z}}_{k_1, k_2, j_1, j_2}  |\right\|_{q/2} = O\left(\frac{\log^2 n}{w} + \frac{w}{nb} + \sqrt{\frac{w}{nb}}(n\log n)^{4/q}\right).
    \end{align}
    By (5.12) and (5.13) of \cite{Dette2021ConfidenceSF}, we have 
    \begin{align}
      \sup_{x \in \R}\left|P\left(\left. |\tilde{\mf Z}^{\dagger}|\leq x \right|\F_n\right) - P\left(\left| \frac{1}{\sqrt{n b}}\sum_{i=1}^{2 \nb} \tilde{\mf Z}_i \right|_{\infty} \leq x\right)\right| = \Op(\vartheta_n^{1/3}\left\{ 1 \vee \log (n|\B|/\vartheta_n) \right\}^{2/3}),
      \label{eq:dagger}
    \end{align}
    where $\vartheta_n = \frac{\log^2 n}{w} + \frac{w}{nb} + \sqrt{\frac{w}{nb}}(n\log n)^{4/q}$. 
    Let $g^{\prime}_n =|\B|^{1/q} (\sqrt{m/(n \eta^{2})}+1/m+\eta  +\sqrt{m/(nb)}(mb/n)^{-1/2q})$, $f_n = (nb)^{-1/2} b^{-1/q} |\B|^{1/q}$. Recall the estimators $\tilde \rho^{i,l}_k(\cdot)$ and $\hat{\tilde  \Gamma}^{i,l,2}_{k}(\cdot)$ as defined in \eqref{eq:rho_network} and \eqref{eq:Gammacross}, respectively.
    Define the $\FF_n$ measurable events
    \begin{align}
        A_n = \left\{ \max_{(i,l, k) \in \B} \sup_{t \in \TT} \left|\hat{\tilde  \Gamma}^{i,l,2}_{k}(t) - \tilde \Gamma^{i,l,2}_{k}(t) \right|  > g^{\prime}_n q_n\right\}, \quad B_n = 
    \left\{ \max_{(i,l, k) \in \B} \sup_{t \in \TT} \left|\tilde \beta^{i,l}_k(t) - \beta^{i,l}_k(t) \right|  > f_n q_n\right\},\label{eq:ABset}
    \end{align}
    where the positive sequence $q_n$ goes to infinity such that $(g_n^{\prime}+ f_n)q_n \to 0$, $\TT = [b, 1-b]$. Then by \cref{lm:loclin} and \cref{prop:Gammacp}, $P(A_n \cup B_n) = O(q_n^{-q})$.
    Then, for some sufficiently large constant $M$ and due to the conditional normality, we have
    \begin{align}
        \E(|\tilde{\mf Z}^{\diamond} - \tilde{\mf Z}^{\dagger}|^q_{\infty}\mf 1(\bar A_n \cap \bar B_n) |\FF_n) \leq M\left| \frac{\log (n |\B|)}{2w\nb}\max_{\substack{1 \leq r \leq |\B|\\0 \leq l \leq  n-2\nb }}\sum_{j=w}^{2\nb - w}(\hat{\tilde S}_{l,j,r} - \tilde S_{l,j,r})^2 \mf 1(\bar A_n \cap \bar B_n) \right|^{q/2}.\label{eq:square0}
    \end{align}
    
    Recall the definition of  $\hat{\bar{\mf \Xi}}_{i + l, \nb +l}^{\B}$ in \eqref{eq:hatXi}.
 Let $\hat{\bar{\Xi}}_{i + l, \nb +l, r}$ denote the $r$th element of $\hat{\bar{\mf \Xi}}_{i + l, \nb +l}^{\B}$.  Define  $\hat{\tilde S}^{\epsilon}_{a,j,r}$ by substituting $\tilde y^{i}_{j,k}$ in $\hat{\tilde S}_{a,j,r}$ with $\tilde \epsilon^{i}_{j,k}$,  
  $\hat{\tilde S}^{e}_{a,j,r}$ by substituting $\tilde \epsilon^{i}_{j,k}\tilde \epsilon^{l}_{j,h}$ in $\hat{\tilde S}^{\epsilon}_{a,j,r}$ with  $\tilde e^{i,l}_{j,k}$,
   $\hat{\tilde S}^{e, \Gamma}_{a,j,r}$ by substituting $\hat{\tilde \Gamma}^{i,l}_k(t)$ in $\hat{\tilde S}^{e}_{a,j,r}$ with ${\tilde \Gamma}^{i,l}_k(t)$ , and  $\hat{\tilde S}^{e, \Gamma, \rho}_{a,j,r}$ by substituting $\tilde \rho^{i,l}_k(t)$ in $\hat{\tilde S}^{e, \Gamma}_{a,j,r}$ with $ \rho^{i}_{j,k}(t)$. 
    Note that $\hat{\tilde S}^{e, \Gamma}_{a,j,r} = \tilde S_{l,j,r}$. 
    By the continuity of $\beta^{i,l}_k(\cdot)$, via the proof of  \cref{prop:rhok}
    , we have 
    \begin{align}
        &\frac{1}{\sqrt{2w\nb}}\left\| \max_{\substack{1 \leq r \leq |\B|\\0 \leq l \leq  n-2\nb }}\left(\sum_{j=w}^{2\nb - w}(\hat{\tilde S}_{l,j,r} - \tilde S_{l,j,r})^2 \mf 1(\bar A_n \cap \bar B_n)\right)^{1/2}  \right\|_q \\ 
          &\leq  \frac{C(n|\B|)^{1/q}}{\sqrt{2w\nb}}\max_{\substack{1 \leq r \leq |\B|\\0 \leq l \leq  n-2\nb }}\left|\sum_{j=w}^{2\nb - w}\left( \|\hat{\tilde S}_{l,j,r} -  \hat{\tilde S}^{\epsilon}_{l,j,r})\|_q^2 +  \|\hat{\tilde S}^{\epsilon}_{l,j,r} - \hat{\tilde S}^{e}_{l,j,r} \|_q^2 \right. \right. \\
           &\left. \left. + \|\hat{\tilde S}^{e}_{l,j,r} - \hat{\tilde S}^{e, \Gamma}_{l,j,r}\|_q^2 + \|\hat{\tilde S}^{e, \Gamma}  - \tilde S_{l,j,r}\|_q^2 \right) \mf 1(\bar A_n \cap \bar B_n)  \right|^{1/2} \\ 
        &= O((n^{\phi}h/\sqrt{w} +w^{3/2}/n  +  g^{\prime}_nq_n + f_nq_n) (n|\B|)^{1/q}),\label{eq:square}
    \end{align}
    where $C$ is a sufficiently large constant, and the last inequality is due to triangle inequality.
    Let $q_n = \left\{\left(g^{\prime}_n + f_n +  w^{3/2}/n +  n^{\phi} h/\sqrt{w}\right) (n|\B|)^{1/q}\right\}^{-1/(q+2)}$. Under \eqref{eq:rate}, we have $q_n \to \infty$, and $(g_n^{\prime} + f_n) q_n \leq g_n q_n \to 0$. By \eqref{eq:dagger}, \eqref{eq:square0} and \eqref{eq:square}, Lemma C.1 of \cite{Dette2021ConfidenceSF}, we have shown \eqref{eq:tildeZcompare}, i.e.,
    \begin{align}
        &\sup_{x \in \R} \left| P(\tilde Z_{\bt} \leq x| \F_n) - P\left(\left| \frac{1}{\sqrt{n b}}\sum_{i=1}^{2 \nb -1} \tilde{\mf Z}_i \right|_{\infty} \leq x\right)\right|\\
        & =\Op \left(\vartheta_n^{1/3}\left\{ 1 \vee \log (n|\B|/\vartheta_n) \right\}^{2/3} +\Theta\left(q_n^{-q}, n |\mathbb B |\right)\right).
    \end{align}
Under \cref{ass:alter}, for $t \in [b, 1-b]$, $z = (i,l,k) \in \tilde \B(t)$, by the concentration inequality of high-dimensional Gaussian process, we have uniformly 
\begin{align}
    \sqrt{nb_z}|\tilde \rho_z(t) - g_z(t)|/\hat{\tilde \Gamma}_z(t) &\geq  \sqrt{nb_z}|g_z(t) - \rho_z(t)|/\hat{\tilde \Gamma}_z(t) -  \max_{z\in \B} \sup_{t\in \TT}\sqrt{nb_z}|\tilde \rho_z(t) - \rho_z(t)|/\hat{\tilde \Gamma}_z(t) \\
    & \geq \lambda_n \log(n|\B|) - \Op(\log(n|\B|)) \to \infty.\label{eq:fixalter}
\end{align}
 Combining  \eqref{eq:lim} and \eqref{eq:fixalter}, under \cref{ass:alter}, it follows that $\underset{n \to \infty}{\lim} \underset{B \to \infty} {\lim} \int_0^1 P(\{\tilde T_z(t)
    \leq \tilde r_{\bt},z\in \B \cap 
    \bar{\tilde \B}(t)\} \cap  \{\tilde T_z(t)
    > \tilde r_{\bt}, z \in \tilde \B(t)\}| \F_n) dt \geq \underset{n \to \infty}{\lim} \underset{B \to \infty} {\lim}  P(\{\rho_z(\cdot), z \in \B \} \in \tilde{\mathcal C}_n|\FF_n) =1 - \alpha$, with probability approaching $1$. \hfill $\Box$


\subsection{Proof of \texorpdfstring{\cref{boot:reduce}}{Theorem 4.1}}
 Proof of (i). The result follows from \cref{prop:rhok} and similar arguments  in the proof of \cref{boot:net}.\\
Proof of (ii).  After a further investigation of \eqref{eq:fixalter},  we have
$
   \left|\sqrt{nb_z}|\check \rho_z(t) - g_z(t)|- \sqrt{nb_z}|\tilde \rho_z(t) - g_z(t)|\right| = \Op(\log(n|\B|))=\op(\lambda_n)$. Therefore, if \eqref{eq:uniformreduce} holds, \eqref{eq:recovery} will follow from \cref{boot:net} and elementary calculation for a sufficiently small $\alpha$. 
In the following we shall prove \eqref{eq:uniformreduce}. By \cref{prop:Gammacp}, we have \begin{align}
    \max_{z \in \B} \sup_{t\in\TT} \left|\hat{\check \Gamma}_z(t) /\hat{\tilde \Gamma}_z(t) - {\check \Gamma}_z(t) /\tilde \Gamma_z(t)\right| = \op(1).
\end{align}
Write $\Xi_{j,k}^{i,l} = \tilde L^{i,l}_k(t_j , \FF_j)$, where $ \tilde L^{i,l}_k(\cdot ,\cdot)$ is a nonlinear filter, $(i,l,k) \in \B$. By a careful investigation of Lemma C.3 of \cite{dette2022}, we have $\tilde \Gamma_{k}^{i,l,2}(t) =   \kappa \sigma^2(\tilde L^{i,l}_{k}, t)$ and $\check \Gamma_{k}^{i,l,2} (t)=\check \kappa \sigma^2(\tilde L^{i,l}_{k}, t)$, where $\sigma^2(\tilde L^{i,l}_{k}, t)$ is the long-run variance of $\tilde L^{i,l}_k(t , \FF_j)$.  Recall that $\check{K}(t) = (\check{K}_{+}(t) + \check{K}_{-}(t))/2$ where $\check{K}_{\pm}(t)=\sum_{j=0, 1, 2}A_{j}(\pm r)K (t+(\pm r + 1 - j)\delta)$,  $A_{0}(r)=r(r-1) / 2, A_{1}(r)=\left(1-r^{2}\right),  A_{2}(r)=r(r+1) / 2$ for some selected $r \in (-1,1)$ and non-negative constant $\delta$. It follows from the elementary calculation that 
   \begin{align}
   \int \check K^2(t) dt =  \int K^2(t) dt  - r^2(1-r^2)C( \delta) - \frac{1}{2}D(\delta),\label{eq:defcd}
   \end{align}
   where 
$         C(\delta) = 1.5 C(0, \delta)-2 C(0.5, \delta)+0.5 C(1, \delta)$,
         $C(a, \delta)= \frac{1}{2}\int K(u-a \delta) K(u+a \delta) + K(u-a \delta) K(u+a \delta) d u$ and 
         $D(\delta) = \kappa - r^2(1-r^2)C(\delta)- \frac{1}{4}\{2(2-3r^2)(1 - r^2)C(r, \delta)  + r^2(1-r)^2C(r+1, \delta) + r^2(r+1)^2C( r-1, \delta) - 4r(1-r)(1-r^2)C(r+1/2, \delta)+ 4r(1+r)(1-r^2)C(r-1/2, \delta) \}.
$
   The positivity of $ C(\delta) $ and $D(\delta)$ when $\delta > 0$ follows from \cref{A:K} and a careful investigation of  Proposition 1 and Proposition 2 of \cite{cheng2007reducing}.

     For a pre-specified significance level $0 < \alpha < 1$, we shall show $|\tilde r_{\bt}/ \check r_{\bt} - 1| = \op(1) $ by proving
  \begin{align}
   |\tilde r_{\bt} -\sqrt{2\log b^{-1} }| = \op(1),\quad     |\check r_{\bt} -  \sqrt{2\log b^{-1}}| = \op(1).\label{eq:rsim}
   \end{align}
   
   We only give the proof for $ \tilde r_{\bt} $ for brevity, and that of $\check r_{\bt}$ will follow similarly. Write $\tilde r$ short for $\tilde r_{\bt}$. We shall first show that uniformly for $x > \sqrt{2\log b^{-1}-2c}$ for some arbitrary non-negative $c$, 
   \begin{align}
   P(\tilde Z_{\bt} > x |\FF_n) &\geq  \prod_{z \in \B} (b^{-1} c_z^2 C_{\kappa} \exp(-x^2/2)  + 2(1 - \Phi(x))) + \op(1)\\
   & \geq  (b^{-1} C_{\kappa} \exp(-x^2/2)  + 2(1 - \Phi(x)))^{|\B|}+ \op(1).\label{eq:approx_tildeZ}
   \end{align}
where $C_{\kappa} =  \sqrt{\int |K^{\prime}(x)|^2 dx /\int K^2(x) dx}$. Recall that $P(\tilde Z_{\bt}>\tilde r|\FF_n)=1-\alpha$. Hence, \eqref{eq:approx_tildeZ} will lead to $\tilde r^- :=\sqrt{2\log (b^{-1} C_{\kappa})- 2|\B|^{-1}\log \alpha }\lesssim_1 \tilde r $. Recall the definition of $c^*$ in \cref{Ass:ck}. On the other hand, we shall show that under \cref{Ass:ck}, for any $x$ such that $x / \sqrt{2\log \log n} \to \infty$,
\begin{align}
P(\tilde Z_{\bt} > x |\FF_n) &\leq  1- \prod_{z \in \B} (1- b^{-1} c_z^2 C_{\kappa} \exp(-x^2/2)  - 2(1 - \Phi(x))) + \op(1),\\ 
& \leq 1-(1- b^{-1}(c^*)^{-1} C_{\kappa} \exp(-x^2/2)  - 2(1 - \Phi(x)))^{|\B|}+ \op(1),
\label{eq:approx_tildeZup}
\end{align}
which will lead to $ \tilde r \lesssim_1\sqrt{2\log (b^{-1} C_{\kappa}/c^*) - 2\log(1-(1-\alpha)^{1/|\B|})}=: \tilde r^{+}$. Since $|\B|^{-1}\log \alpha  = o(b^{-1})$,  $ - \log(1-(1-\alpha)^{1/|\B|}) = O(\log |\B|) = o(b^{-1})$, we have $\tilde r^- \sim_1 \sqrt{2\log b^{-1} }$ and $\tilde r^{+} \sim_1 \sqrt{2\log b^{-1}}$, which yields the result $|\tilde r  - \sqrt{2\log b^{-1}}|=\op(1)$. Finally, after a further investigation of \eqref{eq:fixalter}, we have
$
   \left|\sqrt{nb_z}|\check \rho_z(t) - g_z(t)|- \sqrt{nb_z}|\tilde \rho_z(t) - g_z(t)|\right| = \Op(\log(n|\B|))=\op(\lambda_n \log(n|\B|)).$ Since we have shown that $\hat{\check \Gamma}_{z}(t)/\hat{\tilde \Gamma}_{z}(t) \overset{p}{\to} \sqrt{\check 
 \kappa/\kappa} < 1$ and $|\tilde r_{\bt}/\check r_{\bt} - 1| = \op(1)$,
\eqref{eq:recovery} then follows from \eqref{eq:lim} and result (i). 

\textbf{Proof of \eqref{eq:approx_tildeZ} and \eqref{eq:approx_tildeZup}.} By \cref{boot:net}, the tail probabilities of $\tilde Z_{\bt}$ and $\check Z_{\bt}$ given data are asymptotically  equal to 
  \begin{align}
  P\left(\left| \frac{1}{\sqrt{n b}}\sum_{i=1}^{2 \nb} \tilde{\mf Z}_i \right|_{\infty} > x\right) ~\text{and}~  P\left(\left| \frac{1}{\sqrt{n b}}\sum_{i=1}^{2 \nb} \check{\mf Z}_i \right|_{\infty} >  x\right), \label{eq:compare}
  \end{align}
uniformly over $x\in \mathbb R$  with probability approaching $1$, see \eqref{eq:tildeZcompare}.
In the subsequent proof, we will derive explicit approximation formulae for the tail probabilities defined in \eqref{eq:compare}. For brevity, we only present the details for the computation of the lower and upper bounds of $ P\left(\left| \frac{1}{\sqrt{n b}}\sum_{i=1}^{2 \nb} \tilde{\mf Z}_i \right|_{\infty} > x\right)$, and the results for the variance reduced tail probability (the second term) in \eqref{eq:compare} follow  analogously.  Recall $\tilde Z_{i,j}$ is the $j$ th element in $\tilde{\mf Z}_i$. Define $\bar{\tilde Z}_{i,j}$ as  $\tilde Z_{i,j}/\sqrt{\mathrm{Var}(\sum_{i=1}^{2\nb} \tilde Z_{i,j})/(nb)}.$ By elementary calculation similar to Lemma C.3 in \cite{Dette2021ConfidenceSF}, $\mathrm{Var}(\sum_{i=1}^{2\nb} \tilde Z_{i,j})/(nb)$ is bounded from $0$,  so that $ \frac{1}{\sqrt{n b}}\sum_{i=1}^{2\nb}\bar{\tilde {\mf Z}}_{i}$ is the normalized version of  $ \frac{1}{\sqrt{n b}}\sum_{i=1}^{2 \nb}{\tilde {\mf Z}}_{i}$ and each component is of variance $1$.  \par
\textbf{Proof of \eqref{eq:approx_tildeZ}}.
We shall break the proof into the following 3 steps.  Fix the order of the elements of $\mathbb B$, and let $z$ be the $j_{th}$ element of $\mathbb B$ in the following proofs. \par
  \textbf{Step 1}: Under \cref{ass:rho}, when $|\B| = O(\log n)$, we shall show that for $x > \sqrt{2\log b^{-1}-2c}$, where $c$ is an arbitrary non-negative constant, 
  \begin{align}
   P\left(\left| \frac{1}{\sqrt{n b}}\sum_{i=1}^{2 \nb} \bar{\tilde{\mf Z}}_i \right|_{\infty} \leq x\right) \leq \prod_{j=1}^{|\B|}P\left(\max_{0 \leq s \leq m^*} \left|\frac{1}{\sqrt{n b}}\sum_{i=1}^{2 \nb} \bar{\tilde{Z}}_{i, s\nb|\B| + j} \right| \leq x \right) + o(1),\label{eq:step1}
  \end{align}
     where $m^* = (n - 2\nb )/(\nb)$.
  \par
  \textbf{Step 2}: For $i.i.d.$ $N(0,1)$ random variables $v_{z, j}$, $z \in \B$, $1 \leq j \leq n$, for $b \to 0$, $nb \to \infty$, we obtain
  \begin{align}
  &\sup_{x \in \R} \left| \prod_{j=1}^{|\B|}P\left(\max_{0 \leq s \leq m^*} \left|\frac{1}{\sqrt{n b}}\sum_{i=1}^{2 \nb} \bar{\tilde{Z}}_{i, s\nb|\B| + j}\right| \leq x \right) \right. \\ & \left.- \prod_{z \in \B} P\left(\max_{0 \leq s \leq m^*} \left|\sum_{i=1}^{2 \nb} \frac{v_{z, i+s\nb} K_{b_z}(\frac{i-\nb}{n})}{\sqrt{nb_z \kappa}} \right| \leq x \right)\right| = o(1).  \label{eq:step2} \end{align}\par
   \textbf{Step 3}: Show that for $x > \sqrt{2\log b^{-1}-2c}$, for the constant $c$ defined in Step 1
   \begin{align}
  &\prod_{z \in \B} P\left(\max_{0 \leq s \leq m^*} \left|\sum_{i=1}^{2 \nb} \frac{v_{z, i+s\nb} K_{b_z}(\frac{i-\nb}{n})}{\sqrt{nb_z \kappa}} \right| > x \right)\\ & = \prod_{z \in \B} (b^{-1} c_z^2 C_{\kappa} \exp(-x^2/2)  + 2(1 - \Phi(x))) + o(1),\label{eq:step3} 
   \end{align}
   where $\Phi(x)$ is the cumulative function of normal distribution, $C_{\kappa} = \sqrt{\int |K^{\prime}(x)|^2 dx /\int K^2(x) dx}$. Therefore, since $c_z \geq 1$, \eqref{eq:approx_tildeZ} follows from \eqref{eq:step3}.
 
   \begin{remark}
   $|\B|$ can also be allowed to diverge at a polynomial rate, such that $|\B| =o\left( b^{-\frac{1-\rho}{4(1+\rho)}}\right)$. 
     We omit its proof for the sake of conciseness.
   \end{remark}

\textbf{Proof of Step 1}. When $|\B|$ is finite, the proof follows from Theorem 3 of  \cite{HUSLER198891}. We extend it for diverging $|\B|$.
Note that 
\begin{align}
P\left(\left| \frac{1}{\sqrt{n b}}\sum_{i=1}^{2 \nb} \bar{\tilde{\mf Z}}_i \right|_{\infty} \leq x\right)  \leq P\left(\max_{1 \leq j \leq |\B|, 0 \leq s \leq m^*} \left|\frac{1}{\sqrt{n b}}\sum_{i=1}^{2 \nb} \bar{\tilde{Z}}_{i, s\nb|\B| + j} \right| \leq x \right).\label{eq:lower}
\end{align}For the sake of simplicity, let 
\begin{align}
J_{k,j} := \frac{1}{\sqrt{n b}}\sum_{i=1}^{2 \nb} \bar{\tilde{Z}}_{i, (k-1)|\B| + j}, \quad k=1,\cdots, n-2\nb +1, j = 1,\cdots, |\B|.
\end{align}

It suffices to show that the distribution of the Gaussian process of $(J_{s\nb+1,j})$ can be approximated by $(J_{s\nb+1,j}^*)$ which is independent for different $j$'s (i.e., $ \mathrm{Cov}(J^*_{s\nb+1,i},J^*_{l\nb+1,j}) = 0$ for any $s,l$ and $i\neq j$) and such that 
\begin{align}
   \mathrm{Cov}(J^*_{s\nb+1,j},J^*_{l\nb+1,j}) = \mathrm{Cov}(J_{s\nb+1,j}, J_{l\nb+1,j}),\quad s,l=0,\cdots, m^*,j=1,\dots, |\B|. \label{eq:Jstardef}
\end{align}
Therefore, by definition of \eqref{eq:Jstardef} we have
\begin{align}
    \prod_{j=1}^{|\B|}P\left(\max_{0 \leq s \leq m^*} \left|\frac{1}{\sqrt{n b}}\sum_{i=1}^{2 \nb} \bar{\tilde{Z}}_{i, s\nb|\B| + j} \right| \leq x \right) = P(J^*_{s\nb+1,j} \leq x,  0 \leq s \leq m^*, 1 \leq  j \leq |\B|).
\end{align}

 By the Theorem 3 of \cite{HUSLER198891}, it follows that
\begin{align}
 &\left|P\left(\max_{1 \leq j \leq |\B|, 0 \leq s \leq m^*} \left|\frac{1}{\sqrt{n b}}\sum_{i=1}^{2 \nb} \bar{\tilde{Z}}_{i, s\nb|\B| + j} \right| \leq x \right) - \prod_{j=1}^{|\B|}P\left(\max_{0 \leq s \leq m^*} \left|\frac{1}{\sqrt{n b}}\sum_{i=1}^{2 \nb} \bar{\tilde{Z}}_{i, s\nb|\B| + j} \right| \leq x \right)\right|  \\
 & = \left| P(J_{s\nb+1,j} < x, 0 \leq s \leq m^*, 1 \leq  j \leq |\B|) -   P(J^*_{s\nb+1,j} < x,  0 \leq s \leq m^*, 1 \leq  j \leq |\B|) \right|\\ 
 & \leq \sum_{1 \leq i,j \leq |\B|} \sum_{0 \leq s,l \leq m^*} \lambda^*_{i,j}(s,l) (1 - \lambda^{*,2}_{i,j}(s,l))^{-1/2}\exp\left\{-\frac{x^2}{1 + \lambda^*_{i,j}(s,l)} \right\}:=S_n,\label{eq:Sn}
\end{align}
where $\lambda^*_{i,j}(s,l) = |\mathrm{Cov}(J_{s\nb+1,i}, J_{l\nb+1,j})|$ if $i \neq j$, $s, l=0,\cdots, m^*$, and  $\lambda^*_{i,j}(s,l) = 0$ otherwise. Since  $(\tilde{\mf Z}_i)_{i=1}^{2\nb} \in \R^{(n - 2\nb+1)|\B|}$ share the same autocovariance structure with  $(\bar{\mf \Xi}^{\B}_i)_{i=1}^{2\nb} $ and $\mathrm{Var}(\lambda^*_{i,j}(s,l))=1$,  
by \cref{ass:rho}, there exists a constant $\rho$, $0< \rho < 1$, such that $| \lambda^*_{i,j}(s,l) | <  \rho$ for $i,j = 1,\cdots, |\B|$.
Let $\theta_n = b^{-\lambda}$, $ \lambda = (1-\rho)/(2(1 + \rho))$.
Note that by Lemma 5 of \cite{zhou2010simultaneous}, since $\mathrm{Var}(\sum_{i=1}^{2\nb} \tilde Z_{i,j})/(nb)$ and $\tilde \Gamma_z(\cdot)$ are bounded from $0$, for $|k_1-k_2| \geq \lfloor nb \log n  \rfloor$, $k_1, k_2 = 1,\cdots, n-2\nb +1$, it follows that
\begin{align}
\mathrm{Cov}(J_{k_1,i}, J_{k_2,j}) &= O\left(\frac{c_zc_y}{nb}\sum_{p, q = 1}^{2\nb} E\left(\Xi_{y, p+k_1-1}\Xi_{z, q+k_2-1}K_{b_y}\left(\frac{p-\nb}{n}\right)K_{b_z}\left(\frac{q-\nb}{n}\right)\right)\right)\\
&=O\left(\frac{1}{nb}\sum_{p, q = 1}^{2\nb}  \chi^{|p+k_1-q-k_2|}\right) = O(\chi^{|k_1-k_2|}/(nb)),
\end{align}
where $z$ and $y$ are the $j$ and $i$ th elements of $\B$, respectively.
For  $|k_1-k_2| \leq \lfloor nb \log n  \rfloor$ it follows immediately that  $\mathrm{Cov}(J_{k_1,i}, J_{k_2,j})  = O(1)$.
We consider $S_n$ separately for $|s-l| < \theta_n$ and $|s-l| \geq \theta_n$, $s, l = 0,\cdots,m^*$,
\begin{align}
S_n& = \left(\sum_{1 \leq i, j\leq |\B|} \sum_{|s-l| < \theta_n} +  \sum_{1 \leq i, j \leq |\B|} \sum_{|s-l| \geq \theta_n} \right) \lambda^*_{i,j}(s,l) (1 - \lambda^{*,2}_{i,j}(s,l))^{-1/2} \exp\left\{-\frac{x^2}{1 + \lambda^*_{i,j}(s,l)} \right\}\\ 
& :=  S_{n,1} + S_{n,2},
\end{align}
where $S_{n,1}$ and $S_{n,2}$ are defined in the obvious way.
Notice that $| \lambda^*_{i,j}(s,l) | <  \rho$ and $m^*=O(b^{-1})$. For a sufficiently large constant $C > 0$ and $x > \sqrt{2 \log b^{-1}-2c}$ for the constant $c$ defined in Step 1, we have 
\begin{align}
S_{n,1} &\leq C |\B|^2 b^{-1} \theta_n \exp \left\{ -\frac{x^2}{2}\right\} \exp \left\{ -\frac{(1- \rho) x^2}{2(1+ \rho)}\right\}\\ &\leq C |\B|^2  \theta_n b^{(1- \rho)/(1+ \rho)}\exp(2c/(1+\rho))  = o(1).\label{eq:Sn1}
\end{align}
On the other hand, for $S_{n,2}$, we obtain
\begin{align}
S_{n,2} = O\left( |\B|^2 b^{-1}  \chi^{nb\theta_n}/(nb) \exp \left\{ -\frac{x^2}{2}\right\}\right) = o(1).\label{eq:Sn2}
\end{align}
Finally, combining \eqref{eq:lower}, \eqref{eq:Sn}, \eqref{eq:Sn1} and \eqref{eq:Sn2}, we have shown \eqref{eq:step1}. 
\par
\textbf{Proof of Step 2}. Write  $J^{\circ}_{k,j} := \sum_{i=1}^{2 \nb} v_{z, i+k-1} K_{b_z}(\frac{i-\nb}{n})/\sqrt{nb_z \kappa}$ in \eqref{eq:step2}, where $z$ is the $j_{th}$  element of $\B$.
By Lemma 3.1 of \cite{victor2013}, we have 
\begin{align}
&\sup_{x \in \R} \left| \prod_{j=1}^{|\B|} P\left(\max_{0 \leq s \leq m^*} \left|\frac{1}{\sqrt{n b}}\sum_{i=1}^{2\nb} \bar{\tilde{Z}}_{s\nb|\B| + j}\right| \leq x \right) \right. \\ & \left.- \prod_{z \in \B} P\left(\max_{0 \leq s \leq m^*} \left|\sum_{i=1}^{2\nb} \frac{v_{z, i+s\nb} K_{b_z}(\frac{i-\nb}{n})}{\sqrt{nb_z \kappa}} \right| \leq x \right)\right| \\ 
& = \sup_{x \in \R} \left| P\left(\cap_{j=1}^{|\B|}\left\{\max_{0 \leq s \leq m^*} |J_{s\nb+1,j}^{*} |\leq x \right\} \right) - P\left(\cap_{j=1}^{|\B|} \left\{ \max_{0 \leq s \leq m^*} |J_{s\nb+1,j}^{\circ}| \leq x\right\}\right) \right|\\ 
& = \sup_{x \in \R} \left| P\left(\max_{0 \leq s \leq m^*, 1 \leq j \leq |\B|} |J_{s\nb+1,j}^{*} |\leq x \right) - P\left(\max_{0 \leq s \leq m^*, 1 \leq j \leq |\B|} |J_{s\nb+1,j}^{\circ}| \leq x\right) \right|\\ 
& \leq C \Delta^{1/3} \left\{1 \vee \log (n|\B| /\Delta)\right\}^{2/3},\label{eq:step2Gauss}
\end{align}
where 
$
\Delta \leq  \max_{1 \leq j \leq |\B| , 1 \leq k, l \leq n - 2\nb +1}\left| \E(J_{k,j}^{*} J_{l,j}^{*}) -   \E(J_{k,j}^{\circ} J_{l,j}^{\circ})\right|
$. We proceed to derive the upper bound of $\Delta$.
First, by elementary calculation we have uniformly for $1 \leq k,l \leq n- 2\nb+1$, $1 \leq j \leq |\B|$,
\begin{align}
  \E(J_{k,j}^{\circ} J_{l,j}^{\circ})  &= \frac{1}{nb_z \kappa} \sum_{p,q=1}^{2\nb} \E v_{z, p+k-1}v_{z, q+l -1} K_{b_z}\left(\frac{p-\nb}{n}\right)K_{b_z}\left(\frac{q-\nb}{n}\right)\\
  & =  \frac{1}{nb_z \kappa} \sum_{p,q=1}^{2\nb} \mf 1(p = q+l-k) K_{b_z}\left(\frac{p-\nb}{n}\right)K_{b_z}\left(\frac{q-\nb}{n}\right)\\ 
  &= \kappa^{-1} \int K(u)K\left(u + \frac{k-l}{nb_z}\right) du + O((nb)^{-1}).\label{eq:Jcirc}
\end{align}
For the calculation of $\E(J_{k,j}^{*} J_{l,j}^{*})$,  since $(\tilde{\mf Z}_i)_{i=1}^{2\nb} \in \R^{(n - 2\nb+1)|\B|}$ share the same autocovariance structure with the vectors $(\bar{\mf \Xi}^{\B}_i)_{i=1}^{2\nb} $, by similar argument of Lemma C.3 in \cite{Dette2021ConfidenceSF} we get  
\begin{align}
&\frac{1}{nb}\mathrm{Cov}\left(\sum_{i=1}^{2\nb} \tilde Z_{i,(k-1)|\B|+j}\sum_{i=1}^{2\nb} \tilde Z_{i,(l-1)|\B|+j}\right) \\
& = \frac{1}{nb}\sum_{p,q=1}^{2\nb} \frac{\E \Xi_{z, p+k-1} \Xi_{z, q+l-1} K_{b_z}\left(\frac{p-\nb}{n}\right)K_{b_z}\left(\frac{q-\nb}{n}\right)}{\tilde \Gamma_z(\frac{k-1+\nb}{n}) \tilde \Gamma_z(\frac{l-1+\nb}{n}) }\\
& = \frac{1}{nb}\sum_{p,q=1}^{n} \frac{\E \Xi_{z, p} \Xi_{z, q} K_{b_z}\left(\frac{p-k+1-\nb}{n}\right)K_{b_z}\left(\frac{q-l+1-\nb}{n}\right)}{ \tilde \Gamma_z(\frac{k-1+\nb}{n}) \tilde \Gamma_z(\frac{l-1+\nb}{n}) }\\
 &=\kappa^{-1} \int K(u)K\left(u + \frac{k-l}{nb_z}\right) du + O(b\log b + \chi^{n b r_n}+ r_n + (nb)^{-1}) ,
\end{align}
where $r_n = M \log n /(nb)$, $M$ is a sufficiently large constant and the big $O$ only depends on the dependence measure and  the coefficients of stochastic Lipschitz continuity, which are uniformly bounded by \cref{Ass:error}. Therefore, uniformly for $1 \leq k, l \leq n-2\nb +1$ and $1 \leq j \leq \B$, we have  \begin{align}
\E(J_{k,j}^{*} J_{l,j}^{*}) = \kappa^{-1} \int K(u)K\left(u + \frac{k-l}{nb_z}\right) du + O(b\log b + \chi^{n b r_n}+ r_n + (nb)^{-1}).\label{eq:Jstar}
\end{align}
By \eqref{eq:step2Gauss}, \eqref{eq:Jcirc} and \eqref{eq:Jstar}, we have $\Delta = O(b\log b + \chi^{n b r_n}+ r_n + (nb)^{-1})$ and thus \eqref{eq:step2} follows.\par
\textbf{Proof of Step 3}. 
Define for $t \in [0,1]$,
\begin{align}
   \bs  \ell_z(t) = (K_{b_z}(1/n - t), \cdots, K_{b_z}(1-t))^{\T}/(\sqrt{nb_z} \kappa^{1/2}), \quad \mf T_z(t) = \bs \ell_z(t)/| \bs \ell_z(t)|.
\end{align}
Observe that 
\begin{align}
    \sum_{i=1}^{2\nb} \frac{v_{z, i + s\nb} K_{b_z}(\frac{i-\nb}{n})}{\sqrt{nb_z \kappa}} = \langle \bs \ell_z((s+1)\nb/n), \mf v_z \rangle,
\end{align}
where $\mf v_z = (v_{z,1}, \cdots, v_{z,n})^{\T}$, $\langle\cdot, \cdot\rangle$ denotes the inner product. To apply \cite{sun1994simultaneous}, we first approximate the maximands over discrete $s = 0, 
\cdots, m^*$ by the supreme over $t \in [0,1]$. Following similar lines in the proof of \cref{nonasynetwork}, we have 
\begin{align}
    &\sup_{x \in \mathbb R}\left|\prod_{z \in \B} P\left( \max_{0 \leq s\leq m^*} |\langle \mf T_z((s+1)\nb/n), \mf v_z \rangle| \leq x \right) - \prod_{z \in \B}  P\left( \sup_{t \in [0,1]} |\langle \mf T_z(t), \mf v_z \rangle| \leq x \right)\right|\\
     &=\sup_{x \in \mathbb R}\left| P\left(\max_{z \in \B}  \max_{0 \leq s\leq m^*} |\langle \mf T_z((s+1)\nb/n), \mf v_z \rangle| \leq x \right) -  P\left( \max_{z \in \B} \sup_{t \in [0,1]} |\langle \mf T_z(t), \mf v_z \rangle| \leq x \right)\right|\\
     & = O(\Theta(|\B|^{1/(q+1)}(nb)^{-q/(q+1)}, n|\B| )).\label{eq:Tdiscrete}
\end{align}
By Proposition 1 of \cite{sun1994simultaneous}
since $x > \sqrt{2\log b^{-1}-2c}$ for the constant $c$ defined in Step 1, we have
\begin{align}
    \prod_{z \in \B}  P\left( \max_{t \in [0,1]} |\langle \mf T_z(t), \mf v_z \rangle| > x \right) = \prod_{z \in \B}(b_z^{-1} C_{\kappa} \exp(-x^2/2)  + 2(1 - \Phi(x))) + O(|\B|b),\label{eq:Tcontinuous}
\end{align}
where $\Phi(x)$ is the cumulative function of normal distribution, $C_{\kappa} =  \sqrt{\int |K^{\prime}(x)|^2 dx /\int K^2(x) dx}$, $|\B|b = O(b\log n ) =o(1)$. Combining \eqref{eq:Tdiscrete}, \eqref{eq:Tcontinuous} and
 $|\bs \ell_z(t)| = 1+O((nb)^{-1})$, we obtain 
\begin{align}
&\prod_{z \in \B} P\left(\max_{0 \leq s \leq m^*} \left|\sum_{i=1}^{2\nb} \frac{v_{z, i+s\nb} K_{b_z}(\frac{i-\nb}{n})}{\sqrt{nb_z \kappa}} \right| > x \right) \\ 
&= \prod_{z \in \B}(b_z^{-1} C_{\kappa} \exp(-x^2/2)  + 2(1 - \Phi(x))) + O(|\B|b + \Theta(|\B|^{1/(q+1)}(nb)^{-q/(q+1)}, n|\B| )) ,
\end{align}
   For a sufficiently large $q$, since $nb^3 \to \infty$, $ \Theta(|\B|^{1/(q+1)}(nb)^{-q/(q+1)}, n|\B| )$ also converges to 0. 
   \par 
\textbf{Proof of \eqref{eq:approx_tildeZup}.}
By Theorem 2 of \cite{latala2017royen}, we have 
   \begin{align}
    P\left(\left| \frac{1}{\sqrt{n b}}\sum_{i=1}^{2\nb} \bar{\tilde{\mf Z}}_i \right|_{\infty} \leq x\right) \geq \prod_{j=1}^{|\B|}P\left(\max_{1 \leq k \leq n - 2\nb +1} \left|\frac{1}{\sqrt{n b}}\sum_{i=1}^{2\nb} \bar{\tilde{Z}}_{i, (k-1)|\B| + j} \right| \leq x \right).\label{eq:step1up}
   \end{align}
   By similar arguments in the proof of Step 2 (i.e., \eqref{eq:step2}), we have 
   \begin{align}
   &\sup_{x \in \R} \left| \prod_{j=1}^{|\B|}P\left(\max_{1 \leq k \leq n - 2\nb +1} \left|\frac{1}{\sqrt{n b}}\sum_{i=1}^{2\nb} \bar{\tilde{Z}}_{i, (k-1)|\B| + j}\right| \leq x \right) \right. \\ & \left.- \prod_{z \in \B} P\left(\max_{1 \leq k \leq n - 2\nb +1} \left|\sum_{i=1}^{2\nb} \frac{v_{z, i+k-1} K_{b_z}(\frac{i-\nb}{n})}{\sqrt{nb_z \kappa}} \right| \leq x \right)\right| = o(1). \label{eq:step2up}
   \end{align}
   Finally,  by similar argument in the proof of Step 3 (i.e., \eqref{eq:step3}), for $x/ \sqrt{2\log \log n} \to \infty$,
   \begin{align}
   & \prod_{z \in \B} P\left(\max_{1 \leq k \leq n - 2\nb +1} \left|\sum_{i=1}^{2\nb} \frac{v_{z, i+k-1} K_{b_z}(\frac{i-\nb}{n})}{\sqrt{nb_z \kappa}} \right| \leq x \right)\\ &= \prod_{z \in \B}\left\{1-(b_z^{-1} C_{\kappa} \exp(-x^2/2)  + 2(1 - \Phi(x)))\right\} + o(1),\label{eq:step3up}
\end{align}
   where $\Phi(x)$ is the cumulative function of normal distribution, $C_{\kappa} = \sqrt{\int |K^{\prime}(x)|^2 dx /\int K^2(x) dx}$. \eqref{eq:approx_tildeZup} then follows from \eqref{eq:compare}, \eqref{eq:step1up}, \eqref{eq:step2up}, \eqref{eq:step3up} and \cref{Ass:ck}.  \hfill $\Box$

\bigskip
\newpage
\begin{center}
   {\fontsize{15}{30} \bf  Supplement to "Time-varying correlation network analysis of non-stationary multivariate time series with complex trends" }\\[6pt]
   {Lujia Bai and Weichi Wu}\\[6pt]
   {Center for Statistical Science, Department of Industrial Engineering, Tsinghua University}
\end{center}

We organize the supplement as follows:  In \cref{sec:alg}, we give the algorithm equipped with plug-in estimators when the trend functions are smooth as mentioned in \cref{rm:ncp}. \cref{sec:proof1} provides lemmas and propositions used in the paper, their corresponding proofs, and the theoretical justification of the algorithm using plug-in estimators. 
\appendix
\setcounter{section}{3}
\section{The plug-in algorithm}\label{sec:alg}\label{sec:plg}
Recall that $\hat \epsilon_{j,i} := Y_{j,i} - \hat \mu_i(t_j)$, $1 \leq i \leq p$, $1 \leq j \leq n$, where $\hat \mu_i(t_j)$ is the local linear estimator, i.e.,\begin{align}
    (\hat \mu_i(t) , \hat \mu^{\prime}_i(t))=\underset{\eta_0, \eta_1}{\mathrm{argmin}}  \sum_{j=1}^n\{Y_{j,i}  - \eta_0 -\eta_1(t_j-t) \}^2 K_{\tau_i}(t_j- t),\label{eq:muloclin}
\end{align}
where $\tau_i$ is the bandwidth parameter for $\hat \mu_i(\cdot)$.
 When \eqref{eq:model_spec} in the main article  has no change point, we can remove the trend function and use                   directly the residuals to estimate the cross-correlations, i.e., $\hat \gamma^{i,l}_{k}(t) = \frac{1}{nb_k^{i,l}} \sum_{j=1}^n \hat \epsilon_{j,i}\hat \epsilon_{j+k,l} K_{b^{i,l}_k}(t_j- t)$, 
and
\begin{align}
      \hat \rho^{i,l}_{k}(t) = \hat \gamma^{i,l}_{k}(t)/\hat \sigma_{i,l}(t),\quad \hat \sigma_{i,l}(t) = \sqrt{\hat \gamma^i_{0}(t) \hat \gamma^l_{0}(t)},\quad (i,l,k) \in \B,\label{eq:rho_network_plug}
\end{align}
which would lead to  a non-trivial extension of \cite{zhao2015inference} from the inference of certain local autocorrelation curve to the joint inference of cross-correlation curves of multivariate and high dimensional non-stationary time series with possibly diverging number of lags.
 Recall that $\hat \gamma^{i}_{k}(t)$ are short for $\hat \gamma^{i,i}_{k}(t)$. Define the processes and the estimators of the residuals as
\begin{align}
e^{i,l}_{j,k}  := \epsilon_{j,i} \epsilon_{j+k,l} - \gamma^{i,l}_{k}(t_j) , \quad 
\hat e^{i,l}_{j,k} := \hat \epsilon_{j,i}\hat \epsilon_{j+k,l} - \hat \gamma^{i,l}_{k}(t_j).\label{eq:res}
\end{align}
Write $e_{j,k}^i$ short for $e_{j,k}^{i,i}$.
The counterpart of \eqref{eq:Xidiscrete} becomes 
\begin{align}
V^{i,l}_{j,k} =e^{i,l}_{j,k}/\sigma_{i,l}(t_j) - 2^{-1} \rho_{k}^{i,l}(t_j)\left(e_{j,0}^i /\gamma_0^i(t_j) +  e_{j,0}^l /\gamma_0^l(t_j) \right).\label{eq:defineV}
\end{align}
Suppose $V^{i,l}_{j,k}$ admits the form $V^{i,l}_{j,k} = L^{i,l}_{k}(t_j, \FF_j)$, where $ L^{i,l}_{k}(\cdot, \cdot)$ is a filter such that  $L^{i,l}_{k}(t, \FF_j)$ is well defined. 
Similarly, we can define $\Gamma^{i,l}_{k}(t)$ as  the square root of the limiting variance of $  (nb^{i,l}_k)^{-1/2}\sum_{j=1}^n K_{b^{i,l}_k}(t_j - t)V^{i,l}_{j,k}$, and
\begin{align}
 \bar{\mf V}_{j,s}^{\B} := (c_k^{i,l} K_{b^{i,l}_{k}}(t_j - t_s)V^{i,l}_{j,k}/  \Gamma^{i,l}_{k}(t_s), (i,l,k)\in \B)^{\T}, \quad 1 \leq s\leq n, ~j=\nt, \cdots, 2\nb - \nt.
 \end{align}
  and we get 
$
 \bar{\mf V}_{j}^{\B} = (\bar{\mf V}_{j,\nb+ \nt}^{\B, \T}, \bar{\mf V}_{j + 1,\nb + \nt+ 1}^{\B, \T}, \cdots, \bar{\mf V}_{n - 2\nb +j, n - \nt-\nb}^{\B, \T})^{\T}$.
 Finally, we estimate $V^{i,l}_{j,k}$ and $\Gamma_{k}^{i,l}(u)$ in $ \bar{\mf V}_{j}^{\B}$ by plugging in the residuals of \eqref{eq:res} and the estimators of \eqref{eq:rho_network_plug}, i.e., 
 \begin{align}
 \tilde V^{i,l}_{j,k} =  \hat \epsilon_{j,i}\hat \epsilon_{j+k,l}/\hat \sigma_{i,l}(t_j) - 2^{-1}  \hat \rho_{k}^{i,l}(t_j)\left( \hat \epsilon_{j,i}^2 /  \hat \gamma_0^i(t_j) +  \hat \epsilon_{j,l}^2  /  \hat \gamma_0^l(t_j) \right), 
 \end{align}
 and 
 \begin{align}
    \hat{\Gamma}_{k}^{i,l,2}(t) =  \frac{\kappa}{m} \sum_{s=1}^{n} \Delta^{i,l,2}_{k,s}\omega(t,s),\quad \omega(t,s) = K_{\eta}(t-t_s)/\sum_{k=1}^n K_{\eta}(t_k - t),
\end{align}
where $\Delta^{i,l}_{k,s} = \sum_{j=s}^{s+m -1} \hat V^{i,l}_{j, k}$, and $\hat V^{i,l}_{j,k} =  \hat e^{i,l}_{j,k}/\hat \sigma_{i,l}(t_j) - 2^{-1}  \hat \rho_{k}^{i,l}(t_j)\left( \hat e_{j,0}^i /  \hat \gamma_0^i(t_j) +   \hat e_{j,0}^l /  \hat \gamma_0^l(t_j) \right)$. 
Define $ \bar{\hat{\mf V}}_{j}^{\B} = (\bar{\hat{\mf V}}_{j,\nb+\nt}^{\B, \T}, \bar{\hat{\mf V}}_{j + 1,\nb + \nt + 1}^{\B, \T}, \cdots, 
 \bar{\hat{\mf V}}_{n - 2\nb +j, n - \nb-\nt}^{\B, \T})^{\T}$, $j=\nt, \cdots, 2\nb - \nt$, as the estimator of $\bar{\mf V}_j^{\B}$ using estimators defined above, where
\begin{align}
   \bar{\hat{\mf V}}_{j,s}^{\B} =  (c_k^{i,l} K_{b^{i,l}_{k}}(t_j - t_s)\tilde V^{i,l}_{j,k}/  \hat \Gamma^{i,l}_{k}(t_s), (i,l,k)\in \B)^{\T}.\label{eq:hatV}
\end{align}

 Write $\tau = \max_{i=1,\cdots,p} \tau_i$. The algorithm using plug-in estimation is shown in \cref{alg:jointncp}. 
\begin{algorithm}
    \caption{Plug-in estimation of time-varying network}
    \label{alg:jointncp}
    \begin{algorithmic} [1]
        \State Compute the residuals $\hat \epsilon_{j,i} := Y_{j,i} - \hat \mu_i(t_j)$, where  $\hat \mu_i(t_j)$ is as defined in \eqref{eq:muloclin}, $1\leq i \leq p$, $1\leq j \leq n$.
        \State Compute $\hat \rho^{i,l}_k(t)$, $(i,l,k) \in \B$, in \eqref{eq:rho_network_plug}.
        \State Compute the $|\B|$-dimensional vectors $\bar{\hat{\mf V}}_{j,s}^{\B}$ in \eqref{eq:hatV}, $1 \leq j, s \leq n$.
        \State For  window size $w$, compute $\hat{\mf S}_{l, j}^{\B} = \sum_{s = j-w+1}^{j} \bar{\hat{\mf V}}_{s+l, \nb + l} ^{\B}-\sum_{s = j+1}^{j+w} \bar{\hat{\mf V}}_{s+l, \nb + l}^{\B}$, where $l = 0, \cdots, n - 2\nb$. 
                \For{$r = 1, \cdots, B$}
        \State Generate independent standard normal random variables $R_i^{(r)}$, $i=1,\cdots,n$. 
         \State Calculate  
                $$Z_{\bt}^{(r)}= \frac{\underset{0 \leq l \leq n - 2\nb}{\max} \left| \sum_{j= w + \nt -1}^{ 2\nb - w - \nt} \hat{\mf S}_{l, j}^{\B}  R^{(r)}_{l+j}\right|_{\infty}}{\sqrt{2 w(\nb -\nt)}}$$
    \EndFor
    \State Let $\hat r_{\bt}$ denote the $(1-\alpha)$-quantile of the bootstrap sample $Z_{\bt}^{(1)}, \cdots,Z_{\bt}^{(B)}$.  
     \State
     Connect $i$ and $l$ at time $t \in [b, 1-b]$ if 
    $ \sqrt{nb^{i,l}_k} |g_{ilk}(t) - \hat \rho^{i,l}_k(t)| > \hat r_{\bt}   \hat{\Gamma}^{i,l}_k(t)$, for some $k$ such that $(i,l, k) \in \B$.
%
    \end{algorithmic}

\end{algorithm}
   \begin{remark}
   We select the smoothing parameters using the schemes presented in \cref{sec:impl} with the following modification.  
  For the estimation of $\mu(\cdot)$, we can write $\hat{\mf Y}_0 = \mf Q_0(\tau)\mf Y_0$ for some square matrix $\mf Q_0$ depending on $\tau$, where $\mf Y_0 = (Y_{1,i}, \cdots, Y_{n,i})^{\T}$, and $\hat{\mf Y}_0=(\hat Y_{1,i},...,\hat Y_{n,i})^\top$ is the estimated value of $\mf Y_0$ via the bandwidth $\tau_i$, i.e., $\hat Y_{j,i}= \hat {\mu}_i(t_j)$. Then we select $\tau_i$ by minimizing
\begin{equation}
     \operatorname{GCV}(\tau)=\frac{n^{-1}|\mathbf{Y_0}-\hat{\mathbf{Y}}_0|^{2}}{[1-\operatorname{tr}\{\mf Q_0(\tau)\} / n]^{2}}.\label{GCV_tau}
    \end{equation}
  We select $w$ and $\eta$ in the bootstrap algorithm \cref{alg:jointncp} also by the extended minimum volatility (MV) method. As discussed in \cref{sec:impl}, we first propose a grid of possible block sizes and bandwidths $\{w_{1}, w_{2},\cdots, w_{M_1}\}$, $\{\eta_1, \eta_2,\cdots, \eta_{M_2}\}$. Define the sample variance $s^2_{w_i,\eta_j}$ of the bootstrap statistics as 
\begin{align}
 s^2_{w_{i},\eta_j} =  \sum_{l=0}^{n- 2\nb} \sum_{r=w_{i}}^{n-w_{i}}\hat{\mf S}_{l, r, w_{i}, \eta_j}^{\B, \top} \hat{\mf S}_{l, r, w_{i}, \eta_j}^{\B}, 
\end{align}
where $\hat{\mf S}_{l, r, w_{i}, \eta_j}^{\B}$ is as defined in \cref{alg:jointncp}  using bandwidth $w_{i}$, and  $m^{i,l}_k = \lf n^{2/7} \rf$ and $\eta_j$ for $\hat{ \Gamma}^{i,l}_k(\cdot)$. 
Then select $(i, j)$ which minimizes the following criterion,
   \begin{align}
    \mathrm{MV}(i,j):= \mathrm{SD}\left\{\cup_{r=-1}^{1}\{s^2_{w_{i}, \eta_{j+r}}\} \cup \cup_{r=-1}^{1}\{s^2_{w_{i+r}, \eta_{j}}\}\right\},\nonumber
    \end{align} where SD stands for the sample standard deviation.
   \end{remark}
   For $z = (i,l,k) \in \B$, $t \in [0,1]$, write $\sqrt{nb_z} |g_{z}(t)-\hat \rho_z(t)|/\hat{ \Gamma}_z(t) $ as $T_z(t)$ and let $\hat N(t) = \{ T_z(t)
    \leq \hat r_{\bt}, z\in \B -\tilde \B(t)\} \cap  \{ T_z(t)
    > \hat r_{\bt}, z \in \tilde \B(t)\}$.
   Define $
        \mathcal C_n = \{\mf x(\cdot)=(x_z(\cdot), z\in 
        \B)^{\top} \in [b,1-b]^{|\B|}: \sqrt{nb_z} |x_z(t) - \rho_z(t)|\leq   \hat r_{\bt}  \hat \Gamma_z(t), \forall t \in [b,1-b]\}
   $. Analogous to \cref{boot:net}, the following theorem ensures the asymptotic type I error control and the recovery probability of  \cref{alg:jointncp}, the proof of which is deferred to Section \ref{pf:ncp}.
   \begin{theorem}\label{boot:ncp}
    Under Assumptions \ref{Ass:error}, \ref{A:K}, \ref{Ass:ck}, assuming that $\underset{t \in (0,1)}{\inf}  \sigma(L^{i,l}_{k},t) > 0$, and the second derivative of $\mu_i(\cdot)$  exists and with uniformly bounded Lipschitz constants on $[0,1]$ for $i=1,2,\cdots, p$,  $|\B|^{1/q} (\sqrt{n b^7} + (nb \tau^2)^{-1/2} )\to 0$. Further assume for some sufficiently large $q$,
    \begin{align}
        \vartheta_n^{1/3}\left\{ 1 \vee \log (n|\B|/\vartheta_n) \right\}^{2/3} +\left(e_n (n|\B|)^{1/q}\right)^{q/(q+2)} \to 0,
    \end{align}
    where $\vartheta_n = \log^2 n/w + w/(nb) + \sqrt{w/(nb)}(n|\B|)^{4/q}$ and   $e_n = |\B|^{1/q} (\sqrt{m/(n \eta^{2})}+1/m+\eta^2  +\sqrt{m/(nb)}(mb/n)^{-1/(2q)}  + (nb)^{-1/2} b^{-1/q} + 1/\sqrt{w\tau^2})+ w^{3/2}/n$.
   Then,
    we have  \par
    (i)  (Type I error control.) As $n$ and $B$ go to infinity 
    \begin{align}
        P_{H_0}\big(T_z(t)
    \leq \hat r_{\bt}, t\in[b, 1-b], z\in \B| \F_n \big) = P(\{\rho_z(\cdot), z \in \B\} \in \mathcal C_n | \F_n) \overset{p}{\to} 1-\alpha.
    \end{align}\par
    (ii) (The lower bound of the recovery probability.) Under \cref{ass:alter}, we have $$\underset{n \to \infty}{\lim} \underset{B \to \infty}{\lim} \int_0^{1} P(\hat N(t)| \F_n) dt\geq 1-\alpha,$$ with probability approaching $1$  for any arbitrarily small $\alpha>0$,
    where $\hat N(t)$ is the event for correctly recovering the network at time $t$. 
\end{theorem}
\begin{remark}

 \cref{boot:ncp} admits both cases when $|\B|$ is fixed and  when $|\B|$ diverges. The conditions of \cref{boot:ncp} can be satisfied for sufficiently large $q$,  if $w \asymp \lfloor n^{2/5} \rfloor$, $\eta \asymp  n^{-1/7}$, $b \asymp  n^{-1/5}$, $m \asymp  \lfloor n^{2/7} \rfloor$ and $\tau \asymp  n^{-1/6}$.
\end{remark}

\section{Proofs}\label{sec:proof1}

\subsection{Proof of \texorpdfstring{\cref{nonasynetwork}}{Theorem 3.1}}

Recall that in \cref{prop:rhok}, 
$\vartheta^{i,l}_k(t) = \frac{1}{nb_{k}^{i,l}}\sum_{j=1}^n K_{b_{k}^{i,l}}(t_i-t)\Xi^{i,l}_{j,k}$.  By \cref{prop:rhok}, we have 
 \begin{align}
  \max_{(i,l,k) \in \B} \left\|\sup _{t \in \mathcal{T}}\left| \tilde{\rho}^{i,l}_{k}(t)-\rho^{i,l}_{k}(t)- \vartheta^{i,l}_k(t)\right|\mf 1(\bar B^{\prime}_n)\right\|_q= O\left( b^{-1/q}(n^{\phi-1}b^{-1}h  + b^{3 }+n^{-1/2}h)\right).
  \label{eq:suprhok_diff}
 \end{align}
Under \cref{A:K}, we have $\mu_2 = \int u^2 K(u) du = 0$.  Let $\TT_n = [\nb, n-\nb]$.

Notice that 
$ \stackrel[(i,l,k) \in \B]{}{\max}\underset{j \in \TT_n}{\max}\left|\tilde \rho_k^{i,l}(t_j) - \rho_k^{i,l}(t_j) - \vartheta^{i,l}_k(t_j) \right|^q \leq \underset{(i,l,k) \in \B}{\sum} \underset{j \in \TT_n}{\max} \left|\tilde \rho_k^{i,l}(t_j) - \rho_k^{i,l}(t_j) - \vartheta^{i,l}_k(t_j) \right|^q$.
By \eqref{eq:suprhok_diff} we have
\begin{align}
    \left\| \max_{ (i, l, k) \in \B}\max_{ j \in \TT_n}\left|\tilde \rho_k^{i,l}(t_j) - \rho_k^{i,l}(t_j) - \vartheta^{i,l}_k(t_j)\right| \mf 1(\bar B^{\prime}_n)\right\|_q = O( (|\B|/b)^{1/q} (hn^{\phi - 1}b^{-1}  + b^{3}+n^{-1/2}h)).\label{eq:rhoqnorm}
\end{align}
Similarly, by \cref{lm:loclin} and \eqref{eq:betagamma}, we have $P(B^{\prime}_n) = O((q_n^{\prime})^{-q})$. 
Write $c_n =(|\B|/b)^{1/q} (hn^{\phi - 1}b^{-1}  + b^{3}+n^{-1/2}h) $ for short. 
By elementary calculation similar to Lemma C.3 in \cite{Dette2021ConfidenceSF}, \begin{align}
  nb^{i,l}_k\left\| \vartheta^{i,l}_k(t) \right\|^2 &= \tilde \Gamma^{i,l,2}_{k}(t)+ O(b\log b + \chi^{n b r_n}+ r_n + (nb)^{-1}),\label{eq:varV_diff}
\end{align}
where $r_n = M \log n/(nb)$ for some sufficiently large positive $M$. Since $\sum_{s=1}^{2\nb}c_k^{i,l}\Xi^{i,l}_{j,k}K_{b^{i,l}_k}(t_j -t_s)/(\sqrt{nb} \tilde \Gamma_k^{i,l}(t_j)) = \sqrt{nb} |\vartheta_k^{i,l}(t_j)|/\tilde \Gamma_k^{i,l}(t_j)$,  it can be verified that under \eqref{eq:varV_diff} 
 $\frac{1}{\sqrt{nb}}\sum_{i=1}^{2\nb}\bar{\mf \Xi}_{i}^{\B}$ satisfies Condition (9) in Corollary 2  of \cite{zhang2018}. Then, there exists a sequence of zero-mean Gaussian vectors $(\tilde{\mf Z}_i)_{i=1}^{2\nb} \in \R^{(n - 2\nb+1)|\B|}$, which share the same autocovariance structure with the vectors $(\bar{\mf \Xi}^{\B}_i)_{i=1}^{2\nb} $  such that
\begin{align}
    \sup_{x \in \R}\left|\PP\left(\left| \frac{1}{\sqrt{nb}}\sum_{i=1}^{2\nb}\bar{\mf \Xi}_{i}^{\B}\right|_{\infty}\leq x \right) -\PP\left(\left| \frac{1}{\sqrt{nb}}\sum_{i=1}^{2\nb} \tilde{\mf Z}_i\right|_{\infty} \leq x\right) \right| = O((nb)^{-(1-11\iota)/8}).
\end{align}
Then by Lemma C.1 in \cite{Dette2021ConfidenceSF}, we have 
\begin{align}
    &\sup_{x \in \R}\left|\PP\left(\max_{(i,l,k) \in \B} \max_{j\in \TT_n}\sqrt{nb^{i,l}_k}|\tilde \rho_k^{i,l}(t_j) - \rho_k^{i,l}(t_j)|/\tilde \Gamma_k^{i,l}(t_j)\leq x\right)- \PP\left(\left|  \frac{1}{\sqrt{nb}} \sum_{i=1}^{2\nb} \tilde{\mf Z}_i\right|_{\infty}\leq x\right)\right|\\ 
    &= O((nb)^{-(1 - 11\iota)/8}+(\sqrt{nb}c_n/\delta)^{q}
    + \Theta\left(\delta , n |\B|\right) + (q_n^{\prime})^{-q} ),\label{eq:ncpstep1_diff}
\end{align}
where solving $(\sqrt{nb}c_n/\delta)^q = \delta$ and $\delta = (q_n^{\prime})^{-q}$, we have $\delta = (\sqrt{nb}c_n)^{q/(q+1)}$ and $q_n^{\prime} = (\sqrt{nb}c_n)^{-1/(q+1)}$.
Since $|\B|^{1/q} h n^{\phi - 1/2}b^{-1/2-1/q} \to 0$, $|\B|^{1/q} n^{1/2}b^{7/2-1/q} \to 0$, we have $c_n = (|\B|/b)^{1/q} (hn^{\phi - 1}b^{-1}  + b^{3}+n^{-1/2}h) = o((nb)^{-1/2})$. Therefore, $q_n^{\prime} \to \infty$.
Under the bandwidth condition $(nb)^{-1/2}\{|\B|/(hn^{\phi-1}+b^4 + n^{-1/2}hb)\}^{1/(q+2)} \to 0$, we have $f_n q_n^{\prime} \to \infty$.
Note that $c_n  = o((nb)^{-1/2}) = o(b^{-1})$. Then, it also holds that $$\sup_{t \in \TT} \left\|\frac{\partial}{\partial t } \sqrt{nb^{i,l}_k} (\tilde \rho_k^{i,l}(t) - \rho_k^{i,l}(t))\mf 1(\bar B^{\prime}_n)\right\|_q \leq M / b,$$ where $M$ is a sufficiently large constant. By Taylor's expansion, the continuity of $\tilde \Gamma^{i,l}_k(t)$ as well as the strict positiveness of $\tilde \Gamma^{i,l}_k(t)$ by \cref{Ass:diff}, we have 
\begin{align}
    \left\| \max_{(i,l,k) \in  \B} \sqrt{nb^{i,l}_k} \sup_{|t_j- t| \leq n^{-1}, j \in \TT_n, t\in[0,1]}\left |\frac{\tilde \rho_k^{i,l}(t) - \rho_k^{i,l}(t)}{ \tilde \Gamma_k^{i,l}(t)} - \frac{ \tilde \rho_k^{i,l}(t_j) - \rho_k^{i,l}(t_j)}{\tilde \Gamma_k^{i,l}(t_j)}\right|\mf 1(\bar B^{\prime}_n)\right\|_q = O(|\B|^{1/q} (nb)^{-1}).\label{eq:ncpstep2_diff}
\end{align}
Combining \eqref{eq:ncpstep1_diff} and \eqref{eq:ncpstep2_diff}, following similar arguments of Theorem 3.2 in \cite{Dette2021ConfidenceSF}, we obtain
\begin{align}
        \sup_{x \in \R}\left|\PP\left(\max_{(i,l,k) \in \B} \sup_{t \in \TT} \sqrt{nb^{i,l}_{k}}|\tilde \rho^{i,l}_{k}(t) - \rho^{i,l}_{k}(t)|/\tilde \Gamma^{i,l}_{k}(t)\leq x\right)- \PP\left(\left|\frac{1}{\sqrt{n b}} \sum_{i=1}^{2\nb} \tilde{\mf Z}_i\right|_{\infty}\leq x\right)\right| = O(\tilde \theta_n),
        \label{eq:lemma1}
 \end{align}
where 
$\tilde \theta_n = (n b)^{-(1 - 11\iota)/8}
+ \Theta\left((\sqrt{nb}c_n)^{q/(q+1)} , n |\B|\right) + \Theta\left( |\B|^{1/(q+1)}(nb)^{-q/(q+1)}, n |\B|\right)$.
Note that for a sufficiently large $q$, we have $\tilde \theta_n = o(1)$. \hfill $\Box$
\subsection{Proof of \texorpdfstring{\cref{boot:ncp}}{Theorem D.1}}\label{pf:ncp}

Following (A.8) and (A.10) of the proof of Theorem 3.1 in \cite{dette2018change}, assuming $\tau \to 0$, we have 
\begin{align}
    \max_{(i,l,k) \in \B} \left\|\max_{\nt \leq r \leq n - \nt }\left|\sum_{j=1}^r \hat \epsilon_{j,i}\hat \epsilon_{j+k,l}  - \sum_{j=1}^r \epsilon_{j,i} \epsilon_{j+k,l} \right|\right\|_q = O(\log^2 n/\sqrt{n\tau^2} + \tau^{-1}) = O(\tau^{-1}). \label{eq:doubleres}
\end{align}
Recall that $\hat \gamma^{i,l}_k(t) = \frac{1}{nb^{i,l}_k} \sum_{j=1}^n \hat \epsilon_{j,i}\hat \epsilon_{j+k,l} K_{b^{i,l}_k}(t_j-t)$.
Then, by the summation-by-parts formula and \eqref{eq:doubleres}, we have 
\begin{align}
    \max_{(i,l,k) \in \B}  \left \|\sup_{t \in [\tau, 1-\tau] }\left|\hat \gamma^{i,l}_k(t) - \frac{1}{nb^{i,l}_k} \sum_{i=1}^n \epsilon_{j,i} \epsilon_{j+k,l}  K_{b_k^{i,l}}(t_i-t)\right| \right\|_q  = O((nb\tau)^{-1}). \label{eq:secondres}
\end{align}
Recall the definition of $e^{i,l}_{j,k}$ in  \eqref{eq:res}. 
By Lemma B.1 in the supplememt of \cite{dette2018change}, under \cref{A:K}, which yields that $\mu_2 = \int u^2 K(u) du = 0$ , uniformly for $(i, l, k) \in \B$, we have
\begin{align}
    \left\| \sup_{t \in  [b, 1-b]}\left| \frac{1}{nb^{i,l}_k} \sum_{i=1}^n \epsilon_{j,i} \epsilon_{j+k,l} K_{b^{i,l}_k}(t_i - t) - \gamma^{i,l}_k(t) - \frac{1}{nb^{i,l}_k}\sum_{j=1}^n K_{b^{i,l}_k}(t_j - t)e^{i,l}_{j,k} \right| \right\|_q = O((nb)^{-1}+ b^3).  \label{eq:secondlocal}
\end{align}
 Combining \eqref{eq:secondres} and \eqref{eq:secondlocal}, it follows that uniformly for $(i, l, k) \in \B$,
\begin{align}
    \left\|\sup_{t \in   [\tau+b , 1-\tau-b]}\left| \hat \gamma^{i,l}_k(t) - \gamma^{i,l}_k(t) - \frac{1}{nb^{i,l}_k}\sum_{j=1}^n K_{b^{i,l}_k}(t_j - t)e^{i,l}_{j,k} \right| \right\|_q = O((nb\tau)^{-1} + b^3) = O(a_n),\label{eq:supergamma}
\end{align}
where $a_n =(nb\tau)^{-1} + b^3$. Let $\TT^{\prime} = [\tau+b , 1-\tau-b]$, $\TT_n^{\prime} = [\nt + \nb , n-\nt-\nb]$.

By \eqref{eq:supergamma} and similar arguments of \eqref{eq:rhoqnorm} we have,
\begin{align}
    \left\| \max_{ (i,l,k) \in \B}\max_{s \in \TT^{\prime}_n}\left|\hat \gamma^{i,l}_k(t_s) - \gamma^{i,l}_k(t_s) - \frac{1}{nb^{i,l}_k}\sum_{j=1}^n K_{b^{i,l}_k}(t_j - t)e^{i,l}_{j,k}\right| \right\|_q = O(|\B|^{1/q} a_n).\label{eq:superBgamma}
\end{align}
Recall the definition of $V^{i,l}_{j,k}$ in \eqref{eq:defineV}.
By elementary calculation similar to Lemma C.3 in \cite{Dette2021ConfidenceSF}, uniformly for $t \in \TT^{\prime}$, 
\begin{align}
   \frac{1}{nb^{i,l}_k}\left\| \sum_{j=1}^n V^{i,l}_{j,k} K_{b^{i,l}_k}(t_j - t)\right\|^2 &= \kappa \sigma^2(L_k^{i,l},t)+ O(b\log b + \chi^{nb r_n}+ r_n + (nb)^{-1}),\label{eq:varV}
\end{align}
where $r_n = a \log n/(nb)$ for some sufficiently large positive constant $a$. Therefore, under \cref{Ass:ck}, it can be verified Condition (9) of Corollary 2  of \cite{zhang2018} is satisfied.
Then, it follows that
\begin{align}
    \sup_{x \in \R}\left|\PP\left(\left| \frac{1}{\sqrt{nb}}\sum_{i=\nt }^{2\nb -\nt}\bar{\mf V}_{i}^{\B}\right|_{\infty}\leq x \right) -\PP\left(\left| \frac{1}{\sqrt{nb}}\sum_{i=\nt }^{2\nb-\nt } \bar{\mf Z}_i\right|_{\infty} \leq x\right) \right| = O((nb)^{-(1-11\iota)/8}),
\end{align}
where $(\bar{\mf Z}_i)_{i=\nt }^{2\nb - \nt } \in \R^{(n - 2\nb+ 1)|\B|}$ is a sequence of zero-mean Gaussian vectors, which share the same autocovariance structure with the vectors $(\bar{\mf V}_i^{\B})_{i=\nt }^{2\nb - \nt} $,  $n|\B| = O(\exp(n^{\iota}))$ for some $0 \leq \iota < 1/11$. 
Since $|\B|^{1/q} a_n = o((nb)^{-1/2})$, $|\B|^{1/q} a_n = o(b^{-1})$, following similar arguments in the proof of \cref{prop:rhok} and similar arguments of the proof of \eqref{eq:lemma1} of \cref{nonasynetwork}, we have
\begin{align}
    \sup_{x \in \R}\left|\PP\left(\max_{(i,l,k) \in \B} \sup_{t \in \TT^{\prime}}\sqrt{nb^{i,l}_k} |\hat \rho^{i,l}_k(t) - \rho^{i,l}_k(t)|/ \Gamma^{i,l}_k(t)\leq x\right)- \PP\left(\left|\frac{1}{\sqrt{n b}} \sum_{i=\nt }^{2\nb-\nt } \bar{\mf Z}_i\right|_{\infty}\leq x\right)\right| = o(1).\label{eq:lemmaV}
\end{align}
Let $Z_{\bt}$ denote $Z_{\bt}^{(r)}$ in the one iteration of \cref{alg:jointncp}. Following \eqref{eq:lemmaV}, it's sufficient to show that 
    \begin{align}
        \sup_{x \in \mathbb R} \left| P(Z_{\bt} \leq x| \F_n) - P\left(\left| \frac{1}{\sqrt{n b}}\sum_{i=\nt }^{2 \nb -\nt } \bar{\mf Z}_i \right|_{\infty} \leq x\right)\right|  =\op(1).
    \end{align}
    Define 
    \begin{align}
        \mf Z^{\diamond}_{a|\B| + c} = \left( \sum_{j = w + \nt -1 }^{2\nb - w - \nt} \hat S^{\B}_{(a-1), j, c} R_{a+j-1}, a=1, \cdots, n-2\nb +1, 1 \leq c \leq |\B| \right),
    \end{align}
    where $\hat S^{\B}_{l, j, r}$ denote the $r$th element of $\hat{\mf S}^{\B}_{l, j}$. Let
    \begin{align}
        \mf Z^{\diamond} = \left( \mf Z_1^{\diamond, \T},\cdots, \mf Z_{(n-2\nb +1)|\B|}^{\diamond, \T}\right)^{\T},
    \end{align}
    and it follows that $Z_{\bt} = |\mf Z^{\diamond}|_{\infty}$. 
    Define 
    \begin{align}
        S_{(a-1), j,c }^{\B} = \sum_{i=j-w+1}^{j} \bar{ V}^{\B}_{i + (a-1), \nb + (a-1), c} - \sum_{i=j+1}^{j+w}\bar{ V}^{\B}_{i + (a-1), \nb + (a-1), c},
    \end{align}
    and $\mf Z^{\dagger}$ by substituting $\hat{S}^{\B}_{(a-1), j, c}$ in $\mf Z^{\diamond}$ by $ S^{\B}_{(a-1), j, c}$.
 Similar to the proof of \eqref{eq:dagger} of \cref{boot:net}, we have 
    \begin{align}
      &\sup_{x \in \R}\left|P\left(\left.\frac{ |\mf Z^{\dagger}|}{\sqrt{2w(\nb-\nt)}}\leq x \right|\F_n\right) - P\left(\left| \frac{1}{\sqrt{n b}}\sum_{i=\nt }^{2 \nb -\nt } \bar{\mf Z}_i \right|_{\infty} \leq x\right)\right|\\& = \Op(\vartheta_n^{1/3}\left\{ 1 \vee \log (n|\B|/\vartheta_n) \right\}^{2/3}),
      \label{eq:tobarZ}
    \end{align}
    where $\vartheta_n = \frac{\log^2 n}{w} + \frac{w}{nb} + \sqrt{\frac{w}{nb}}(n|\B|)^{4/q}$. Recall that $f_n = (nb)^{-1/2} b^{-1/q} |\B|^{1/q}$ and $g^{\prime}_n = |\B|^{1/q} (\sqrt{m/(n \eta^{2})} + 1/m + \eta  +\sqrt{m/(nb)}(mb/n)^{-1/(2q)})$.  Let $l_n = |\B|^{1/q}\tau^{-1}$.
    Define the $\FF_n$ measurable event
   $$
        A_n^{\circ} = \left\{ \max_{(i,l,k) \in \B} \sup_{t \in \TT^{\prime}} \left|\hat \Gamma^{i,l,2}_k(t) - \Gamma^{i,l,2}_k(t) \right|  > g^{\prime}_n q_n^{\circ}\right\},\quad B_n^{\circ} = \left\{ \max_{(i,l,k) \in \B} \sup_{t \in \TT^{\prime}} \left|\hat \gamma^{i,l,2}_k(t) - \gamma^{i,l,2}_k(t) \right|  > f_n q_n^{\circ}\right\},$$ and $$
        C_n^{\circ} = \left\{\max_{(i,l,k) \in \B}  \max_{\nt \leq r \leq n - \nt } \left|\sum_{j=1}^r \hat \epsilon_{j,i}\hat \epsilon_{j+k,l}  - \sum_{j=1}^r \epsilon_{j,i} \epsilon_{j+k,l} \right|  > l_n q_n^{\circ}\right\},$$
    where $q_n^{\circ}$ is a positive sequence which goes to infinity such that $(g_n^{\prime} + l_n + f_n)q_n^{\circ} \to 0$. Then by similar arguments in \cref{prop:Gammacp},  \eqref{eq:doubleres}  and \eqref{eq:superBgamma}, we have 
    \begin{align}
        P(A_n^{\circ} \cup B_n^{\circ} \cup C_n^{\circ}) = O((q_n^{\circ})^{-q}).\label{eq:set}
    \end{align}
    Then, for some large constant $M$, by the conditional normality we have
    \begin{align}
        &\E(|\mf Z^{\diamond} - \mf Z^{\dagger}|^q_{\infty}\mf 1(\bar A_n^{\circ} \cap \bar B_n^{\circ}  \cap \bar C_n^{\circ}) |\FF_n) \\ 
        &\leq M\left| \sqrt{\log n |\B|} \max_{1 \leq r \leq |\B|, 0 \leq l \leq n - 2\nb}\left(\sum_{j=w + \nt -1 }^{2\nb - w - \nt }(\hat S_{l, j, r}^{\B} - S_{l, j, r}^{\B})^2 \mf 1(\bar A_n^{\circ} \cap \bar B_n^{\circ}  \cap \bar C_n^{\circ})\right)  \right|^{q/2}.\label{eq:Zboot1}
    \end{align}
    By the continuity of $\gamma_k(\cdot)$ and $\gamma_0(\cdot)$, we have
    \begin{align}
        &\frac{1}{\sqrt{2w(\nb - \nt)}}\left\| \max_{1 \leq r \leq |\B|, 0 \leq l \leq n - 2\nb}\left(\sum_{j=w + \nt -1 }^{2\nb - w - \nt }(\hat S_{l, j, r}^{\B} - S_{l, j, r}^{\B})^2 \mf 1(\bar A_n^{\circ} \cap \bar B_n^{\circ}  \cap \bar C_n^{\circ})\right)^{1/2}  \right\|_q \\ &= O\left(\left((g^{\prime}_n + f_n + l_n/\sqrt{w})q_n^{\circ} + w^{3/2}/n\right) (n|\B|)^{1/q}\right).\label{eq:Zboot2}
    \end{align}
   Let $q_n^{\circ} =\left(\left(g^{\prime}_n+f_n+l_n/\sqrt{w}+ w^{3/2}/n\right)  (n|\B|)^{1/q}\right)^{-1/(q+2)} $. Combining \eqref{eq:tobarZ}, \eqref{eq:set}, \eqref{eq:Zboot1} and \eqref{eq:Zboot2}, following similar arguments in proving \cref{boot:net}, we have 
    \begin{align}
        &\sup_{x \in \R} \left| P(Z_{\bt} \leq x| \F_n) - P\left(\left| \frac{1}{\sqrt{n b}}\sum_{i= \nt }^{2 \nb - \nt } \bar{\mf Z}_i \right|_{\infty} \leq x\right)\right|\\ 
         & =\Op \left(\vartheta_n^{1/3}\left\{ 1 \vee \log (n|\B|/\vartheta_n) \right\}^{2/3} +\Theta\left((q_n^{\circ})^{-q}, n |\mathbb B |\right)\right)\\ 
        &=\op(1).
    \end{align}
    Finally, (i) and (ii) follow from similar arguments of the proof of \cref{boot:net}. \hfill $\Box$
\section{Proof of auxiliary results}\label{sec:aux}
\subsection{Proofs of \texorpdfstring{\cref{prop:rhok} and a corollary}{Proposition C.1 and a corollary}}
In order to show \cref{prop:rhok}, we first prove the following lemma.
\begin{lemma}\label{lm:loclin}
Under Assumptions \ref{Ass:error}, \ref{A:K}, \ref{Ass:ck}, \ref{Ass:diff} and bandwidth conditions $nb^4 \to \infty$, $nb^7 \to 0$, $n^{2\phi - 1}b^{-1}h^2 \to 0$.  For any $q \geq 4$, we have the following results:
\begin{enumerate}[label=(\roman*)]
\item The difference-based estimator has the following approximation,
$$
\max_{(i,l,k) \in \B} \left\| \sup _{t \in \mathcal{T}}\left|\hat{\beta}^{i,l}_{k}(t)-\beta^{i,l}_{k}(t)-\frac{1}{nb_k^{i,l}}\sum_{j=1}^n K_{b^{i,l}_k}(t_j - t) \tilde e^{i,l}_{j,k} \right| \right\|_{q}=O\left(b^{-1/q}(n^{\phi-1}b^{-1}h+ b^{3}+ n^{-1/2}h)\right).
$$
\item  For the difference-based estimator with variance reduction, we have 
\begin{align}
    \max_{(i,l,k) \in \B} \left\| \sup _{t \in \mathcal{T}}\left|\check{\beta}^{i,l}_{k}(t)-\beta^{i,l}_{k}(t) -\frac{1}{nb_k^{i,l}}\sum_{j=1}^n \check{K}_{b^{i,l}_k}(t_j - t) \tilde e^{i,l}_{j, k} \right| \right\|_{q} =O\left(b^{-1/q}(n^{\phi-1}b^{-1}h+ b^{3}+ n^{-1/2}h)\right).
\end{align}
\end{enumerate}
\end{lemma}

\begin{proof}[Proof of of \cref{lm:loclin}]

Proof of (i). For simplicity, since $i, l, k$ will be fixed in the subsequent analysis, we omit them in the subscripts and  superscripts for short. That is we omit the dependence on $i$, $l$ and $k$ in $b^{i,l}_{k}$, $\beta^{i,l}_{k}$, when no confusion arises. Let $\tilde e_j$ denote $\tilde e^{i,l}_{j,k}$.
Define 
\begin{align}
    \mathbf{S}_{n}(t) = \begin{pmatrix}
        {S}_{n, 0}(t),& {S}_{n, 1}(t)\\
        {S}_{n, 1 }(t),& {S}_{n, 2}(t)
    \end{pmatrix},
\end{align}
where for $l=0, 1,2$,
\begin{equation}
    {S}_{n, l}(t)=\left(n b\right)^{-1} \sum_{j=1}^{n} \left\{\left(t_{j}-t\right) / b\right\}^{l} K_{b}\left(t_{j}-t\right).
  \end{equation}
 Further define 
\begin{align}
    &T_{n,s}(t) = \frac{1}{nb}\sum_{j=1}^n \tilde e_j \{(t_j-t)/b\}^s K_{b}(t_j-t), \quad \mf T_n(t) = (T_{n,0}(t),T_{n,1}(t))^{\T}, \label{def:T}\\ 
    &\delta_{n,s}(t) = \frac{1}{nb}\sum_{j=1}^n  \tilde \mu^{i}_{j,k}\tilde \mu^{l}_{j,h} \{(t_j-t)/b\}^s K_{b}(t_j-t), \quad \mf D_n(t) = (\delta_{n,0}(t),\delta_{n,1}(t))^{\T},\\ 
    &C_{n,s}(t) =  \frac{1}{nb} \sum_{j=1}^n(\tilde \epsilon^i_{j , k}\tilde \mu^{l}_{j,h} +  \tilde \epsilon^l_{j , h}\tilde \mu^{i}_{j,k}))\{(t_j-t)/b\}^s K_{b}(t_j-t),\quad \mf C_n(t)= (C_{n,0}(t),C_{n,1}(t))^{\T}.
\end{align}
Recall that,
$
   (nb)^{-1} \sum_{j=1}^n K_b (t_j - t) \tilde y^{i}_{j, k}\tilde y^{l}_{j, h} =  (nb)^{-1} \sum_{j=1}^n K_b (t_j - t) (\beta(t_j) +\tilde e_j+ \tilde \mu^{i}_{j,k}\tilde \mu^{l}_{j,h} + \tilde \epsilon^i_{j , k}\tilde \mu^{l}_{j,h} +  \tilde \epsilon^l_{j , h}\tilde \mu^{i}_{j,k})+ O(n^{-1}h).
$
Let $\hat{\bs \eta} (t) = \left(\hat{\beta}(t), b \hat{\beta}^{\prime}(t)\right)^{\T}$, $ \bs \eta(t) = \left(\beta(t), b \beta^{\prime}(t)\right)^{\T}$, where $\beta^{\prime}(t) = \partial \beta(t)/\partial t$.
Under condition \ref{A:gamma}, by Taylor's expansion,  if $|t_j - t| < b$,
$\beta(t_j) = \beta(t) + \beta'(t)(t_j - t) + (\beta^{\prime\prime}(t)/2 + O(b))(t_j - t)^2$. Therefore, 
\begin{align}
    \mf S_n(t) (\hat{\bs \eta} (t) - \bs \eta(t)) = \begin{pmatrix}
         b^2 S_{n,2}(t)(\beta^{\prime\prime}(t)+ O(b))/2 + O(n^{-1}h)  \\ 
         b^2  S_{n,3}(t)(\beta^{\prime\prime}(t)+ O(b))/2 + O(n^{-1}h)\\ 
    \end{pmatrix} +  \mf T_n(t) + \mf D_n(t) + \mf C_n(t).
    \label{eq:loc_mat}
\end{align}
We investigate $\mf T_n(t)$, $\mf D_n(t)$, $\mf C_n(t)$ in \eqref{eq:loc_mat} as follows. \par
(a) The order of $\sup_{t \in \mathscr{T}}\|\mf T_n(t)\|_q $.\par 
By Lemma B.3 of \cite{dette2018change}
\begin{align}
    \sup_{t \in \mathscr{T}}\| \mf T_n(t)\|_q = O((nb)^{-1/2}).\label{eq:T_bound}
\end{align}

(b) Calculation of $\mf D_n$. Recall that $t_i=i/n$.
Define the set
\begin{align}
    \mathbb{I}^i_k:=\left\{t_{j}:\left[t_{j-k}, t_{j}\right) \text { contains a change point of } \mu_i(\cdot) \right\}.
\end{align}
Let $\overline{\mathbb{I}}^i_k$ be the complement of $\mathbb{I}^i_k$ in $\left\{t_{1}, t_{2}, \cdots, t_{n}\right\}$. Here we omit the dependence  of $\mathbb{I}^i_k$ and $\overline{\mathbb{I}}^i_k$ on $n$ for the sake of brevity.
If there is no change point between $t_{j-k}$ and $t_{j}$, then $\tilde \mu^i_{j,k} = O(n^{-1}h)$, otherwise from condition \ref{A:jump}, we shall see $\tilde \mu^i_{j, k} = O(1)$. Let $\mathbb I =\mathbb I_k^i \cup \mathbb I_h^l$, and  $\overline{\mathbb{I}}$ be the complement of $\mathbb I$, where we omit the dependence on $i,l,k$ as long as no confusion is caused.
Then,
\begin{align}
    \sup_{t \in \mathscr{T}}| \delta_{n,s}(t)| &\leq \frac{1}{nb}  \sup_{t \in \mathscr{T}}\left|\sum_{j \in \mathbb I} \tilde \mu^i_{j,k}\tilde \mu^l_{j,h} \{(t_j-t)/b\}^s K_{b}(t_j-t)\right|\\ &+  \frac{1}{nb}  \sup_{t \in \mathscr{T}}\left|\sum_{j \in \bar{\mathbb I}} \tilde \mu^i_{j,k}\tilde \mu^l_{j,h} \{(t_j-t)/b\}^s K_{b}(t_j-t)\right|\\ 
    & = D_1 + D_2.
    \label{eq:delta}
\end{align}
Since there are at most $ O(n^{\phi} h)$ elements in $\mathbb{I}$, 
\begin{align}
    D_1 \leq \frac{1}{nb} \sum_{j \in \mathbb I}  \sup_{t \in \mathscr{T}}\left|\tilde \mu^i_{j,k}\tilde \mu^l_{j,h} \right| \sup_{u \in [-1,1]} |K(u)u^s|= O(n^{\phi-1}b^{-1}h).
    \label{eq:D_1}
\end{align}
Since $n^{2\phi -1}b^{-1}h^2 \to 0$, $D_1$ is of smaller order of $\mf T_n$. Note that under \cref{A:K}, for $t \in \mathcal T$,
\begin{align}
    (nb)^{-1}\sum_{i=1}^n\{(t_i-t)/b\}^{2l} K^2_{b}(t_i-t)
    & = \int_{-\infty}^{\infty} u^{2l}  K^2(u) du + O( (nb)^{-1}).
\end{align}
By Cauchy-Schwarz inequality, 
\begin{align}
    D_2 &\leq \frac{1}{nb}  \sup_{t \in \mathscr{T}}\left|\left(\sum_{j \in \overline{\mathbb I}} (\tilde \mu^i_{j,k}\tilde \mu^l_{j,h})^2 \right)^{1/2} \left(\sum_{j \in \overline{\mathbb I}}\{(t_j-t)/b\}^{2l} K^2_{b}(t_j-t)\right)^{1/2}\right|\\ 
    & \leq \frac{1}{nb}  \left(\sum_{j \in \overline{\mathbb I}}  (\tilde \mu^i_{j,k}\tilde \mu^l_{j,h})^2 \right)^{1/2} \sup_{t \in \mathscr{T}} \left(\sum_{i = 1}^n\{(t_i-t)/b\}^{2l} K^2_{b}(t_i-t)\right)^{1/2}\\
    &= O(n^{-2}b^{-1/2}h^2).
    \label{eq:D_2}
\end{align}
Therefore, combining \eqref{eq:delta}, \eqref{eq:D_1} and \eqref{eq:D_2}, we have
\begin{align}
    \sup_{t \in \mathscr{T}}| \delta_{n,s}(t)|  = O(n^{\phi-1}b^{-1}h).\label{eq:D_bound}
\end{align}

(c) Calculation of $\mf C_n$. 
 Let $c^{i,s}_{j,k}(t) = \tilde \mu^i_{j,k} \{(t_j-t)/b\}^s   K_{b}(t_j-t)$. If there is no change point between $t_{j-k}$ and $t_{j}$, $c^{i,s}_{j,k}(t)  = O(n^{-1}h)$, else from condition \ref{A:jump}, $c^{i,s}_{j,k}(t)  = O(1)$.
 
 \begin{align}
   \sup_{t \in \mathscr{T}}\| C_{n,s}(t)\|_q &=
    \frac{1}{nb}\sup_{t \in \mathscr{T}}\left \|\sum_{j\in \mathbb{I}} \tilde \epsilon^i_{j,k} c^{l,s}_{j,h}(t) \right\|_q +  \frac{1}{nb}\sup_{t \in \mathscr{T}}\left \|\sum_{j\in \mathbb{I}} \tilde \epsilon^l_{j,h} c^{i,s}_{j,k}(t) \right\|_q\\ 
    & + \frac{1}{nb}\sup_{t \in \mathscr{T}}\left \|\sum_{j \in \Ic} \tilde \epsilon^i_{j,k} c^{l,s}_{j,h}(t) \right\|_q +  \frac{1}{nb}\sup_{t \in \mathscr{T}}\left \|\sum_{j \in \Ic} \tilde \epsilon^l_{j,h} c^{i,s}_{j,k}(t) \right\|_q\\ 
    &:= C_{11} + C_{12} + C_{21} + C_{22}.
    \label{eq:C}
\end{align}
Similar to \eqref{eq:D_1}, under condition \ref{A:exp}, we have
\begin{align}
    C_{11}  = O(n^{\phi-1}b^{-1}h), \quad C_{12}  = O(n^{\phi-1}b^{-1}h).\label{eq:C_1}
\end{align}
Write $\Ic = \{j_1, j_2, \cdots, j_{K}\}$, where $j_s$ denotes $\tilde \mu_i(\cdot) $ does not contain a change point over  time interval $\left[t_{j_s-k}, t_{j_s}\right)$, and $\tilde \mu_l(\cdot) $ does not contain a change point over  time interval $\left[t_{j_s-h}, t_{j_s}\right)$,  $1 \leq s \leq K$. The cardinality of $\Ic$ is $O(n)$. Let $c^{l,s}_{j_0, h}(t) \equiv 0$.
Then, we have
\begin{align}
   & \sup_{t \in \mathcal{T}}\sum_{m=1}^{K} \left | c^{l,s}_{j_m, h}(t) - c^{l,s}_{j_{m-1}, h}(t)\right|\\
    &\leq \sup_{t \in \mathcal{T}}\sum_{m=1}^{K} \left |\tilde \mu^l_{j_m, h}  \{(t_{j_m}-t)/b\}^s K_{b}(t_{j_m}-t) \right|+\sup_{t \in \mathcal{T}}\sum_{m = 1}^K \left | \tilde \mu^l_{j_{m-1}, h}\{(t_{j_{m-1}}-t)/b\}^s  K_{b}(t_{j_{m-1}}-t)\right|\\
    &= O(hb).
\end{align}
Using summation-by-parts formula, we have
\begin{align}
    C_{21} &= \frac{1}{nb}\left \|\sum_{m = 1}^K \tilde \epsilon^i_{j_m,k} \right\|_q\sup_{t \in \mathcal{T}} \left| c^{l,s}_{j_K,h}(t)\right|+ \frac{1}{nb}\sup_{1 \leq r \leq K} \left \|\sum_{m=1}^r \tilde \epsilon^i_{j_m,k}\right\|_q \sup_{t \in \mathcal{T}}\sum_{m=1}^{K} \left | c^{l,s}_{j_m,h}(t) - c^{l,s}_{j_{m-1},h}(t)\right|\\ 
   & =O(n^{-1/2} h).\label{eq:C_3}
\end{align}
Similarly, we have $C_{22}  =O(n^{-1/2} h) $.
Combining \eqref{eq:C}, \eqref{eq:C_1}, and \eqref{eq:C_3}, we have
\begin{align}
    \sup_{t \in \mathcal{T}}\| C_{n,s}(t)\|_q = O(n^{\phi-1}b^{-1}h + n^{-1/2}h ).\label{eq:C_bound}
\end{align}
From calculus, 
$
  \sup_{t \in \mathcal{T}}|S_{n,l}(t) - \mu_l|=
  O((nb)^{-1})
$, 
where $\mu_l = \int_{\mathbb{R}} x^l K(x) dx$.
Combining \eqref{eq:loc_mat}, \eqref{eq:T_bound}, \eqref{eq:D_bound} and \eqref{eq:C_bound}, and by the invertibility of $\mf S_n(t)$, under the bandwidth condition $n^{2\phi - 1}b^{-1}h^ 2 \to 0$, we have
\begin{align}
    \sup_{t \in \mathcal{T}} \|\hat{\bs \eta} (t) - \bs \eta(t)\|_q = O((nb)^{-1/2}+ b^2).
\end{align}
 Under \cref{A:K}, we have $\mu_0 = 1$, $\mu_1 = 0$, $\mu_2 = 0$.
Then, it follows that 
\begin{align}
    \mf S_n(t) (\hat{\bs \eta} (t) -  \bs \eta(t)) = \begin{pmatrix}
         b^2 \mu_2 \beta^{\prime\prime}(t)/2+ O(b^3 + b/n + n^{-1}h)  \\ 
         O(b^3 + b/n + n^{-1}h)\\ 
    \end{pmatrix} +  \mf T_n(t) + \mf D_n(t) + \mf C_n(t),
    \label{eq:loc_mat1}
\end{align}
Under bandwidth conditions $nb^3 \to \infty$, $n^{2\phi-1}b^{-1}h^2 \to 0$, following the proof of Theorem 1 in \cite{zhou2010simultaneous}, 
\begin{align}
   \sup _{t \in \mathcal{T}}\left\|\hat{\beta}(t)-\beta(t)  -\frac{1}{nb}\sum_{j=1}^n K_{b}(t_j - t) \tilde e_{j} \right\|_q &= O ((nb)^{-3/2}+ b^3 + n^{-1}b+ n^{\phi-1}b^{-1}h + n^{-1/2}h )\\ & =O (n^{\phi-1}b^{-1}h + b^3 + n^{-1/2}h).\label{eq:betamoment}
\end{align}

Using Proposition B.1. of \cite{dette2018change}, we have for fixed $(i,l,k) \in \B$,
\begin{align}
   \left\| \sup _{t \in \mathcal{T}} \left|\hat{\beta}(t)-\beta(t)  -\frac{1}{nb}\sum_{j=1}^n K_{b}(t_j - t) \tilde e_{j} \right| \right\|_q = O(b^{-1/q}(n^{\phi-1}b^{-1}h+ b^{3}+ n^{-1/2}h)).
\end{align}

Under \cref{Ass:error} and the uniformly bounded Lipschitz constants of $\mu_i(t)$, the constant in the big $O$ for $(i,l,k)\in \B$ can be also uniformly bounded, i.e., 
\begin{align}
  \max_{(i,l,k) \in \B}   \left\| \sup _{t \in \mathcal{T}} \left|\hat{\beta}(t)-\beta(t)  -\frac{1}{nb}\sum_{j=1}^n K_{b}(t_j - t) \tilde e_{j} \right| \right\|_q = O(b^{-1/q}(n^{\phi-1}b^{-1}h+ b^{3}+ n^{-1/2}h)).
\end{align}

Proof of (ii).
Since $\check{\beta}_{+}(t)=\sum_{j=0, 1, 2}A_{j}(r)\hat{\beta}_j (t-(r + 1 - j)\omega(t))$, we have  by (i)  
\begin{align}
   & \max_{(i,l,k) \in \B}  \left\| \sup _{t \in \mathcal{T}} \left|\check{\beta}_{+}(t) - \sum_{j=0, 1, 2}A_{j}(r)\beta_{j}(t-(r + 1 - j)\omega(t)) \right. \right.\\ &\left.\left.- \frac{1}{nb}\sum_{i=1}^n \sum_{j=0, 1, 2}A_{j}(r)K_{b}(t_i - t + (r + 1 - j)\omega(t)) \tilde e_{i} \right|\right\|_q  = O(b^{-1/q}(n^{\phi-1}b^{-1}h+ b^{3}+ n^{-1/2}h)).\label{eq:plus1}
\end{align}
Under condition \ref{A:gamma} of \cref{Ass:error},  by (7.4) of \cite{cheng2007reducing}, we have
\begin{align}
   \left| \sum_{j=0, 1, 2}A_{j}(r)\beta_j(t-(r + 1 - j)\omega(t)) - \beta (t)\right|   = O( b^3).\label{eq:plus2}
\end{align}
Combining \eqref{eq:plus1} and \eqref{eq:plus2}, we have
\begin{align}
    & \max_{(i,l,k) \in \B}  \left\| \sup _{t \in \mathcal{T}} \left|\check{\beta}_{+}(t) - \beta(t)  - \frac{1}{nb}\sum_{i=1}^n \sum_{j=0, 1, 2}A_{j}(r)K_{b}(t_i - t + (r + 1 - j)\omega(t)) \tilde e_{i} \right| \right\|_q\\&  = O(b^{-1/q}(n^{\phi-1}b^{-1}h+ b^{3}+ n^{-1/2}h)).
 \end{align}
 Similarly, we have 
 \begin{align}
&  \max_{(i,l,k) \in \B}  \left\|\sup _{t \in \mathcal{T}} \left|\check{\beta}_{-}(t) - \beta(t) - \frac{1}{nb}\sum_{i=1}^n \sum_{j=0, 1, 2}A_{j}(-r)K_{b}(t_i - t + (-r + 1 - j)\omega(t)) \tilde e_{i} \right| \right\|_q \\ & = O(b^{-1/q}(n^{\phi-1}b^{-1}h+ b^{3}+ n^{-1/2}h)).
 \end{align}
 Therefore, (ii) holds. 
 \end{proof}
\subsubsection{Proof of \texorpdfstring{\cref{prop:rhok}}{Proposition C.1}}
  
Proof of (i).   Recall that $\mu_l = \int_{\mathbb{R}} x^l K(x) dx$. Recall the definitions of $\gamma_k^{i,l}(t) $ and $ \sigma_{i,l}(t)$ in \cref{def:cor},  $\tilde \gamma_k^{i,l}(t) $  and  $\tilde  \sigma_{i,l}(t)$ in \eqref{eq:rho_network}.
By Lemma 5 of \cite{zhou2010simultaneous}, under \cref{Ass:error}, we have for $t \in (0,1)$,
\begin{align}
   \max_{(i,l,k) \in \B} \sup_{t\in \TT}| \beta_k^{i,l}(t) - \gamma_0^{i,l}(t)| =  O(\chi^h + h/n),\quad \max_{(i,l,k) \in \B} \sup_{t\in \TT} | \beta_h^{i,l}(t) - 2\gamma_0^{i,l}(t)| =  O(\chi^h + h/n).\label{eq:betagamma}
\end{align}
Note that  $\chi^h + h/n = o(n^{\phi-1}b^{-1}h) = o((nb)^{-1/2})$.
We start by studying the bound and representation of $\tilde \sigma_{i,l}(t) - \sigma_{i,l}(t)$. \par 
(a) By the definitions of $ \sigma_{i,l}(t)$ and $\tilde  \sigma_{i,l}(t)$,  
\begin{align}
    \| \tilde \sigma_{i,l}(t) - \sigma_{i,l}(t)\|_{2q} &= \left\| \sqrt{\tilde \gamma^i_0(t) \tilde \gamma^l_0(t)} - \sqrt{ \gamma^i_0(t) \gamma^l_0(t)} \right\|_{2q} \\
    & = \left\| \frac{ \tilde \gamma^i_0(t) \tilde \gamma^l_0(t) - \gamma^i_0(t) \gamma^l_0(t)}{ \sqrt{\tilde \gamma^i_0(t) \tilde \gamma^l_0(t)} + \sqrt{ \gamma^i_0(t) \gamma^l_0(t)}}\right\|_{2q}\\ 
     & = \left\| \frac{ \tilde \gamma^i_0(t) (\tilde \gamma^l_0(t) - \gamma^l_0(t)) + (\tilde \gamma^i_0(t)-\gamma^i_0(t)) \gamma^l_0(t)}{ \sqrt{\tilde \gamma^i_0(t) \tilde \gamma^l_0(t)} + \sqrt{ \gamma^i_0(t) \gamma^l_0(t)}}\right\|_{2q}\\ 
     & \leq \left\| \frac{ \tilde \gamma^i_0(t)}{\tilde \sigma_{i,l}(t) + \sigma_{i,l}(t)}\right\|_{4q} \| \tilde \gamma^l_0(t) - \gamma^l_0(t))\|_{4q} + \left\| \frac{ \gamma^l_0(t)}{\tilde \sigma_{i,l}(t) + \sigma_{i,l}(t)}\right\|_{4q} \| \tilde \gamma^i_0(t) - \gamma^i_0(t))\|_{4q}\\
     &= O((nb)^{-1/2}), \label{eq:sigmatilde}
\end{align}
where the last inequality follows from triangle inequality, and the last equality follows from \ref{A:exp}, \eqref{eq:T_bound} and \eqref{eq:betamoment}.\par
(b) The representation of $\tilde \sigma_{i,l}(t) - \sigma_{i,l}(t)$.\par
Let $T^{i,l}_{k}(t) = \frac{1}{nb_k^{i,l}} \sum_{j=1}^n \tilde e^{i,l}_{j,k} K_{b^{i,l}_k}(t_j - t).$ When $i=l$, we use a single index for the sake of simplicity. For example, we use $T^{i}_{h}(t)$ to represent $T^{i,i}_{h}(t)$.
Observe that by  \eqref{eq:betamoment}, \eqref{eq:T_bound} and condition \ref{A:exp}, 
\begin{align}
  & \|  \tilde \gamma^i_0(t) \tilde \gamma^l_0(t) - \gamma^i_0(t) \gamma^l_0(t) -  \gamma^i_0(t)T^{l}_{h}(t)/2+  \gamma^l_0(t)T^{i}_{h}(t)/2\|_q  \\ 
  & = \|(\tilde \gamma^i_0(t) -  \gamma^i_0(t)) (\tilde \gamma^l_0(t) - \gamma^l_0(t)) + (\tilde \gamma^i_0(t)-\gamma^i_0(t)-T^{i}_{h}(t)/2) \gamma^l_0(t) -   \gamma^i_0(t) (\tilde \gamma^l_0(t) - \gamma^l_0(t)-T^{l}_{h}(t)/2)  \|_q\\
  & = O(n^{\phi-1}b^{-1}h + b^{3} + n^{-1/2}h).\label{eq:gammapro}
\end{align}

Then we have, 
\begin{align}
   & \left\| \tilde \sigma_{i,l}(t) - \sigma_{i,l}(t) -  \left(\gamma^i_0(t)T^{l}_{h}(t)+  \gamma^l_0(t)T^{i}_{h}(t)\right)/(4\sigma_{i,l}(t) )\right\|_q  \\ 
 & = \left\| \frac{  \tilde \gamma^i_0(t) \tilde \gamma^l_0(t) - \gamma^i_0(t) \gamma^l_0(t)}{ 2\sigma_{i,l}(t) } \frac{2\sigma_{i,l}(t) }{\tilde \sigma_{i,l}(t) +   \sigma_{i,l}(t) } - \left(\gamma^i_0(t)T^{l}_{h}(t)/2+  \gamma^l_0(t)T^{i}_{h}(t)/2\right)/(2\sigma_{i,l}(t) )\right\|_q\\ 
& \leq \left\| \frac{  \tilde \gamma^i_0(t) \tilde \gamma^l_0(t) - \gamma^i_0(t) \gamma^l_0(t)-\gamma^i_0(t)T^{l}_{h}(t)/2-  \gamma^l_0(t)T^{i}_{h}(t)/2}{ \sigma_{i,l}(t) } \frac{\sigma_{i,l}(t) }{\tilde \sigma_{i,l}(t) +   \sigma_{i,l}(t) }  \right\|_q\\ 
&+ \left\| \frac{\sigma_{i,l}(t) - \tilde  \sigma_{i,l}(t)}{\tilde \sigma_{i,l}(t) +   \sigma_{i,l}(t) }  \left(\gamma^i_0(t)T^{l}_{h}(t)/2+  \gamma^l_0(t)T^{i}_{h}(t)/2\right)/(2\sigma_{i,l}(t) )\right\|_q\\ 
& \leq   \left\| \frac{  \tilde \gamma^i_0(t) \tilde \gamma^l_0(t) - \gamma^i_0(t) \gamma^l_0(t)-\gamma^i_0(t)T^{l}_{h}(t)/2-  \gamma^l_0(t)T^{i}_{h}(t)/2}{ \sigma_{i,l}(t) } \right\|_q\\ &+ \left\| \frac{\sigma_{i,l}(t) - \tilde  \sigma_{i,l}(t)}{\tilde \sigma_{i,l}(t) +   \sigma_{i,l}(t) } \right\|_{2q}  \left\| \left(\gamma^i_0(t)T^{l}_{h}(t)+  \gamma^l_0(t)T^{i}_{h}(t)\right)/(4\sigma_{i,l}(t) )\right\|_{2q}\\
& = O(n^{\phi-1}b^{-1}h + b^{3} + n^{-1/2}h),\label{eq:sigmaexp}
\end{align}
where the second inequality follows from the non-negativity of $\tilde \sigma_{i,l}(t)$, and the last equality follows from \eqref{eq:T_bound},  \eqref{eq:sigmatilde} and \eqref{eq:gammapro}.

Recall  $\vartheta^{i,l}_k(t) := \frac{1}{nb_{k}^{i,l}}\sum_{j=1}^n K_{b_{k}^{i,l}}(t_i-t)\Xi^{i,l}_{j,k}$, where $\Xi^{i,l}_{j,k}$ is as defined in \eqref{eq:Xidiscrete}. By \cref{lm:loclin} \eqref{eq:betagamma}, and \eqref{eq:sigmaexp}, we have uniformly for all $t \in \mathcal{T}$, 
\begin{align} 
     \left\|\tilde \rho^{i,l}_k(t)  - \rho^{i,l}_k(t) - \vartheta^{i,l}_k(t)\right\|_q
     & =  \left\| \frac{\left[\tilde{\gamma}^{i,l}_{k}(t)-\gamma^{i,l}_{k}(t)\right]-\left[\tilde{\sigma}_{i,l}(t)-\sigma_{i,l}(t)\right] \rho^{i,l}_{k}(t)}{\tilde{\sigma}_{i,l}(t)} - \vartheta^{i,l}_k(t)\right\|_q \\
    & = \left\|\frac{\sigma_{i,l}(t)}{\tilde{\sigma}_{i,l}(t)}\vartheta^{i,l}_k(t) - \vartheta^{i,l}_k(t)\right\|_q + O(n^{\phi-1}b^{-1}h + b^{3} + n^{-1/2}h).
    \label{eq:rho_gamma}
    \end{align}
By \eqref{eq:sigmatilde}, it follows that
\begin{align}
     &\sup_{t \in \mathcal{T}} \left\|\left(\frac{\sigma_{i,l}(t)}{\tilde{\sigma}_{i,l}(t)}\vartheta^{i,l}_k(t) - \vartheta^{i,l}_k(t)\right) \mf 1(\bar B^{\prime}_n) \right\|_q \\ &=  \sup_{t \in \mathcal{T}} \left\|\left(\frac{\tilde  \sigma_{i,l}(t)- \sigma_{i,l}(t)}{\tilde \sigma_{i,l}(t)}\right) \mf 1(\bar B^{\prime}_n) \right\|_{2q} \left\|\vartheta^{i,l}_k(t)\right\|_{2q}= O((nb)^{-1}).\label{cn1}
\end{align}
Therefore, by \eqref{eq:rho_gamma}, \eqref{cn1}, we obtain
\begin{align}
   & \sup_{t \in \mathcal{T}} \left\|(\tilde \rho^{i,l}_k(t)  - \rho^{i,l}_k(t) - \vartheta^{i,l}_k(t))\mf 1(\bar B^{\prime}_n)\right\|_q
   =O(n^{\phi-1}b^{-1}h + b^{3} + n^{-1/2}h).
\end{align}
The result follows from Proposition B.1 in \cite{dette2018change} and a close investigation of the constants in the big $O$'s.

Proof of (ii). The result follows from \cref{lm:loclin} (ii).\hfill $\Box$

\subsubsection{A corollary of \texorpdfstring{\cref{lm:loclin}}{Lemma F.1}}
\begin{corollary}\label{cor:suprhok}
Under the conditions of \cref{prop:rhok}, for a sufficiently large $q$, we have 
\begin{align}
    &\max_{(i,l,k) \in \B} \sup _{t \in \mathcal{T}}\left| \tilde{\beta}^{i,l}_{k}(t)-\beta^{i,l}_{k}(t)- \frac{1}{nb_k^{i,l}}\sum_{j=1}^n K_{b^{i,l}_k}(t_j - t) \tilde e^{i,l}_{j,k}\right|= \Op\left(c_n\right) = \op(1),
\end{align}
where $c_n$ is as defined in \cref{nonasynetwork}.
\end{corollary}
\begin{proof}
The proof follows from \cref{lm:loclin} and the fact that 
\begin{align}
    &\stackrel[(i,l,k) \in \B]{}{\sup}\underset{t \in \TT}{\max}\left| \tilde{\beta}^{i,l}_{k}(t)-\beta^{i,l}_{k}(t)- \frac{1}{nb_k^{i,l}}\sum_{j=1}^n K_{b^{i,l}_k}(t_j - t) \tilde e^{i,l}_{j,k} \right|^q \\ & \leq \underset{(i,l,k) \in \B}{\sum} \underset{t \in \TT}{\sup} \left|  \tilde{\beta}^{i,l}_{k}(t)-\beta^{i,l}_{k}(t)- \frac{1}{nb_k^{i,l}}\sum_{j=1}^n K_{b^{i,l}_k}(t_j - t) \tilde e^{i,l}_{j,k} \right|^q.
\end{align}
\end{proof}

\subsection{Proof of \texorpdfstring{\cref{prop:Gammacp}}{Proposition C.2}}
\begin{proof}
Notice that 
$ \stackrel[(i,l,k) \in \B]{}{\max}\underset{t \in \TT}{\sup}\left|\hat{\tilde \Gamma}^{i,l,2}_{k}(t) - \tilde \Gamma^{i,l,2}_{k}(t) \right|^q \leq \underset{(i,l,k) \in \B}{\sum} \underset{t \in \TT}{\sup} \left|\hat{\tilde \Gamma}^{i,l,2}_{k}(t) - \tilde \Gamma^{i,l,2}_{k}(t) \right|^q$, hence
\begin{align}
\left\|  \max_{(i,l,k) \in \B} \sup_{t \in \TT} \left|\hat{\tilde \Gamma}^{i,l,2}_{k}(t) - \tilde \Gamma^{i,l,2}_{k}(t) \right|  \right\|_q \leq |\B|^{1/q} \max_{(i,l,k) \in \B}  \left\| \sup_{t \in \TT} \left|\hat{\tilde \Gamma}^{i,l,2}_{k}(t) - \tilde \Gamma^{i,l,2}_{k}(t) \right|  \right\|_q.
\end{align}
By \cref{prop:rhok}, \eqref{eq:T_bound} and Proposition B.1 of \cite{dette2018change}, under \cref{Ass:error}, we have uniformly for $(i,l,k)\in \B$, 
$$\left\| \sup _{t \in \mathcal{T}}\left|\hat{\beta}^{i,l}_{k}(t)-\beta^{i,l}_{k}(t)\right| \right\|_{2q}=O\left(n^{-\frac{1}{2}}b^{-\frac{1}{2}-\frac{1}{2q}}\right), \quad \left\|
\max_{1 \leq s\leq n-m+1}\left |\sum_{j=s}^{s+m-1} \Xi_{j, k}^{i, l}\right | \right\|_{2q}=O\left(m^{\frac{1}{2}}(m/n)^{-\frac{1}{2q}}\right).$$
    The proof follows from a careful investigation of Theorem 5 of \cite{zhou2010simultaneous} and Proposition 1.3 in \cite{dette2019detecting}.
    
\end{proof}

\normalem
\bibliographystyle{apalike}
\bibliography{main_arxiv}
\end{document}